\documentclass[11pt,a4paper]{article}

\usepackage[margin = 1in]{geometry}
\usepackage{charter}
\usepackage{amsmath}
\allowdisplaybreaks
\usepackage{multirow}
\usepackage{float}
\usepackage{bigints}
\usepackage{extarrows}
\usepackage{bm}
\usepackage{array}
\usepackage[dvipsnames]{xcolor}
\usepackage{mathrsfs}
\usepackage{amsfonts}
\usepackage{amssymb}
\usepackage{subfigure}
\usepackage{graphicx}
\usepackage{tikz}
\usetikzlibrary{arrows}
\usepackage{multicol}
\usepackage{bbm}
\usepackage[linktoc=page]{hyperref}
\hypersetup{
    colorlinks,
    linkcolor={red},
    citecolor={blue},
    urlcolor={lightgray}
}
\usepackage[capitalise]{cleveref}
\usepackage{natbib}
\usepackage[inline]{enumitem}

\title{Tight Revenue Gaps among Simple Mechanisms\footnote{A preliminary version of this paper appeared in SODA 2019.}}
\author{
Yaonan Jin\thanks{Columbia University; work done in part as a student at IEDA, HKUST. {\tt jin.yaonan@columbia.edu}.}
\and Pinyan Lu\thanks{ITCS, Shanghai University of Finance and Economics. {\tt
lu.pinyan@mail.shufe.edu.cn}.}
\and
Zhihao Gavin Tang\thanks{ITCS, Shanghai University of Finance and Economics. {\tt tang.zhihao@mail.shufe.edu.cn}.}
\and
Tao Xiao\thanks{Department of Computer Science, Shanghai Jiao Tong University. {\tt xt\_1992@sjtu.edu.cn}.}
}
\date{}

\newcolumntype{M}[1]{>{\centering\arraybackslash}m{#1}}
\newcolumntype{N}{@{}m{0pt}@{}}

\newcommand{\Xcomment}[1]{{}}

\def\E{\mathbf E}
\def\Prob{\mathbf{Pr}}

\newcommand{\dd}{\mathrm{d}}    % for integrals
\newcommand{\lhs}{\mathrm{LHS}} % for inequalities
\newcommand{\rhs}{\mathrm{RHS}} % for inequalities
\newcommand{\given}{\;\mid\;}

\newcommand{\Int}[2]{\bigintsss_{\;\;#1}^{#2}}

\newcommand{\argmin}{\mathop{\rm argmin}}
\newcommand{\argmax}{\mathop{\rm argmax}}
\newcommand{\eqdef}{\stackrel{\textrm{\tiny def}}{=}}
\newcommand{\eps}{\varepsilon}

\newcommand{\RR}{\mathbb{R}}
\newcommand{\RRP}{\RR_{\geq 0}}
\newcommand{\NN}{\mathbb{N}}
\newcommand{\NNP}{\NN_{\geq 1}}

\newcommand{\opt}{{\sf OPT}}
\newcommand{\spm}{{\sf SPM}}
\newcommand{\opm}{{\sf OPM}}
\newcommand{\ar}{{\sf AR}}
\newcommand{\ap}{{\sf AP}}

\newcommand{\reg}{\textsc{Reg}}
\newcommand{\tri}{\textsc{Tri}}

\newcommand{\C}{{\cal C^*}}
\newcommand{\R}{{\cal R}}
\newcommand{\Q}{{\cal Q}}
\newcommand{\V}{{\cal V}}

\newcommand{\FF}{\overline{F}}
\newcommand{\rr}{\overline{r}}
\newcommand{\vv}{\overline{v}}
\newcommand{\qq}{\overline{q}}

\newcommand{\first}{D_1}
\newcommand{\second}{D_2}

\newcommand{\bid}{b}
\newcommand{\bids}{{\mathbf \bid}}

\newcommand{\distr}{F}
\newcommand{\distrs}{{\mathbf \distr}}
\newcommand{\distrmi}{{\mathbf \distrmi}_{-i}}
\newcommand{\distri}[1][i]{{\distri_{#1}}}

\newcommand{\price}{p}
\newcommand{\prices}{{\mathbf \price}}

\usepackage{amsthm}
\def\namedlabel#1#2{\begingroup
    #2%
    \def\@currentlabel{#2}%
    \phantomsection\label{#1}\endgroup
}

\theoremstyle{plain}
\newtheorem{theorem}{Theorem}
\newtheorem{lemma}{Lemma}

\newtheorem{fact}{Fact}

\theoremstyle{definition}

\newtheorem{example}{Example}
\newtheorem{remark}{Remark}

\crefname{theorem}{Theorem}{Theorems}
\crefname{lemma}{Lemma}{Lemmas}
\crefname{corollary}{Corollary}{Corollaries}
\crefname{fact}{Fact}{Facts}
\crefname{definition}{Definition}{Definitions}
\crefname{example}{Example}{Examples}
\crefname{remark}{Remark}{Remarks}

\begin{document}
\maketitle

We consider a fundamental problem in microeconomics: selling a single item to a number of potential buyers, whose values are drawn from known independent and regular (not necessarily identical) distributions. There are four widely-used and widely-studied mechanisms in the literature: {\sf Myerson Auction}~({\sf OPT}), {\sf Sequential Posted-Pricing}~({\sf SPM}), {\sf Second-Price Auction with Anonymous Reserve}~({\sf AR}), and {\sf Anonymous Pricing}~({\sf AP}).

{\sf OPT} is revenue-optimal but complicated, which also experiences several issues in practice such as fairness; {\sf AP} is the simplest mechanism, but also generates the lowest revenue among these four mechanisms; {\sf SPM} and {\sf AR} are of intermediate complexity and revenue. We explore revenue gaps among these mechanisms, each of which is defined as the largest ratio between revenues from a pair of mechanisms. We establish two tight bounds and one improved bound:
\begin{enumerate}
  \item {\sf SPM} vs.\ {\sf AP}: this ratio studies the power of discrimination in pricing schemes. We obtain the tight ratio of $\C \approx 2.62$, closing the gap between $\big[\frac{e}{e - 1}, e\big]$ left before.
  \item {\sf AR} vs.\ {\sf AP}: this ratio measures the relative power of auction scheme vs.\ pricing scheme, when no discrimination is allowed. We attain the tight ratio of $\frac{\pi^2}{6} \approx 1.64$, closing the previously known bounds $\big[\frac{e}{e - 1}, e\big]$.
  \item {\sf OPT} vs.\ {\sf AR}: this ratio quantifies the power of discrimination in auction schemes, and is previously known to be somewhere between $\big[2, e\big]$. The lower-bound of $2$ was conjectured to be tight by Hartline and Roughgarden (2009) and Alaei et al.\ (2015). We acquire a better lower-bound of $2.15$, and thus disprove this conjecture.
\end{enumerate}

\setcounter{page}{0}
\thispagestyle{empty}
\newpage

\tableofcontents

\newpage

\section{Introduction}
\label{sec:intro}
How to maximize the expected revenue of a seller, who wants to sell an indivisible item to a number of buyers, is a central problem in microeconomics. The simplest mechanism is {\sf Anonymous Pricing} (denoted by $\ap$). Such a mechanism simply posts a price of $p \in \RRP$ to all buyers, and the item is sold out iff at least one buyer has a value no less than this price. If the seller knows the value distributions of the buyers, he can leverage a proper price to maximize the revenue (among this family of mechanisms). Although widely-used, this is not the revenue-optimal selling method; the optimal mechanism is the prominent {\sf Myerson Auction} \citep[denoted by $\opt$; see][]{M81}. In comparison, $\opt$ is far more complex than $\ap$, due to two reasons:
\begin{enumerate}[label = (\alph*), font = {\bfseries}]
\item It discriminates different buyers with different value distributions. Conceivably, this may incur some fairness issues, and is not feasible in some markets.
\item It is an {\em auction scheme} instead of a {\em pricing scheme}, and thus requires more seller-to-buyer communication. This may also raise privacy concerns to the buyers, since they need to report their private values, rather than make take-it-or-leave-it decisions.
\end{enumerate}
These complications and other undesirable issues hinder the prevalence of {\sf Myerson Auction}. To address these issues, two mechanisms with intermediate complexities (compared to $\opt$ and $\ap$) are widely studied in the literature, and are widely adopted in practice:
\begin{enumerate*}[label = (\alph*), font = {\bfseries}]
\item to avoid price discrimination, the seller can use {\sf Second-Price Auction with Anonymous Reserve} \citep[denoted by $\ar$; see][]{HR09}; and
\item to reduce communication, the seller can employ {\sf Sequential Posted-Pricing} \citep[denoted by $\spm$; see][]{CHMS10,CMS15}.
\end{enumerate*}
We defer the formal definitions of all mechanisms to \Cref{subsec:prelim:mech}.

These four mechanisms together form the lattice structure in \Cref{fig:result}, in terms of both revenue-domination and complexity. It is well known that there is a revenue gap between any pair of mechanisms. The reader may query that how large these gaps can be.

Indeed, quantitative analysis of these gaps is also a striking theme in algorithmic economics. To this end, the notion of approximation ratio (originated from the TCS community) turns out to be a powerful language. There is a rich literature on studying revenue gaps/approximation ratios among various mechanisms \citep[e.g.\ see][]{BK94, GHW01, BHW02, GHKKKM05, KP13, CGL14, CGL15, FILS15, DFK16, AHNPY15, CFHOV17}.

\begin{figure}[!htbp]
\centering
\begin{tikzpicture}[scale = 2.25]
\node(n1) at (2.5, 0.8) [draw, thick] {$\opt$: discriminate auction};
\node(n2) at (0, 0) [draw, thick] {$\spm$: discriminate pricing};
\node(n3) at (5, 0) [draw, thick] {$\ar$: anonymous auction};
\node(n4) at (2.5, -0.8) [draw, thick] {$\ap$: anonymous pricing};
\draw[thick, ->, >=triangle 45] (n1.west) to node[anchor = -15] {$\big[1.34, 1.49\big]$} (n2.north);
\draw[thick, ->, >=triangle 45] (n1.east) to node[anchor = -165] {$[2, e]$} (n3.north);
\draw[thick, ->, >=triangle 45] (n1.east) to node[anchor = 15] {$[{\bf 2.15}, e]$} (n3.north);
\draw[thick, ->, >=triangle 45] (n1.south) to node[left] {$[2.23, e]$} (n4.north);
\draw[thick, ->, >=triangle 45] (n1.south) to node[right] {$[{\bf 2.62}, e]$} (n4.north);
\draw[thick, ->, >=triangle 45] (n2.south) to node[anchor = -165] {${\bf 2.62}$} (n4.west);
\draw[thick, ->, >=triangle 45] (n2.south) to node[anchor = 15] {$\big[\frac{e}{e - 1}, e\big]$} (n4.west);
\draw[thick, ->, >=triangle 45] (n3.south) to node[anchor = -15] {${\bf \pi^2/6}$} (n4.east);
\draw[thick, ->, >=triangle 45] (n3.south) to node[anchor = 165] {$\big[\frac{e}{e - 1}, e\big]$} (n4.east);
\end{tikzpicture}
\caption{Revenue gaps among basic mechanisms in the asymmetric regular setting. Our new results are marked in {\bf bold}: the $\C \approx 2.62$ bound is by solving the $\spm$ vs.\ $\ap$ problem (see \Cref{sec:spm_ap}); the $\pi^2 / 6 \approx 1.64$ bound is by solving the $\ar$ vs.\ $\ap$ problem (see \Cref{sec:ar_ap}); and the $2.15$ bound is by constructing a better lower-bound instance for the $\opt$ vs.\ $\ar$ problem (see \Cref{sec:opt_ar}). For the other previously known results, we give a survey in \Cref{app:summary}.}
\label{fig:result}
\end{figure}
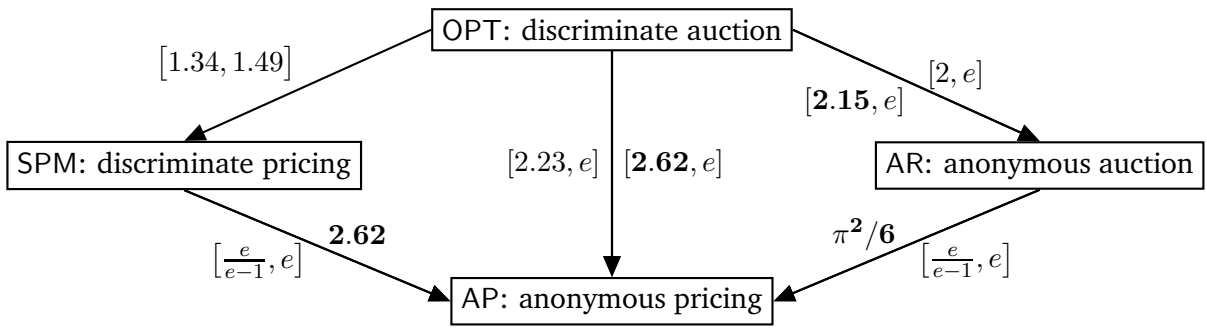

\subsection{Our Results}
In the natural setting with asymmetric\footnote{Throughout the paper, asymmetric distributions refer to the setting when different buyers can have distinct value distributions, as opposed to identical distributions.} and regular distributions, no tight revenue gap between any pair of the four mechanisms was previously known. In this work, we get two tight bounds and an improved bound, where the improved lower bound between $\opt$ and $\ar$ disproves a conjecture asked by \citet{HR09,H13,AHNPY15}.

\vspace{.1in}
\noindent
{\bf SPM vs.\ AP.}
This comparison measures the power of discrimination in pricing schemes. We establish the tight ratio of constant $\C \approx 2.62$. Prior to this work, the tight ratios in the other three settings were known to be
\begin{enumerate*}[label = (\alph*), font = {\bfseries}]
\item $n$ in the asymmetric general setting \citep[see][]{AHNPY15};
\item $\frac{e}{e - 1} \approx 1.58$ in the i.i.d.\ regular setting; and
\item $2$ in the i.i.d.\ general setting \citep[e.g.\ see][]{H13,DFK16}.\footnote{The referenced book and paper do not state their results explicitly in the language of $\spm$ vs.\ $\ap$. However, these $\spm$ vs.\ $\ap$ results are easy corollaries of the referenced results.}
\end{enumerate*}
Actually, we can also get the last two ratios by combining the results in \citet{M81,KS78,HK82}, which was first observed by \cite{HKS07}.

\vspace{.1in}
\noindent
{\bf AR vs.\ AP.}
This comparison studies the relative power between auction schemes and pricing schemes, when no discrimination is allowed. We first
\begin{enumerate*}[label = (\alph*), font = {\bfseries}]
\item prove an upper bound of $\frac{\pi^2}{6} \approx 1.64$ in the asymmetric general setting, and then
\item respectively construct matching lower-bound instances in the asymmetric regular setting and the i.i.d.\ general setting. Prior to this work,
\item in the i.i.d.\ regular setting, where $\ar$ is identical to $\opt$, an upper-bound of $\frac{e}{e - 1} \approx 1.58$ was obtained by \cite{CHMS10}, and afterward was shown to be tight by \cite{H13}.
\end{enumerate*}

\vspace{.1in}
\noindent
{\bf OPT vs.\ AR.}
This comparison studies the power of discrimination in auction schemes. Previously, the tight ratios were known in all settings \citep[see][]{M81,H13,AHNPY15} except for the asymmetric regular setting. \cite{HR09} first tackled the problem in this setting: they
\begin{enumerate*}[label = (\alph*), font = {\bfseries}]
\item proved an upper-bound of $4$ \citep[later improved to $e \approx 2.72$ by][]{AHNPY15}, and
\item provided  a $2$-approximation lower-bound instance.
\end{enumerate*}
Although this lower bound of $2$ has never been broken (for a decade), and is widely believed to be the tight ratio, we will demonstrate a sharper $2.15$-approximation instance in \Cref{sec:opt_ar}.

Interestingly,
\begin{enumerate*}[label = (\alph*), font = {\bfseries}]
\item the instance of \cite{HR09} consists of two buyers; yet
\item the two sharper instances in this work respectively involve three and four buyers.
\end{enumerate*}
Given this and other observations (see \Cref{sec:opt_ar}), we conjecture that the tight ratio is reached by an instance with infinite number of buyers.

\vspace{.1in}
\noindent
{\bf Extensions.}
{The above three results also improve other related bounds by implication. For example, due to \cite{AHNPY15}, the tight ratio of the $\opt$ vs.\ $\ap$ problem is somewhere between $[2.23, e]$. This interval now shrinks to $[\C, e]$, by taking into account our tight result of $\C \approx 2.62$ for the $\spm$ vs.\ $\ap$ problem.}

We settle both of the $\spm$ vs.\ $\ap$ problem and the $\ar$ vs.\ $\ap$ problem by formulating a revenue gap as the objective function of a mathematical program. This methodology was initiated by \cite{CGL14} and \cite{AHNPY15}. Employing a similar approach, \cite{BMTT17} recently obtained a tight {\em price of anarchy} for multi-unit auction. Our work further supports the power of this framework in proving tight bounds. En route, we develop an abundance of tools to handle these mathematical programs, which may find extra applications in the future.

For many revenue gaps well understood in the literature \citep[e.g.\ see][]{BK94,HR09,CHMS10,KW12}, a corresponding worst-case instance consists of merely two or several buyers. By contrast, any worst-case instance of our tight results includes infinitely many buyers; the ideas behind these instances may be advantageous to lower-bound analysis of other related problems.

\subsection{Subsequent Work}
In a conference version of this paper \citep{JLTX19}, it is conjectured that the revenue gap between $\opt$ and $\ap$ equals $\C \approx 2.62$, due to the following two observations.
\begin{enumerate}[label = (\alph*), font = {\bfseries}]
	\item From our lower-bound instance (i.e.\ \Cref{exp:spm_ap}) of the $\spm$ vs.\ $\ap$ problem, $\opt$ does extract the same revenue as $\spm$ (see \Cref{lem:spm_tri_best,rmk:lem:spm_tri_best}).
	\item $\ap$ admits the same revenue gap against either $\opt$ or $\spm$, in each of the asymmetric general, i.i.d.\ general, and i.i.d.\ regular settings (see \Cref{tbl:summary2,tbl:summary3}).
\end{enumerate}
This conjecture is confirmed by \cite{JLQTX2019}. Together with the tight revenue gaps developed in this paper, this result suggests the following economic interpretations. To extract more revenue, using price discrimination may be more powerful than conducting anonymous-reserve type auctions -- not only because ``$\spm$ vs.\ $\ap$'' has a greater revenue gap than ``$\ar$ vs.\ $\ap$'', but also because for the worst-case instance of ``$\opt$ v.s.\ $\ap$'', the optimal {\sf Myerson Auction} can be implemented as an $\spm$ mechanism.

%Recall the lattice structure in \Cref{fig:result}. (Notice that $\ar$ and $\spm$ are incomparable; therefore, the ratio between them is not that interesting). In the asymmetric regular setting, although we have attained two tight ratios, there are still three tight ratios remaining open:
%\begin{enumerate*}[label = (\alph*), font = {\bfseries}]
%\item $\opt$ vs.\ $\ar$ (discussed above),
%\item $\opt$ vs.\ $\ap$, and
%\item $\opt$ vs.\ $\spm$.
%\end{enumerate*}
%
%To grasp these comparisons, the main obstacle is that we have yet to fully characterize the $\opt$ revenue. It is easy to formulate the $\ap$ revenue, and the tools to handle the $\ar$ revenue and the $\spm$ revenue have been developed in our work. Even so, formulating each comparison as a mathematical program may serve as the right approach towards the ultimate solution.
%
%%, and the tools developed in this work for analyzing $\spm$ may be helpful.
%
%\paragraph{\fbox{OPT vs.\ AP}}
%As mentioned, this ratio is now restricted to a narrow interval $[\C, e]$. We conjecture that the lower bound of $\C \approx 2.62$ is tight, based on the next two observations:
%\begin{enumerate}[label = (\alph*), font = {\bfseries}]
%\item From our lower-bound instance (i.e.\ \Cref{exp:spm_ap}) of the $\spm$ vs.\ $\ap$ problem, $\opt$ does extract the same revenue as $\spm$ (see \Cref{lem:spm_tri_best,rmk:lem:spm_tri_best}).
%\item $\ap$ admits the same revenue gap against either $\opt$ or $\spm$, in each of the asymmetric general, i.i.d. general, and i.i.d. regular settings (see \Cref{tbl:summary2,tbl:summary3}).
%\end{enumerate}

\subsection{Further Related Work}

This work fits in the ``{\em simple versus optimal}'' paradigm proposed by \cite{HR09}. For a full survey on current progress and future direction in this research agenda, the reader can refer to the book ``Mechanism Design and Approximation'' by \cite{H13}.

\cite{AHNPY15} developed the mathematical-program-based approach in this context, aiming to tackle the $\opt$ vs.\ $\ap$ problem. However, because it is hard to directly quantify the $\opt$ revenue, the authors instead considered an upper-bound revenue formula called {\sf Ex-Ante Relaxation}\footnote{{\sf Ex-Ante Relaxation} is in spirit a ``fake'' mechanism, yet is useful to analyze approximation guarantees of simple mechanisms. To introduce this technique, \cite{CHMS10} actually employed the ideas involved in $\spm$. Later, this technique was further developed in \citet{Y11,A14,CM16}.} \citep[see][]{CHMS10}. They formulated the {\sf Ex-Ante Relaxation} vs.\ $\ap$ problem as a simplified mathematical program, and the resulting tight ratio of $e \approx 2.72$ gives an upper bound of the original $\opt$ vs.\ $\ap$ problem.

\vspace{.1in}
\noindent
{\bf Sequential Posted-Pricing.}
In the literature, researchers also studied the revenue gap between $\opt$ and $\spm$. This comparison admits the same ratio for general distributions as for regular distributions, due to a standard technique called {\em ironing} \citep[see][]{M81}.
In the i.i.d.\ setting, \cite{CFHOV17} established the tight ratio of \footnote{
		\label{footnote:alpha}
		More precisely, constant $\alpha \approx 1.34$ is the unique solution to equation $\int_{0}^{1} \big(x - x \cdot \ln x - 1 + 1 / \alpha\big)^{-1} \cdot \dd x = 1$.} $\alpha \approx 1.34$. Notably, the corresponding worst-case instance consists of infinitely many buyers.
In the asymmetric setting, \cite{CHMS10} first got an upper bound of $\frac{e}{e - 1} \approx 1.58$; \cite{Y11} and \cite{EHKS18} showed that this bound holds in broader settings. Recently, better upper bounds were acquired in \citet{ACK18,BGPPS18,CSZ19}, and the state-of-the-art result is constant\footnote{
		\label{footnote:beta}
		More precisely, constant $\beta = \big(\frac{28}{27} - \frac{1}{e}\big)^{-1} \approx 1.49$.} $\beta \approx 1.50$.
On the other hand, the best known lower bound is actually the tight ratio of $\alpha \approx 1.34$ in i.i.d.\ settings.
For more details about this comparison, the reader can refer to \citep[][Section~1]{CSZ19}.

\vspace{.1in}
\noindent
{\bf Prophet Inequalities.}
\cite{HKS07} first observed connections between the $\opt$ vs.\ $\spm$ problem and the notion of {\em prophet inequality} \citep[e.g.\ see][]{KS77,KS78} in stopping theory. Due to the numerous applications of those inequalities to algorithm design and mechanism design, the last decade has seen extensive progress on them \citep[e.g.\ see also][]{BIK07,KW12,R16,RS17,AEEHK17,EHLM17,EHKS18,focs/DuettingFKL17, ec/CorreaDFS19, ec/DuttingK19, ec/AnariNSS19}. For more literature, the reader can refer to the survey by \cite{L17} and the references therein.

\vspace{.1in}
\noindent
{\bf Multi-Item Mechanism Design.}
In multi-item environments, optimal mechanisms might be more complicated and weird, and there is a rich literature on studying how well simple mechanisms approximate the optima. We can categorize previous work based on valuation functions\footnote{There is a hierarchy: $\mbox{unit-demand \,\&\, additive} \subsetneq \mbox{constraint additive \,\&\, submodular} \subsetneq \mbox{XOS} \subsetneq \mbox{subadditive}$.} of the buyers:
\begin{enumerate*}[label = (\alph*), font = {\bfseries}]
\item in unit-demand settings, see \citet{CHK07,CHMS10,CMS15,CD15,KW12,CDW16};
\item in additive settings, see \citet{HN12,LY13,BILW14,Y15,CDW16,EFFTW17:b}
\item in constraint additive or submodular settings, see \citet{CM16,CZ17}; and
\item in XOS or subadditive settings, see \citet{FGL15,RW15,AS17,CZ17}.
\end{enumerate*}

%correlated buyers
\vspace{.1in}
\noindent
{\bf Beyond Independent Distributions.}
All the above mentioned works assume that there is no correlation among buyers' value distributions. Without this assumption:
\begin{enumerate*}[label = (\alph*), font = {\bfseries}]
\item {\sf Myerson Auction} is not the optimal single-item mechanism \citep[e.g.\ see][]{cremer1988full,DFK15,PP15};
\item to provide good revenue guarantees, simple mechanisms \citep[e.g.\ $k$-look-ahead auction, see][]{R01,CHLW11} come to the rescue once again.
\end{enumerate*}

\section{Notation and Preliminaries}
\label{sec:prelim}
{\bf Notations.}
Denote by $\RRP$ (resp. $\NNP$) the set of non-negative real numbers (resp. positive integers). Given two positive integers $n$ and $m$, where $n \geq m$, denote by $[n]$ the set $\{1, 2, \cdots, n\}$, and by $[m: n]$ the set $\{m, m + 1, \cdots, n\}$. Function $(\cdot)_+$ maps any real number $z \in \RR$ to $\max\{0, z\}$.

\subsection{Distribution Families}
\label{subsec:prelim:distr}

We always focus on the single-item Bayesian mechanism design environment, where $n$ buyers independently draw values $\bids = \{b_i\}_{i = 1}^{n} \in \RRP^n$ from publicly known distributions $\distrs = \{F_i\}_{i = 1}^{n}$. Most of our results are established under the standard {\em regularity} assumption on distributions $\{F_i\}_{i = 1}^{n}$. Besides, the family of triangular distributions will be useful for our lower-bound analysis. We introduce both concepts below, and then elaborate on the mechanisms to be studied.

\vspace{.1in}
\noindent
{\bf Regular Distribution and Revenue-Quantile Curve.}
For any CDF $F$ and the corresponding PDF $f$:
\begin{enumerate*}[label = (\alph*), font = {\bfseries}]
\item the {\em virtual value function} is defined as $\varphi(p) \eqdef p - \frac{1 - F(p)}{f(p)}$; and
\item the {\em revenue-quantile curve} is defined as $r(q) \eqdef q \cdot F^{-1}(1 - q)$.
\end{enumerate*}
By definition, distribution $F$ is regular (i.e.\ $F \in \reg$) iff virtual value function $\varphi$ is non-decreasing, or equivalently, and iff revenue-quantile curve $r$ is a concave function. We interchange these two definitions whenever either is more convenient for our use.

\vspace{.1in}
\noindent
{\bf Triangular Distributions.}
This family of distributions was introduced by \cite{AHNPY15}, named based on the shapes of their revenue-quantile curves (as \Cref{fig:tri} shows). With parameters $v_i \in \RRP$ (i.e.\ the monopoly price) and $q_i \in [0, 1]$ (i.e.\ the monopoly quantile), a triangular distribution $\tri(v_i, q_i)$ has a CDF of $F_i(p) = \frac{(1 - q_i) \cdot p}{(1 - q_i) \cdot p + v_i q_i}$ for any $p \in [0, v_i)$, and $F_i(p) = 1$ for any $p \in [v_i, \infty)$.

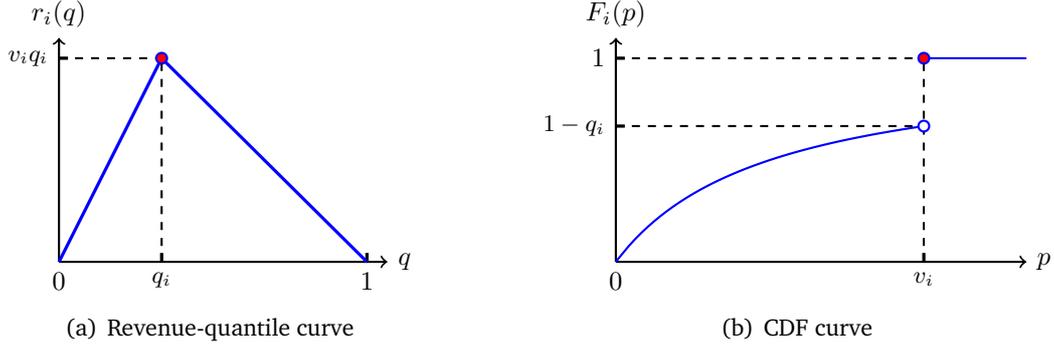
\begin{figure}[!htbp]
\centering
\subfigure[Revenue-quantile curve]{
\begin{tikzpicture}[thick, smooth, domain = 0: 1, scale = 2.7]
\draw[->, thick] (0, 0) -- (1.6, 0);
\draw[->, thick] (0, 0) -- (0, 1.1);
\node[above] at (0, 1.1) {\small $r_i(q)$};
\node[right] at (1.6, 0) {\small $q$};
\draw[dashed] (0.5, 0) -- (0.5, 1);
\draw[dashed] (0, 1) -- (0.5, 1);
\draw[very thick] (1.5, 0pt) -- (1.5, 1.25pt);
\node[below] at (1.5, 0) {\small $1$};
\draw[very thick] (0.5, 0pt) -- (0.5, 1.25pt);
\node[below] at (0.5, 0) {\footnotesize $q_i$};
\draw[very thick] (0pt, 1) -- (1.25pt, 1);
\node[left] at (0, 1) {\footnotesize $v_i q_i$};
\node[below] at (0, 0) {\small $0$};

\draw[color = blue, very thick] (0, 0) -- (0.5, 1) -- (1.5, 0);
\draw [color = blue, fill = red] (0.5, 1) circle(0.75pt);
\end{tikzpicture}
\label{fig:tri_revenue_quntile}
}
\quad\quad\quad
\subfigure[CDF curve]{
\begin{tikzpicture}[thick, smooth, domain = 0: 1.5, scale = 2.7]
\draw[->] (0, 0) -- (2, 0);
\draw[->] (0, 0) -- (0, 1.1);
\node[above] at (0, 1.1) {\small $F_i(p)$};
\node[right] at (2, 0) {\small $p$};
\node[below] at (1.5, 0) {\footnotesize $v_i$};
\node[left] at (0, 1) {\small $1$};
\node[below] at (0, 0) {\small $0$};
\node[left] at (0, 0.6667) {\footnotesize $1 - q_i$};
\draw[style = dashed] (0, 1) -- (1.5, 1);
\draw[style = dashed] (1.5, 0) -- (1.5, 1);

\draw[color = blue] plot (\x, {\x / (\x + 0.75)});
\draw[style = dashed] (0, 0.6667) -- (1.5, 0.6667);
\draw[blue, fill = white] (1.5, 0.6667) circle (0.75pt);
\draw[color = blue] (1.5, 1) -- (2, 1);
\draw [color = blue, fill = red] (1.5, 1) circle(0.75pt);

\draw[very thick] (0pt, 1) -- (1.25pt, 1);
\draw[very thick] (0pt, 0.6667) -- (1.25pt, 0.6667);
\draw[very thick] (1.5, 0pt) -- (1.5, 1.25pt);
\end{tikzpicture}
\label{fig:tri_CDF}
}
\caption{Demonstration for triangular distribution $\tri(v_i, q_i)$.}
\label{fig:tri}
\end{figure}

Particularly, when $N \to \infty$, distribution $\tri(N,\frac{1}{N})$ has a limitation CDF of $F(p) = \frac{p}{p + 1}$, for all $p \in \RRP$. Denote this special limitation distribution by $\tri(\infty)$, which actually is involved in the worst-case instance of the $\spm$ vs.\ $\ap$ problem (see \Cref{sec:spm_ap}), as well as in the improved lower-bound instances of the $\opt$ vs.\ $\ar$ problem (see \Cref{sec:opt_ar}).

\subsection{Mechanisms}
\label{subsec:prelim:mech}

\vspace{.1in}
\noindent
{\bf Anonymous Pricing ($\ap$).}
The seller posts a price of $p \in \RRP$ to all buyers; the item is sold out iff at least one buyer values the item no less than this price. For brevity,
\begin{enumerate*}[label = (\alph*), font = {\bfseries}]
\item we denote by $\ap(p, \distrs) \eqdef p \cdot \big(1 - \prod_{i = 1}^{n} F_i(p)\big)$ the resulting revenue, and
\item by $\ap(\distrs) \eqdef \max_{p \in \RRP} \big\{\ap(p, \distrs)\big\}$ the best revenue among all choices of posted price $p \in \RRP$.
\end{enumerate*}
If there is no ambiguity from the context, we would drop the term $\{F_i\}_{i = 1}^{n}$. The same convention applies to the next three mechanisms as well.

\vspace{.1in}
\noindent
{\bf Sequential Posted-Pricing ($\spm$).}
Such a mechanism is defined by prices $\prices = \{p_i\}_{i = 1}^n \in \RRP^n$ and a permutation $\sigma \in \Pi$: $[n] \to [n]$ over buyers. That is, the seller sequentially posts price $p_{\sigma^{-1}(i)}$ to each $i$-th coming buyer; the first coming buyer with value $b_{\sigma^{-1}(i)} \ge p_{\sigma^{-1}(i)}$ wins the item.
\begin{enumerate*}[label = (\alph*), font = {\bfseries}]
\item We denote by $\spm(\sigma, \prices, \distrs)$ the resulting revenue. Respecting a specific order $\sigma \in \Pi$,
\item let $\spm(\sigma, \distrs) \eqdef \max_{\prices \in \RRP^n} \big\{\spm(\sigma, \prices, \distrs)\big\}$ be the revenue from the optimal pricing strategy.
\item We assume that the seller can choose the prices and the order, which leads to a revenue of $\spm(\distrs) \eqdef \max_{\sigma \in \Pi} \big\{\spm(\sigma, \distrs)\big\}$.
\end{enumerate*}

\vspace{.1in}
\noindent
{\bf Second-Price Auction with Anonymous Reserve ($\ar$).}
The seller sets a reserve price of $p \in \RRP$ to all buyers, and there are three possible outcomes:
\begin{enumerate*}[label = (\alph*), font = {\bfseries}]
\item if no buyer has value $b_i \geq p$, then the auction would be aborted;
\item if exactly one buyer has value $b_i \geq p$, then the item would be sold to him at price $p$;
\item otherwise, the item would be sold to the highest buyer $i^* \eqdef \argmax_{i \in [n]} \big\{b_i\big\}$ at the second highest value $\max_{i \neq i^*} \big\{b_i\big\}$, i.e.\ the well-known {\sf Second-Price Auction}.
\end{enumerate*}

Let $\ar(p, \distrs)$ be the expected revenue from the above scenario, and then define the optimum (among all choices of reserve price $p \in \RRP$) as $\ar(\distrs) \eqdef \max_{p \in \RRP} \big\{\ar(p, \distrs)\big\}$. We postpone the explicit revenue formulas to \Cref{sec:ar_ap}. However, it is easy to see that $\ar(p, \distrs) \geq \ap(p, \distrs)$ for all $p \in \RRP$, and thus $\ar(\distrs) \geq \ap(\distrs)$.

\vspace{.1in}
\noindent
{\bf Myerson Auction ($\opt$).}
Recall virtual value function $\varphi(p) = p - \frac{1 - F(p)}{f(p)}$. {\sf Myerson Auction} runs as follows: upon receiving values $\bids = \{b_i\}_{i = 1}^n$ from the buyers, the seller allocates the item to the buyer with highest virtual value $\max_{i \in [n]} \big\{\varphi_i(b_i)\big\}$ (required to be {\em non-negative}), and charges this buyer a threshold/minimum price for him to keep winning.

The next structural lemma (proved in \Cref{app:prelim}) will be useful for settling the $\spm$ vs.\ $\ap$ problem, and for constructing improved lower-bound instances of the $\opt$ vs.\ $\ar$ problem. The subsequent remark explains the intuition behind this lemma.

\begin{lemma}
\label{lem:spm_tri_best}
For any triangular instance $\{\tri(v_i, q_i)\}_{i = 1}^n$ that $v_1 \geq v_2 \geq \cdots \geq v_n$, it follows that:
\begin{enumerate}[font = {\em\bfseries}]
\item $\opt = \spm = \sum_{i = 1}^n v_i q_i \cdot \prod_{j = 1}^{i - 1} (1 - q_j)$.
\item An optimal $\spm$ lets the buyers come in lexicographic order, and posts price $p_i = v_i$ to each buyer $i \in [n]$.
\end{enumerate}
\end{lemma}

\begin{remark}
\label{rmk:lem:spm_tri_best}
Any triangular distribution $\tri(v_i, q_i)$ is supported on interval $p \in [0,v_i]$, and only the maximum possible value of $b_i = v_i$ corresponds to a non-negative virtual value.

Thus, the {\sf Sequential Posted-Pricing} proposed in Part 2 of \Cref{lem:spm_tri_best} is equivalent to {\sf Myerson Auction}:
\begin{enumerate*}[label = (\alph*), font = {\bfseries}]
\item the seller first sorts all buyers in the decreasing order of $\{v_i\}_{i = 1}^n$; and then
\item lets each buyer $i \in [n]$ sequentially make a take-it-or-leave-it decision at the posted price of $p_i = v_i$; thus,
\item this {\sf Sequential Posted-Pricing} extracts the maximum possible virtual welfare.
\end{enumerate*}
\end{remark}

\section{Sequential Posted-Pricing vs. Anonymous Pricing}
\label{sec:spm_ap}
We first study the revenue gap between $\spm$ and $\ap$, in the setting with (possibly) asymmetric and regular distributions $\distrs = \{F_i\}_{i = 1}^n \in \reg^n$. We interpret this question as the next mathematical program, and safely drop constraint~\eqref{cstr:spm_ap0} on interval $p \in [0, 1]$ as it trivially holds.

\vspace{.1in}
\noindent\fcolorbox{black}{lightgray!50}{\begin{minipage}{0.977\textwidth}
\vspace{-.1in}
\begin{align}
\label{prog:spm_ap0}\tag{P1}
& \max_{\distrs \in \reg^n} && \spm = \max_{\sigma \in \Pi, \prices \in \RRP^n} \big\{\spm(\sigma, \prices)\big\} \\
\label{cstr:spm_ap0}\tag{C1}
& \mbox{subject to:} && \mbox{$\ap(p) = p \cdot \big(1 - \prod_{i = 1}^n F_i(p)\big) \leq 1$}, && \forall p \in (1, \infty)
\end{align}
\end{minipage}}

\vspace{.1in}
By getting the optimal solution to this mathematical program, we derive the next theorem. In particular, \Cref{subapp:spm_ap_calculations} includes the numeric calculations for the constant $\C \approx 2.6202$.

\begin{theorem}
\label{thm:spm_ap}
In asymmetric regular setting, the supremum ratio of $\spm$ to $\ap$ is equal to
\[
\label{eq:spm_ap}
\mbox{$\C \eqdef 2 + \Int{1}{\infty} \big(1 - e^{-\Q(x)}\big) \cdot \dd x \approx 2.6202$},
\]
where function $\Q(p) \eqdef -\ln(1 - p^{-2}) - \frac{1}{2} \cdot \sum_{k = 1}^{\infty} k^{-2} \cdot p^{-2k}$.
\end{theorem}

\noindent
{\bf Proof Overview.}
The upper-bound part of \Cref{thm:spm_ap} is verified in \Cref{subsec:spm_ap_reduction,subsec:spm_ap_solution}. Our proof follows the framework developed by \cite{AHNPY15}, who derived the tight ratio between {\sf Anonymous Pricing} and another benchmark called {\sf Ex-Ante Relaxation}. Since the counterpart program in that work also involves constraint~\eqref{cstr:spm_ap0}, towards characterizing the worst-case instance of Program~\eqref{prog:spm_ap0} we reuse several reductions of Alaei et al. That is, we show that these reductions also work when the objective function is replaced by the $\spm$ revenue.

In \Cref{subsec:spm_ap_reduction} ({\bf Reduction Part}), we first demonstrate that worst-case distributions of Program~\eqref{prog:spm_ap0} w.l.o.g.\ fall into the family of {\em triangular distributions}. As a result, the complicated objective function of Program~\eqref{prog:spm_ap0} can be replaced by another explicit formula (recall Part 1 of \Cref{lem:spm_tri_best}, namely the $\spm$ revenue formula for triangular instance). Afterwards, based on several other reductions, we show that a worst-case instance w.l.o.g.\ includes two special distributions $\tri(\infty)$ and $\tri(1, 1)$, each of which contributes one unit to the $\spm$ revenue. This explains the term of $2$ for the constant $\C \approx 2.6202$.

In \Cref{subsec:spm_ap_solution} ({\bf Optimization Part}), we measure the $\spm$ revenue derived from the remaining distributions. Indeed, we can replace these distributions by a spectrum of ``small'' triangular distributions, under which the $\spm$ revenue increases or remains the same. The $\spm$ revenue from these ``small'' distributions corresponds to the integral term for constant $\C \approx 2.6202$.

Notably, \cite{AHNPY15} did not introduce the special distribution $\tri(1, 1)$, i.e.\ a deterministic value of $1$, in solving their counterpart program. For this reason, they constructed the spectrum of ``small'' triangular distributions in a more complicated way. Here is our interpretation of $\tri(1, 1)$: this distribution never violates the feasibility as constraint~\eqref{cstr:spm_ap0} is restricted on the {\em open} interval $(1, \infty)$, but ensures one unit of the $\spm$ revenue even if all other buyers refute their take-it-or-leave-it offers. Although the ideas behind $\tri(1, 1)$ seem to be intuitive, introducing it not only simplifies our proof, but also paves the way for the follow-up work of \cite{JLQTX2019} that settles the {\sf Myerson Auction} vs. {\sf Anonymous Pricing} problem.

In \Cref{subsec:spm_ap_lower}, we provide a matching lower-bound instance. The high-level idea is simple: use finitely many ``small'' triangular distributions $\{\tri(v_i, q_i)\}_{i = 1}^n$ to surrogate the worst-case instance derived in the upper-bound proof. While being feasible to Program~\eqref{prog:spm_ap0}, this instance generates an $\spm$ revenue arbitrarily close to the constant $\C \approx 2.6202$, when the population $n \in \NNP$ is sufficiently large.

\subsection{Upper-Bound Analysis I: Reduction}
\label{subsec:spm_ap_reduction}

We first show a reduction from a regular instance to another triangular instance, which is very similar to the one by \citet[Lemma~4.1]{AHNPY15}. Under this, the $\spm$ revenue increases or keeps the same, whereas for any posted price $p \in \RRP$ the $\ap(p)$ revenue decreases.
\begin{lemma}
\label{lem:spm_ap_reduction1}
Given any regular instance $\{F_i\}_{i = 1}^n$, there exists a triangular instance $\{\tri(v_i, q_i)\}_{i = 1}^n$ satisfying the following:
\begin{enumerate}[font = {\em\bfseries}]
\item $\spm\big(\{\tri(v_i, q_i)\}_{i = 1}^n\big) \geq \spm\big(\{F_i\}_{i = 1}^n\big)$.
\item $\ap\big(p, \{\tri(v_i, q_i)\}_{i = 1}^n\big) \leq \ap\big(p, \{F_i\}_{i = 1}^n\big)$ for all $p \in \RRP$.
\end{enumerate}
\end{lemma}

\begin{proof}
Given any optimal $\spm$ mechanism for regular instance $\{F_i\}_{i = 1}^n$, denote by $\sigma^* \in \Pi$ the buyer order and $\{p_i^*\}_{i = 1}^n \in \RRP^n$ the posted prices. As mentioned and illustrated in \Cref{fig:transform_rq}, each regular distribution $F_i$ has a concave revenue-quantile curve. We define triangular instance $\{\tri(v_i, q_i)\}_{i = 1}^n$ by letting $v_i \eqdef p_i^*$ and $q_i \eqdef 1 - F_i(p_i^*)$ for each $i \in [n]$.

To see Part 1 of the lemma, by reusing the order $\sigma^* \in \Pi$ and posted prices $\{p_i^*\}_{i = 1}^n \in \RRP^n$, the winning probability of each buyer $i \in [n]$ remains the same. That is, this $\spm$ mechanism gives a revenue of $\spm\big(\sigma^*, \{p_i^*\}_{i = 1}^n, \{F_i\}_{i = 1}^n\big) = \spm\big(\{F_i\}_{i = 1}^n\big)$. This implies Part 1, since the optimal $\spm\big(\{\tri(v_i, q_i)\}_{i = 1}^n\big)$ mechanism at least generates the same revenue.

\begin{figure}[!htbp]
\centering
\subfigure[Revenue-quantile curves]{
\begin{tikzpicture}[thick, smooth, domain = 0: 1.732, scale = 2.5]
\draw[->] (0, 0) -- (1.85, 0);
\draw[->] (0, 0) -- (0, 1.425);
\node[above] at (0, 1.425) {\small $r(q)$};
\node[right] at (1.85, 0) {\small $q$};
\node[below] at (0, 0) {\small $0$};
\draw[very thick, color = blue] plot (\x, {2 * ((\x - 0.732)^3 / 3 - (\x - 0.732)^2 + 2 / 3)});
\draw[very thick, color = red] (0, 0) -- (0.6, 1.297) -- (1.732, 0);
\node[anchor = south west, color = blue] at (1.166, 1.011) {\footnotesize $r_i(q)$};
\draw[dashed] (0.6, 0) -- (0.6, 1.297);
\draw[dashed] (0, 1.297) -- (0.6, 1.297);
\draw[color = blue, fill = red] (0.6, 1.297) circle(0.835pt);
\draw[very thick] (1.732, 0) -- (1.732, 1.5pt);
\node[below] at (1.732, 0) {$1$};
\draw[very thick] (0.6, 0) -- (0.6, 1.5pt);
\node[below] at (0.8, 0) {\footnotesize $q_i = 1 - F_i(p_i^*)$};
\draw[very thick] (0, 1.297, 0) -- (1.5pt, 1.297);
\node[left] at (0, 1.297) {\footnotesize $v_i q_i = r_i(q_i)$};
\end{tikzpicture}
\label{fig:transform_rq}
}
\subfigure[CDF curves]{
\begin{tikzpicture}[thick, smooth, domain = 0: 1.0808, scale = 3.34]
\draw[color = blue, very thick] (0.0000, 0.000) -- (0.0010, 0.001) -- (0.0020, 0.002) -- (0.0030, 0.003) -- (0.0040, 0.004) -- (0.0050, 0.005) -- (0.0060, 0.006) -- (0.0070, 0.007) -- (0.0081, 0.008) -- (0.0091, 0.009) -- (0.0101, 0.010) -- (0.0111, 0.011) -- (0.0121, 0.012) -- (0.0132, 0.013) -- (0.0142, 0.014) -- (0.0152, 0.015) -- (0.0163, 0.016) -- (0.0173, 0.017) -- (0.0183, 0.018) -- (0.0194, 0.019) -- (0.0204, 0.020) -- (0.0214, 0.021) -- (0.0225, 0.022) -- (0.0235, 0.023) -- (0.0246, 0.024) -- (0.0256, 0.025) -- (0.0267, 0.026) -- (0.0277, 0.027) -- (0.0288, 0.028) -- (0.0298, 0.029) -- (0.0309, 0.030) -- (0.0320, 0.031) -- (0.0330, 0.032) -- (0.0341, 0.033) -- (0.0352, 0.034) -- (0.0362, 0.035) -- (0.0373, 0.036) -- (0.0384, 0.037) -- (0.0394, 0.038) -- (0.0405, 0.039) -- (0.0416, 0.040) -- (0.0427, 0.041) -- (0.0438, 0.042) -- (0.0448, 0.043) -- (0.0459, 0.044) -- (0.0470, 0.045) -- (0.0481, 0.046) -- (0.0492, 0.047) -- (0.0503, 0.048) -- (0.0514, 0.049) -- (0.0525, 0.050) -- (0.0536, 0.051) -- (0.0547, 0.052) -- (0.0558, 0.053) -- (0.0569, 0.054) -- (0.0580, 0.055) -- (0.0591, 0.056) -- (0.0602, 0.057) -- (0.0614, 0.058) -- (0.0625, 0.059) -- (0.0636, 0.060) -- (0.0647, 0.061) -- (0.0658, 0.062) -- (0.0670, 0.063) -- (0.0681, 0.064) -- (0.0692, 0.065) -- (0.0704, 0.066) -- (0.0715, 0.067) -- (0.0726, 0.068) -- (0.0738, 0.069) -- (0.0749, 0.070) -- (0.0760, 0.071) -- (0.0772, 0.072) -- (0.0783, 0.073) -- (0.0795, 0.074) -- (0.0806, 0.075) -- (0.0818, 0.076) -- (0.0829, 0.077) -- (0.0841, 0.078) -- (0.0852, 0.079) -- (0.0864, 0.080) -- (0.0876, 0.081) -- (0.0887, 0.082) -- (0.0899, 0.083) -- (0.0911, 0.084) -- (0.0922, 0.085) -- (0.0934, 0.086) -- (0.0946, 0.087) -- (0.0957, 0.088) -- (0.0969, 0.089) -- (0.0981, 0.090) -- (0.0993, 0.091) -- (0.1005, 0.092) -- (0.1016, 0.093) -- (0.1028, 0.094) -- (0.1040, 0.095) -- (0.1052, 0.096) -- (0.1064, 0.097) -- (0.1076, 0.098) -- (0.1088, 0.099) -- (0.1100, 0.100) -- (0.1112, 0.101) -- (0.1124, 0.102) -- (0.1136, 0.103) -- (0.1148, 0.104) -- (0.1160, 0.105) -- (0.1172, 0.106) -- (0.1184, 0.107) -- (0.1197, 0.108) -- (0.1209, 0.109) -- (0.1221, 0.110) -- (0.1233, 0.111) -- (0.1245, 0.112) -- (0.1258, 0.113) -- (0.1270, 0.114) -- (0.1282, 0.115) -- (0.1295, 0.116) -- (0.1307, 0.117) -- (0.1319, 0.118) -- (0.1332, 0.119) -- (0.1344, 0.120) -- (0.1356, 0.121) -- (0.1369, 0.122) -- (0.1381, 0.123) -- (0.1394, 0.124) -- (0.1406, 0.125) -- (0.1419, 0.126) -- (0.1431, 0.127) -- (0.1444, 0.128) -- (0.1456, 0.129) -- (0.1469, 0.130) -- (0.1482, 0.131) -- (0.1494, 0.132) -- (0.1507, 0.133) -- (0.1520, 0.134) -- (0.1532, 0.135) -- (0.1545, 0.136) -- (0.1558, 0.137) -- (0.1570, 0.138) -- (0.1583, 0.139) -- (0.1596, 0.140) -- (0.1609, 0.141) -- (0.1622, 0.142) -- (0.1634, 0.143) -- (0.1647, 0.144) -- (0.1660, 0.145) -- (0.1673, 0.146) -- (0.1686, 0.147) -- (0.1699, 0.148) -- (0.1712, 0.149) -- (0.1725, 0.150) -- (0.1738, 0.151) -- (0.1751, 0.152) -- (0.1764, 0.153) -- (0.1777, 0.154) -- (0.1790, 0.155) -- (0.1803, 0.156) -- (0.1816, 0.157) -- (0.1830, 0.158) -- (0.1843, 0.159) -- (0.1856, 0.160) -- (0.1869, 0.161) -- (0.1882, 0.162) -- (0.1896, 0.163) -- (0.1909, 0.164) -- (0.1922, 0.165) -- (0.1936, 0.166) -- (0.1949, 0.167) -- (0.1962, 0.168) -- (0.1976, 0.169) -- (0.1989, 0.170) -- (0.2002, 0.171) -- (0.2016, 0.172) -- (0.2029, 0.173) -- (0.2043, 0.174) -- (0.2056, 0.175) -- (0.2070, 0.176) -- (0.2083, 0.177) -- (0.2097, 0.178) -- (0.2110, 0.179) -- (0.2124, 0.180) -- (0.2138, 0.181) -- (0.2151, 0.182) -- (0.2165, 0.183) -- (0.2179, 0.184) -- (0.2192, 0.185) -- (0.2206, 0.186) -- (0.2220, 0.187) -- (0.2233, 0.188) -- (0.2247, 0.189) -- (0.2261, 0.190) -- (0.2275, 0.191) -- (0.2289, 0.192) -- (0.2302, 0.193) -- (0.2316, 0.194) -- (0.2330, 0.195) -- (0.2344, 0.196) -- (0.2358, 0.197) -- (0.2372, 0.198) -- (0.2386, 0.199) -- (0.2400, 0.200) -- (0.2414, 0.201) -- (0.2428, 0.202) -- (0.2442, 0.203) -- (0.2456, 0.204) -- (0.2470, 0.205) -- (0.2484, 0.206) -- (0.2498, 0.207) -- (0.2513, 0.208) -- (0.2527, 0.209) -- (0.2541, 0.210) -- (0.2555, 0.211) -- (0.2569, 0.212) -- (0.2584, 0.213) -- (0.2598, 0.214) -- (0.2612, 0.215) -- (0.2627, 0.216) -- (0.2641, 0.217) -- (0.2655, 0.218) -- (0.2670, 0.219) -- (0.2684, 0.220) -- (0.2698, 0.221) -- (0.2713, 0.222) -- (0.2727, 0.223) -- (0.2742, 0.224) -- (0.2756, 0.225) -- (0.2771, 0.226) -- (0.2785, 0.227) -- (0.2800, 0.228) -- (0.2814, 0.229) -- (0.2829, 0.230) -- (0.2844, 0.231) -- (0.2858, 0.232) -- (0.2873, 0.233) -- (0.2888, 0.234) -- (0.2902, 0.235) -- (0.2917, 0.236) -- (0.2932, 0.237) -- (0.2946, 0.238) -- (0.2961, 0.239) -- (0.2976, 0.240) -- (0.2991, 0.241) -- (0.3006, 0.242) -- (0.3021, 0.243) -- (0.3035, 0.244) -- (0.3050, 0.245) -- (0.3065, 0.246) -- (0.3080, 0.247) -- (0.3095, 0.248) -- (0.3110, 0.249) -- (0.3125, 0.250) -- (0.3140, 0.251) -- (0.3155, 0.252) -- (0.3170, 0.253) -- (0.3185, 0.254) -- (0.3200, 0.255) -- (0.3215, 0.256) -- (0.3231, 0.257) -- (0.3246, 0.258) -- (0.3261, 0.259) -- (0.3276, 0.260) -- (0.3291, 0.261) -- (0.3306, 0.262) -- (0.3322, 0.263) -- (0.3337, 0.264) -- (0.3352, 0.265) -- (0.3368, 0.266) -- (0.3383, 0.267) -- (0.3398, 0.268) -- (0.3414, 0.269) -- (0.3429, 0.270) -- (0.3444, 0.271) -- (0.3460, 0.272) -- (0.3475, 0.273) -- (0.3491, 0.274) -- (0.3506, 0.275) -- (0.3522, 0.276) -- (0.3537, 0.277) -- (0.3553, 0.278) -- (0.3568, 0.279) -- (0.3584, 0.280) -- (0.3600, 0.281) -- (0.3615, 0.282) -- (0.3631, 0.283) -- (0.3647, 0.284) -- (0.3662, 0.285) -- (0.3678, 0.286) -- (0.3694, 0.287) -- (0.3709, 0.288) -- (0.3725, 0.289) -- (0.3741, 0.290) -- (0.3757, 0.291) -- (0.3773, 0.292) -- (0.3789, 0.293) -- (0.3804, 0.294) -- (0.3820, 0.295) -- (0.3836, 0.296) -- (0.3852, 0.297) -- (0.3868, 0.298) -- (0.3884, 0.299) -- (0.3900, 0.300) -- (0.3916, 0.301) -- (0.3932, 0.302) -- (0.3948, 0.303) -- (0.3964, 0.304) -- (0.3980, 0.305) -- (0.3996, 0.306) -- (0.4013, 0.307) -- (0.4029, 0.308) -- (0.4045, 0.309) -- (0.4061, 0.310) -- (0.4077, 0.311) -- (0.4093, 0.312) -- (0.4110, 0.313) -- (0.4126, 0.314) -- (0.4142, 0.315) -- (0.4159, 0.316) -- (0.4175, 0.317) -- (0.4191, 0.318) -- (0.4208, 0.319) -- (0.4224, 0.320) -- (0.4240, 0.321) -- (0.4257, 0.322) -- (0.4273, 0.323) -- (0.4290, 0.324) -- (0.4306, 0.325) -- (0.4323, 0.326) -- (0.4339, 0.327) -- (0.4356, 0.328) -- (0.4372, 0.329) -- (0.4389, 0.330) -- (0.4406, 0.331) -- (0.4422, 0.332) -- (0.4439, 0.333) -- (0.4456, 0.334) -- (0.4472, 0.335) -- (0.4489, 0.336) -- (0.4506, 0.337) -- (0.4522, 0.338) -- (0.4539, 0.339) -- (0.4556, 0.340) -- (0.4573, 0.341) -- (0.4590, 0.342) -- (0.4607, 0.343) -- (0.4623, 0.344) -- (0.4640, 0.345) -- (0.4657, 0.346) -- (0.4674, 0.347) -- (0.4691, 0.348) -- (0.4708, 0.349) -- (0.4725, 0.350) -- (0.4742, 0.351) -- (0.4759, 0.352) -- (0.4776, 0.353) -- (0.4793, 0.354) -- (0.4810, 0.355) -- (0.4827, 0.356) -- (0.4845, 0.357) -- (0.4862, 0.358) -- (0.4879, 0.359) -- (0.4896, 0.360) -- (0.4913, 0.361) -- (0.4930, 0.362) -- (0.4948, 0.363) -- (0.4965, 0.364) -- (0.4982, 0.365) -- (0.5000, 0.366) -- (0.5017, 0.367) -- (0.5034, 0.368) -- (0.5052, 0.369) -- (0.5069, 0.370) -- (0.5086, 0.371) -- (0.5104, 0.372) -- (0.5121, 0.373) -- (0.5139, 0.374) -- (0.5156, 0.375) -- (0.5174, 0.376) -- (0.5191, 0.377) -- (0.5209, 0.378) -- (0.5226, 0.379) -- (0.5244, 0.380) -- (0.5262, 0.381) -- (0.5279, 0.382) -- (0.5297, 0.383) -- (0.5315, 0.384) -- (0.5332, 0.385) -- (0.5350, 0.386) -- (0.5368, 0.387) -- (0.5385, 0.388) -- (0.5403, 0.389) -- (0.5421, 0.390) -- (0.5439, 0.391) -- (0.5457, 0.392) -- (0.5475, 0.393) -- (0.5492, 0.394) -- (0.5510, 0.395) -- (0.5528, 0.396) -- (0.5546, 0.397) -- (0.5564, 0.398) -- (0.5582, 0.399) -- (0.5600, 0.400) -- (0.5618, 0.401) -- (0.5636, 0.402) -- (0.5654, 0.403) -- (0.5672, 0.404) -- (0.5690, 0.405) -- (0.5708, 0.406) -- (0.5727, 0.407) -- (0.5745, 0.408) -- (0.5763, 0.409) -- (0.5781, 0.410) -- (0.5799, 0.411) -- (0.5818, 0.412) -- (0.5836, 0.413) -- (0.5854, 0.414) -- (0.5872, 0.415) -- (0.5891, 0.416) -- (0.5909, 0.417) -- (0.5927, 0.418) -- (0.5946, 0.419) -- (0.5964, 0.420) -- (0.5982, 0.421) -- (0.6001, 0.422) -- (0.6019, 0.423) -- (0.6038, 0.424) -- (0.6056, 0.425) -- (0.6075, 0.426) -- (0.6093, 0.427) -- (0.6112, 0.428) -- (0.6130, 0.429) -- (0.6149, 0.430) -- (0.6168, 0.431) -- (0.6186, 0.432) -- (0.6205, 0.433) -- (0.6224, 0.434) -- (0.6242, 0.435) -- (0.6261, 0.436) -- (0.6280, 0.437) -- (0.6299, 0.438) -- (0.6317, 0.439) -- (0.6336, 0.440) -- (0.6355, 0.441) -- (0.6374, 0.442) -- (0.6393, 0.443) -- (0.6411, 0.444) -- (0.6430, 0.445) -- (0.6449, 0.446) -- (0.6468, 0.447) -- (0.6487, 0.448) -- (0.6506, 0.449) -- (0.6525, 0.450) -- (0.6544, 0.451) -- (0.6563, 0.452) -- (0.6582, 0.453) -- (0.6601, 0.454) -- (0.6620, 0.455) -- (0.6639, 0.456) -- (0.6659, 0.457) -- (0.6678, 0.458) -- (0.6697, 0.459) -- (0.6716, 0.460) -- (0.6735, 0.461) -- (0.6755, 0.462) -- (0.6774, 0.463) -- (0.6793, 0.464) -- (0.6812, 0.465) -- (0.6832, 0.466) -- (0.6851, 0.467) -- (0.6870, 0.468) -- (0.6890, 0.469) -- (0.6909, 0.470) -- (0.6929, 0.471) -- (0.6948, 0.472) -- (0.6967, 0.473) -- (0.6987, 0.474) -- (0.7006, 0.475) -- (0.7026, 0.476) -- (0.7045, 0.477) -- (0.7065, 0.478) -- (0.7085, 0.479) -- (0.7104, 0.480) -- (0.7124, 0.481) -- (0.7143, 0.482) -- (0.7163, 0.483) -- (0.7183, 0.484) -- (0.7202, 0.485) -- (0.7222, 0.486) -- (0.7242, 0.487) -- (0.7262, 0.488) -- (0.7281, 0.489) -- (0.7301, 0.490) -- (0.7321, 0.491) -- (0.7341, 0.492) -- (0.7361, 0.493) -- (0.7380, 0.494) -- (0.7400, 0.495) -- (0.7420, 0.496) -- (0.7440, 0.497) -- (0.7460, 0.498) -- (0.7480, 0.499) -- (0.7500, 0.500) -- (0.7520, 0.501) -- (0.7540, 0.502) -- (0.7560, 0.503) -- (0.7580, 0.504) -- (0.7600, 0.505) -- (0.7621, 0.506) -- (0.7641, 0.507) -- (0.7661, 0.508) -- (0.7681, 0.509) -- (0.7701, 0.510) -- (0.7721, 0.511) -- (0.7742, 0.512) -- (0.7762, 0.513) -- (0.7782, 0.514) -- (0.7802, 0.515) -- (0.7823, 0.516) -- (0.7843, 0.517) -- (0.7863, 0.518) -- (0.7884, 0.519) -- (0.7904, 0.520) -- (0.7925, 0.521) -- (0.7945, 0.522) -- (0.7965, 0.523) -- (0.7986, 0.524) -- (0.8006, 0.525) -- (0.8027, 0.526) -- (0.8047, 0.527) -- (0.8068, 0.528) -- (0.8089, 0.529) -- (0.8109, 0.530) -- (0.8130, 0.531) -- (0.8150, 0.532) -- (0.8171, 0.533) -- (0.8192, 0.534) -- (0.8212, 0.535) -- (0.8233, 0.536) -- (0.8254, 0.537) -- (0.8275, 0.538) -- (0.8295, 0.539) -- (0.8316, 0.540) -- (0.8337, 0.541) -- (0.8358, 0.542) -- (0.8379, 0.543) -- (0.8400, 0.544) -- (0.8420, 0.545) -- (0.8441, 0.546) -- (0.8462, 0.547) -- (0.8483, 0.548) -- (0.8504, 0.549) -- (0.8525, 0.550) -- (0.8546, 0.551) -- (0.8567, 0.552) -- (0.8588, 0.553) -- (0.8609, 0.554) -- (0.8630, 0.555) -- (0.8652, 0.556) -- (0.8673, 0.557) -- (0.8694, 0.558) -- (0.8715, 0.559) -- (0.8736, 0.560) -- (0.8757, 0.561) -- (0.8779, 0.562) -- (0.8800, 0.563) -- (0.8821, 0.564) -- (0.8842, 0.565) -- (0.8864, 0.566) -- (0.8885, 0.567) -- (0.8906, 0.568) -- (0.8928, 0.569) -- (0.8949, 0.570) -- (0.8971, 0.571) -- (0.8992, 0.572) -- (0.9014, 0.573) -- (0.9035, 0.574) -- (0.9057, 0.575) -- (0.9078, 0.576) -- (0.9100, 0.577) -- (0.9121, 0.578) -- (0.9143, 0.579) -- (0.9164, 0.580) -- (0.9186, 0.581) -- (0.9208, 0.582) -- (0.9229, 0.583) -- (0.9251, 0.584) -- (0.9273, 0.585) -- (0.9294, 0.586) -- (0.9316, 0.587) -- (0.9338, 0.588) -- (0.9360, 0.589) -- (0.9381, 0.590) -- (0.9403, 0.591) -- (0.9425, 0.592) -- (0.9447, 0.593) -- (0.9469, 0.594) -- (0.9491, 0.595) -- (0.9512, 0.596) -- (0.9534, 0.597) -- (0.9556, 0.598) -- (0.9578, 0.599) -- (0.9600, 0.600) -- (0.9622, 0.601) -- (0.9644, 0.602) -- (0.9666, 0.603) -- (0.9688, 0.604) -- (0.9711, 0.605) -- (0.9733, 0.606) -- (0.9755, 0.607) -- (0.9777, 0.608) -- (0.9799, 0.609) -- (0.9821, 0.610) -- (0.9844, 0.611) -- (0.9866, 0.612) -- (0.9888, 0.613) -- (0.9910, 0.614) -- (0.9933, 0.615) -- (0.9955, 0.616) -- (0.9977, 0.617) -- (1.0000, 0.618) -- (1.0022, 0.619) -- (1.0044, 0.620) -- (1.0067, 0.621) -- (1.0089, 0.622) -- (1.0112, 0.623) -- (1.0134, 0.624) -- (1.0157, 0.625) -- (1.0179, 0.626) -- (1.0202, 0.627) -- (1.0224, 0.628) -- (1.0247, 0.629) -- (1.0269, 0.630) -- (1.0292, 0.631) -- (1.0315, 0.632) -- (1.0337, 0.633) -- (1.0360, 0.634) -- (1.0383, 0.635) -- (1.0405, 0.636) -- (1.0428, 0.637) -- (1.0451, 0.638) -- (1.0474, 0.639) -- (1.0496, 0.640) -- (1.0519, 0.641) -- (1.0542, 0.642) -- (1.0565, 0.643) -- (1.0588, 0.644) -- (1.0611, 0.645) -- (1.0634, 0.646) -- (1.0657, 0.647) -- (1.0679, 0.648) -- (1.0702, 0.649) -- (1.0725, 0.650) -- (1.0748, 0.651) -- (1.0772, 0.652) -- (1.0795, 0.653) -- (1.0818, 0.654) -- (1.0841, 0.655) -- (1.0864, 0.656) -- (1.0887, 0.657) -- (1.0910, 0.658) -- (1.0933, 0.659) -- (1.0956, 0.660) -- (1.0980, 0.661) -- (1.1003, 0.662) -- (1.1026, 0.663) -- (1.1049, 0.664) -- (1.1073, 0.665) -- (1.1096, 0.666) -- (1.1119, 0.667) -- (1.1143, 0.668) -- (1.1166, 0.669) -- (1.1190, 0.670) -- (1.1213, 0.671) -- (1.1236, 0.672) -- (1.1260, 0.673) -- (1.1283, 0.674) -- (1.1307, 0.675) -- (1.1330, 0.676) -- (1.1354, 0.677) -- (1.1377, 0.678) -- (1.1401, 0.679) -- (1.1425, 0.680) -- (1.1448, 0.681) -- (1.1472, 0.682) -- (1.1495, 0.683) -- (1.1519, 0.684) -- (1.1543, 0.685) -- (1.1567, 0.686) -- (1.1590, 0.687) -- (1.1614, 0.688) -- (1.1638, 0.689) -- (1.1662, 0.690) -- (1.1685, 0.691) -- (1.1709, 0.692) -- (1.1733, 0.693) -- (1.1757, 0.694) -- (1.1781, 0.695) -- (1.1805, 0.696) -- (1.1829, 0.697) -- (1.1853, 0.698) -- (1.1877, 0.699) -- (1.1901, 0.700) -- (1.1925, 0.701) -- (1.1949, 0.702) -- (1.1973, 0.703) -- (1.1997, 0.704) -- (1.2021, 0.705) -- (1.2045, 0.706) -- (1.2069, 0.707) -- (1.2093, 0.708) -- (1.2118, 0.709) -- (1.2142, 0.710) -- (1.2166, 0.711) -- (1.2190, 0.712) -- (1.2214, 0.713) -- (1.2239, 0.714) -- (1.2263, 0.715) -- (1.2287, 0.716) -- (1.2312, 0.717) -- (1.2336, 0.718) -- (1.2360, 0.719) -- (1.2385, 0.720) -- (1.2409, 0.721) -- (1.2434, 0.722) -- (1.2458, 0.723) -- (1.2483, 0.724) -- (1.2507, 0.725) -- (1.2532, 0.726) -- (1.2556, 0.727) -- (1.2581, 0.728) -- (1.2605, 0.729) -- (1.2630, 0.730) -- (1.2654, 0.731) -- (1.2679, 0.732) -- (1.2704, 0.733) -- (1.2728, 0.734) -- (1.2753, 0.735) -- (1.2778, 0.736) -- (1.2803, 0.737) -- (1.2827, 0.738) -- (1.2852, 0.739) -- (1.2877, 0.740) -- (1.2902, 0.741) -- (1.2927, 0.742) -- (1.2951, 0.743) -- (1.2976, 0.744) -- (1.3001, 0.745) -- (1.3026, 0.746) -- (1.3051, 0.747) -- (1.3076, 0.748) -- (1.3101, 0.749) -- (1.3126, 0.750) -- (1.3151, 0.751) -- (1.3176, 0.752) -- (1.3201, 0.753) -- (1.3226, 0.754) -- (1.3251, 0.755) -- (1.3276, 0.756) -- (1.3302, 0.757) -- (1.3327, 0.758) -- (1.3352, 0.759) -- (1.3377, 0.760) -- (1.3402, 0.761) -- (1.3428, 0.762) -- (1.3453, 0.763) -- (1.3478, 0.764) -- (1.3503, 0.765) -- (1.3529, 0.766) -- (1.3554, 0.767) -- (1.3579, 0.768) -- (1.3605, 0.769) -- (1.3630, 0.770) -- (1.3656, 0.771) -- (1.3681, 0.772) -- (1.3706, 0.773) -- (1.3732, 0.774) -- (1.3757, 0.775) -- (1.3783, 0.776) -- (1.3809, 0.777) -- (1.3834, 0.778) -- (1.3860, 0.779) -- (1.3885, 0.780) -- (1.3911, 0.781) -- (1.3937, 0.782) -- (1.3962, 0.783) -- (1.3988, 0.784) -- (1.4014, 0.785) -- (1.4039, 0.786) -- (1.4065, 0.787) -- (1.4091, 0.788) -- (1.4117, 0.789) -- (1.4142, 0.790) -- (1.4168, 0.791) -- (1.4194, 0.792) -- (1.4220, 0.793) -- (1.4246, 0.794) -- (1.4272, 0.795) -- (1.4298, 0.796) -- (1.4324, 0.797) -- (1.4350, 0.798) -- (1.4375, 0.799) -- (1.4402, 0.800) -- (1.4428, 0.801) -- (1.4454, 0.802) -- (1.4480, 0.803) -- (1.4506, 0.804) -- (1.4532, 0.805) -- (1.4558, 0.806) -- (1.4584, 0.807) -- (1.4610, 0.808) -- (1.4636, 0.809) -- (1.4663, 0.810) -- (1.4689, 0.811) -- (1.4715, 0.812) -- (1.4741, 0.813) -- (1.4768, 0.814) -- (1.4794, 0.815) -- (1.4820, 0.816) -- (1.4847, 0.817) -- (1.4873, 0.818) -- (1.4899, 0.819) -- (1.4926, 0.820) -- (1.4952, 0.821) -- (1.4979, 0.822) -- (1.5000, 0.8228);
\node[anchor = north west, color = blue] at (0.54, 0.3888) {\small $F_i(p)$};
\draw[->] (0, 0) -- (1.5, 0);
\draw[->] (0, 0) -- (0, 1.1);
\node[above] at (0, 1.1) {\footnotesize $F_i(p)$};
\node[right] at (1.5, 0) {\small $p$};
\node[below] at (1.0808, 0) {\footnotesize $v_i = p_i^*$};
\node[left] at (0, 1) {\small $1$};
\node[below] at (0, 0) {\small $0$};
\node[left] at (0, 0.6536) {\footnotesize $1 - q_i = F_i(p_i^*)$};
\draw[style = dashed] (0, 1) -- (1.0808, 1);
\draw[style = dashed] (1.0808, 0) -- (1.0808, 1);

\draw[color = red] plot (\x, {\x / (\x + 0.5728)});
\draw[style = dashed] (0, 0.6536) -- (1.0808, 0.6536);
\draw[red, fill = white] (1.0808, 0.6536) circle (0.625pt);
\draw[color = red] (1.0808, 1) -- (1.5, 1);
\draw[color = red, fill = red] (1.0808, 1) circle(0.625pt);

\draw[very thick] (0pt, 1) -- (1.25pt, 1);
\draw[very thick] (0pt, 0.6536) -- (1.25pt, 0.6536);
\draw[very thick] (1.0808, 0pt) -- (1.0808, 1.25pt);
\end{tikzpicture}
\label{fig:transform_cdf}
}
\caption{Transformation from regular instance $\{F_i\}_{i = 1}^n$ to triangular instance $\{\tri(v_i, q_i)\}_{i = 1}^n$.}
\label{fig:transform}
\end{figure}
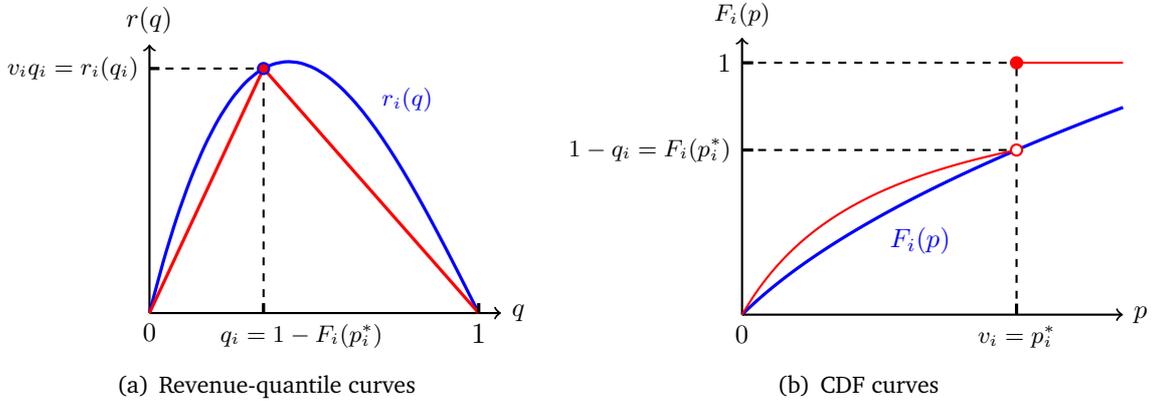

For brevity, denote by $\FF_i$ the CDF of triangular distribution $\tri(v_i, q_i)$, and $\rr_i$ the revenue quantile curve. We infer from \Cref{fig:transform_cdf} that $\rr_i(q) \leq r_i(q)$ for all $q \in [0, 1]$, i.e.\ $\FF_i(p) \geq F_i(p)$ for all $p \in \RRP$. Such stochastic domination directly indicates Part 2 (recall the $\ap(p)$ revenue formula presented in \Cref{subsec:prelim:mech}), and thus completes the proof of \Cref{lem:spm_ap_reduction1}.
\end{proof}

To find the optimal solution to Program~\eqref{prog:spm_ap0}, we safely restrict our attention to triangular instances. For convenience, we re-index a triangular instance such that $v_1 \geq v_2 \geq \cdots \geq v_n$, and still denote by $F_i$ the CDF of each distribution $\tri(v_i, q_i)$. Thus, the $\spm$ revenue formula in Part 1 of \Cref{lem:spm_tri_best} is applicable. Further, we can formulate the $\ap$ revenue explicitly:
\begin{align*}
\mbox{$\ap(p) = p \cdot \Big\{1 - \prod_{i: v_i \geq p} \big[\frac{(1 - q_i) \cdot p}{(1 - q_i) \cdot p + v_i q_i}\big]\Big\} \leq 1$}, && \forall p \in \RRP,
\end{align*}
which converts into constraint~\eqref{cstr:spm_ap0.5} after being rearranged. Therefore, the next mathematical program admits the same optimal objective value as Program~\eqref{prog:spm_ap0}. As mentioned before, the new program has exactly the same set of constraints as in \cite{AHNPY15}, but a different objective function.

\vspace{.1in}
\noindent\fcolorbox{black}{lightgray!50}{\begin{minipage}{0.977\textwidth}
\vspace{-.1in}
\begin{align}
\label{prog:spm_ap0.5}\tag{P2}
& \max_{\{\tri(v_i, q_i)\}^n} && \mbox{$\spm = \sum_{i = 1}^n v_i q_i \cdot \prod_{j = 1}^{i - 1} (1 - q_j)$} \\
\label{cstr:spm_ap0.5}\tag{C2}
& \mbox{subject to:} && \mbox{$\sum_{i: v_i \geq p} \ln\big(1 + \frac{v_i q_i}{1 - q_i} \cdot \frac{1}{p}\big) \leq -\ln(1 - p^{-1})$}, && \forall p \in (1, \infty) \\
\notag
& && v_1 \geq v_2 \geq \cdots \geq v_n
\end{align}
\end{minipage}}

\vspace{.1in}
\noindent
The following propositions further narrow down the optimal-solution space of Program~\eqref{prog:spm_ap0.5}:
\begin{enumerate}[font = {\bfseries}]
\item W.l.o.g.\ a worst-case instance satisfies the {\em strict} inequalities $v_1 > v_2 > \cdots > v_n$. Namely, if two consecutive triangular distributions have the same $v_i = v_{i + 1}$, we can replace them (as a whole) with another triangular distribution, under which the $\spm$ revenue remains the same yet constraint~\eqref{cstr:spm_ap0.5} still holds for all $p \in (1, \infty)$.
    
In \Cref{subapp:spm_ap_reductions}, we formalize this proposition as \Cref{lem:spm_ap_reduction3}.

\item W.l.o.g.\ a worst-case instance includes a special distribution $\tri(\infty)$, whose CDF equals $F_0(p) = \frac{p}{p + 1}$ for all $p \in \RRP$. Hence, we can transform Program~\eqref{prog:spm_ap0.5} as follows:
    \begin{enumerate}[label = (\roman*), font = {\bfseries}]
    \item The objective function now becomes $\spm = 1 + \sum_{i = 1}^n v_i q_i \cdot \prod_{j = 1}^{i - 1} (1 - q_j)$.
    \item Constraint~\eqref{cstr:spm_ap0.5} becomes the following: for any $p \in (1, \infty)$,
    \[
    \mbox{$\sum_{i: v_i \geq p} \ln\Big(1 + \frac{v_i q_i}{1 - q_i} \cdot \frac{1}{p}\Big) \leq -\ln(1 - p^{-1}) + \ln F_0(p) = -\ln(1 - p^{-2})$}.
    \]
    \end{enumerate}
    In \Cref{subapp:spm_ap_reduction4}, we formalize this proposition as \Cref{lem:spm_ap_reduction4}; its proof relies on a reduction very similar to the one by \citet[Lemma~4.2]{AHNPY15}.

\item W.l.o.g.\ a worst-case instance includes a special distribution $\tri(1, 1)$, i.e.\ a special buyer with a deterministic value of $1$. In particular, this buyer has no effect on constraint~\eqref{cstr:spm_ap0.5}, and serves as the default winner in the $\spm$ mechanism (when every other buyer has a value of $b_i \leq 1$), hence a guaranteed $\spm$ revenue of $1$.

Given these, w.l.o.g.\ $v_n \geq 1$, and the objective function of Program~\eqref{prog:spm_ap0.5} now increases to $\spm = 2 + \sum_{i = 1}^n (v_i - 1) \cdot q_i \cdot \prod_{j = 1}^{i - 1} (1 - q_j)$.
\end{enumerate}
For simplicity, we never explicitly mention the special distributions $\tri(\infty)$ and $\tri(1, 1)$ elsewhere. To summarize, the following mathematical program admits the same optimal solution as Program~\eqref{prog:spm_ap0} and Program~\eqref{prog:spm_ap0.5}.

\vspace{.1in}
\noindent\fcolorbox{black}{lightgray!50}{
\begin{minipage}{0.977\textwidth}
\vspace{-.1in}
\begin{align}
\label{prog:spm_ap1}\tag{P3}
& \max_{\big\{\tri(v_i, q_i)\big\}^n} && \mbox{$\spm = 2 + \sum_{i = 1}^n (v_i - 1) \cdot q_i \cdot \prod_{j = 1}^{i - 1} (1 - q_j)$} \\
\label{cstr:spm_ap1}\tag{C3.1}
& \mbox{subject to:} && \mbox{$\sum_{i: v_i \geq p} \ln\big(1 + \frac{v_i q_i}{1 - q_i} \cdot \frac{1}{p}\big) \leq -\ln(1 - p^{-2})$}, && \forall p \in (1, \infty) \\
\notag
& &&  v_1 > v_2 > \cdots > v_n \geq 1
\end{align}
\end{minipage}}

\subsection{Upper-Bound Analysis II: Optimization}
\label{subsec:spm_ap_solution}

We continue to handle Program~\eqref{prog:spm_ap1} by following the proof framework of \cite{AHNPY15}. It is easy to check that $\ln\big(1 + \frac{x}{p}\big) \geq \frac{1}{p} \cdot \ln(1 + x)$ when $p \geq 1$ and $x \geq 0$. Applying this to the $\lhs$ of constraint~\eqref{cstr:spm_ap1} and then restricting ourselves to prices $p \in \{v_k\}_{k = 1}^{n}$, we are left with constraints
\begin{equation}
\label{cstr:spm_ap4}\tag{C3.2}
\mbox{$\sum_{i = 1}^k \ln\big(1 + \frac{v_i q_i}{1 - q_i}\big) \leq \R(v_k)$}, \quad\quad\quad\quad \forall k \in [n],
\end{equation}
where $\R(p) \eqdef -p \cdot \ln\big(1 - p^{-2}\big)$. Although the constraint gets relaxed, later we will see that the upper bound of $\C = 2 + \Int{1}{\infty} \big(1 - e^{-\Q(x)}\big) \cdot \dd x \approx 2.6202$ still holds.

Recall the function $\Q(p) = -\ln\big(1 - p^{-2}\big) - \frac{1}{2} \cdot \sum_{k = 1}^{\infty} k^{-2} \cdot p^{-2k}$ from \Cref{thm:spm_ap}. The reader may wonder how to relate constraint~\eqref{cstr:spm_ap4} and the constant $\C$, which involve two different functions $\R(p)$ and $\Q(p)$. Actually, these two functions together satisfy the ODE in Part 1 of the next \Cref{lem:spm_ap_RQ} (proved in \Cref{subapp:spm_ap_solution}).

In what follows, we characterize a worst-case instance through its monopoly price $\{v_i\}_{i = 1}^n$ and quantile $\{q_i\}_{i = 1}^n$. Particularly, we will construct an interim instance, whose $\{q_i\}_{i = 1}^n$ satisfy an $\R(p)$-based recursive formula of $\{v_i\}_{i = 1}^n$ (see constraint~\eqref{cstr:spm_ap5}). By contrast, the parameters of our worst-case instance satisfy another $\Q(p)$-based formula. It turns out that the worst-case instance generates a better $\spm$ revenue than the interim instance (see \Cref{lem:spm_ap_inequality}), for which the proof crucially relies on \Cref{lem:spm_ap_RQ}.

%\footnote{Note that constraint~\eqref{cstr:spm_ap4} contains $\{v_i\}_{i = 1}^n$ on its both hand sides but $\{q_i\}_{i = 1}^n$ just on its $\lhs$.}

\begin{fact}
\label{lem:spm_ap_RQ}
For functions $\R(p) = -p \cdot \ln(1 - p^{-2})$ and $\Q(p) = -\ln(1 - p^{-2}) - \frac{1}{2} \cdot \sum_{k = 1}^{\infty} k^{-2} \cdot p^{-2k}$, the following holds:
\begin{enumerate}[font = {\em\bfseries}]
\item $\R'(p) = p \cdot \Q'(p) < 0$ for any $p \in (1, \infty)$.
\item $\lim\limits_{p \to 1^+} \R(p) = \lim\limits_{p \to 1^+} \Q(p) = \infty$ and $\lim\limits_{p \to \infty} \R(p) = \lim\limits_{p \to \infty} \Q(p) = 0$.
\end{enumerate}
\end{fact}

With respect to the relaxed program with constraint~\eqref{cstr:spm_ap4}, the optimal solution turns out to make constraint~\eqref{cstr:spm_ap4} tight for every $k \in [n]$. This is formalized as \Cref{lem:spm_ap_relexation}, whose proof is very similar to that of \citet[Lemma~4.4]{AHNPY15} and is deferred to \Cref{subapp:spm_ap_solution}.

\begin{lemma}
\label{lem:spm_ap_relexation}
Given any triangular instance $\{\tri(v_i, q_i)\}_{i = 1}^n$ with constraint~\eqref{cstr:spm_ap4} loose for some $i \in [n]$, there exists another triangular instance $\{\tri(\vv_i, \qq_i)\}_{i = 1}^n$ satisfying the following:
\begin{enumerate}[font = {\em\bfseries}]
\item Constraint~\eqref{cstr:spm_ap4} still holds, and is tight for each $i \in [n]$.
\item $\spm\big(\{\tri(\vv_i, \qq_i)\}_{i = 1}^n\big) \geq \spm\big(\{\tri(v_i, q_i)\}_{i = 1}^n\big)$.
\end{enumerate}
\end{lemma}

For ease of notation, let $v_0 \eqdef \infty$. Based on Part 1 of \Cref{lem:spm_ap_relexation} and Part 2 of \Cref{lem:spm_ap_RQ} (namely $\R(v_0) = 0$), we get a recursive formula from constraint~\eqref{cstr:spm_ap4}: w.l.o.g.\ a worst-case instance $\{\tri(v_i, q_i)\}_{i = 1}^n$ satisfies $\ln\big(1 + \frac{v_k q_k}{1 - q_k}\big) = \R(v_k) - \R(v_{k - 1})$, namely
\begin{align}
\label{cstr:spm_ap5}\tag{C3.3}
q_k = \frac{e^{\R(v_k) - \R(v_{k - 1})} - 1}{v_k + e^{\R(v_k) - \R(v_{k - 1})} - 1}, && \forall k \in [n].
\end{align}
Based on this recursive formula, we prove the following lemma in \Cref{subapp:spm:ineq}.

\begin{lemma}
\label{lem:spm_ap_inequality}
Given any triangular instance $\{\tri(v_i, q_i)\}_{i = 1}^n$ such that $v_1 > v_2 > \cdots > v_n \geq 1$ and the parameters $\{q_i\}_{i = 1}^n$ satisfying constraint~\eqref{cstr:spm_ap5}. The following holds for each $k \in [n]$:
\begin{align*}
\mbox{$(v_k - 1) \cdot q_k \cdot \prod_{j = 1}^{k - 1} (1 - q_j)$}
& \leq \mbox{$\Int{v_k}{v_{k - 1}} (x - 1) \cdot \big(-\Q'(x)\big) \cdot e^{-\Q(x)} \cdot \dd x$} \\
& \equiv \mbox{$\Int{v_k}{v_{k - 1}} (x - 1) \cdot \dd e^{-\Q(x)}$}.
\end{align*}
\end{lemma}

Roughly speaking, in \Cref{lem:spm_ap_inequality} we replace the triangular instance $\{\tri(v_i, q_i)\}_{i = 1}^n$ by a spectrum of ``small'' triangular distributions (denoted by $\mathcal{I}$ for ease of notation). That is, this spectrum contains infinitely many buyers, but each specific buyer accepts his $\spm$ offer with negligible probability. Below, we interpret \Cref{lem:spm_ap_inequality} via the {\em revenue-equivalence theorem} of \citep{M81}, and safely interchange ``$\spm$ revenue'' and ``virtual welfare''.

With respect to $\{\tri(v_i, q_i)\}_{i = 1}^n$, the virtual welfare $\mathbf{W}_{\tri}$ is a {\em discrete} random variable. Concretely, we have $\mathbf{W}_{\tri} = v_k$ iff the index-$k$ buyer accepts his take-it-or-leave-it offer but all the smaller-index buyers refute their offers.\footnote{At the monopoly price $v_k$, the distribution $\tri(v_k, q_k)$ has a {\em probability mass} of $q_k$. Thus, when the index-$k$ buyer has a value of $v_k$, his virtual value also equals $v_k$.} This event occurs with probability $q_k \cdot \prod_{j = 1}^{k - 1} (1 - q_j)$. But recall \Cref{subsec:spm_ap_reduction} that the special distribution $\tri(1, 1)$ has a deterministic value of $1$, and we have already incorporated this one unit into the $\spm$ revenue. Thereby, in expectation the remaining virtual welfare $\E\big[(\mathbf{W}_{\tri} - 1)_+\big]$ from $\{\tri(v_i, q_i)\}_{i = 1}^n$ is exactly the telescoping sum (over all $k \in [n]$) of the $\lhs$ formulas in \Cref{lem:spm_ap_inequality}:
\[
\mbox{$\E\big[(\mathbf{W}_{\tri} - 1)_+\big] = \sum_{k = 1}^{n} (v_k - 1) \cdot q_k \cdot \prod_{j = 1}^{k - 1} (1 - q_j)$}.
\]

Respecting the spectral instance $\mathcal{I}$, the virtual welfare $\mathbf{W}_{\mathcal{I}}$ is almost a {\em continuous} random variable. Particularly, unlike $\mathbf{W}_{\tri}$ that is supported on $\{v_1, v_2, \cdots, v_n\}$, $\mathbf{W}_{\mathcal{I}}$ can take any value between $[v_n, \infty) \equiv \bigcup_{k = 1}^{n} [v_k, v_{k - 1})$.\footnote{In fact, we may also have $\mathbf{W}_{\tri} = 1$, which occurs when only the special distribution $\tri(1, 1)$ accepts the offer.} In this range, we must work with probability density rather than probability mass: $\mathbf{W}_{\mathcal{I}}$ turns out to follow the CDF
\[
\Pr\{\mathbf{W}_{\mathcal{I}} \leq x\} = e^{-\Q(x)},
\]
for any $x \in [v_n, \infty)$. We have $\Pr\{v_n \leq \mathbf{W}_{\mathcal{I}} < \infty\} = 1 - e^{-\Q(v_n)}$ because $\Q(\infty) = 0$ (see Part 2 of \Cref{lem:spm_ap_RQ}). The remaining probability mass of $e^{-\Q(v_n)}$ is at $\mathbf{W}_{\mathcal{I}} = 1$, as the special distribution $\tri(1, 1)$ has a deterministic value of $1$. Again, the expected virtual welfare $\E\big[(\mathbf{W}_{\mathcal{I}} - 1)_+\big]$ from the spectral instance $\mathcal{I}$ is exactly the telescoping sum of the $\rhs$ formulas in \Cref{lem:spm_ap_inequality}:
\[
\mbox{$\E\big[(\mathbf{W}_{\mathcal{I}} - 1)_+\big]$}
= \mbox{$\Int{v_n}{\infty} (x - 1) \cdot \dd e^{-\Q(x)}$}.
\]

The lemma shows that the spectral instance $\mathcal{I}$ gives a higher $\spm$ revenue than the triangular instance $\{\tri(v_i, q_i)\}_{i = 1}^n$. As a result, w.l.o.g.\ the worst-case of our mathematical program contains a continuum of ``small'' buyers. Formally, by applying \Cref{lem:spm_ap_inequality} for all $k \in [n]$, we settle the upper-bound part of \Cref{thm:spm_ap} as follows:
\begin{align*}
\spm
& \overset{\eqref{prog:spm_ap1}}{=} \mbox{$2 + \sum_{i = 1}^n (v_i - 1) \cdot q_i \cdot \prod_{j = 1}^{i - 1} (1 - q_j)$} \\
& \,\,\leq\,\, \mbox{$2 + \Int{v_n}{\infty} (x - 1) \cdot \big(-\Q'(x)\big) \cdot e^{-\Q(x)} \cdot \dd x$}
&& \mbox{\tt (by \Cref{lem:spm_ap_inequality})} \\
& \,\,\leq\,\, \mbox{$2 + \Int{1}{\infty} (x - 1) \cdot \big(-\Q'(x)\big) \cdot e^{-\Q(x)} \cdot \dd x$}
&& \mbox{\tt (as $v_n \geq 1$)}  \\
& \,\,=\,\, \mbox{$2 + \Int{1}{\infty} \big(1 -  e^{-\Q(x)}\big) \cdot \dd x = \C \approx 2.6202$}.
&& \mbox{\tt (integration by parts)}
\end{align*}

%\green{The above lemma bridges a certain triangular instance  and a spectrum of ``small'' triangular distributions $\{\tri(\vv_i, \qq_i)\}_{i \in \NNP}$ that satisfies the following:
%\begin{enumerate}[label = (\alph*), font = {\bfseries}]
%\item Parameters $\{\vv_i\}_{i \in \NNP}$ all lie in interval $[v_n, \infty)$, and $\max_{i \in \NNP} \{\qq_i\} \to 0^+$.
%\item $\sum_{i: \vv_i \geq p} \qq_i = \Q(p)$ for any $p \in [v_n, \infty)$, and $\sum_{i \in \NNP} \qq_i = \Q(v_n)$.
%\item $\spm\big(\{\tri(\vv_i, \qq_i)\}_{i \in \NNP}\big) \xrightarrow{\bf (a, b)} 2 + \Int{v_n}{\infty} (x - 1) \cdot \big(-\Q'(x)\big) \cdot e^{-\Q(x)} \cdot \dd x$.
%\end{enumerate}
%The lemma suggests that the new instance $\{\tri(\vv_i, \qq_i)\}_{i \in \NNP}$ gives a higher $\spm$ revenue than the given instance $\{\tri(v_i, q_i)\}_{i = 1}^n$. In other words, the worst case of the program happens when there is a continuum of agents.}
\subsection{Lower-Bound Analysis}
\label{subsec:spm_ap_lower}

In this part, we consider the triangular instance defined in the following \Cref{exp:spm_ap}, which gives a ratio of $\spm$ to $\ap$ arbitrarily close to constant $\C = 2 + \Int{1}{\infty} \big(1 -  e^{-\Q(x)}\big) \cdot \dd x \approx 2.6202$. Recall that a triangular distribution $\tri(v_i, q_i)$ has a CDF of $F_i(p) = \frac{(1 - q_i) \cdot p}{(1 - q_i) \cdot p + v_i q_i}$ when $p \in [0, v_i)$, and $F_i(p) = 1$ when $p \in [v_i, \infty)$.

\begin{example}
\label{exp:spm_ap}
Given any constant $\eps \in (0, 1)$ and integer $n \geq 2$, let $a \eqdef \min \big\{(1 + \eps), \Q^{-1}(\ln\eps^{-1})\big\}$, $b \eqdef (1 + \eps^{-1})$, $\delta \eqdef \frac{b - a}{n - 1}$, and $v_0 \eqdef \infty$. Define triangular instance $\{\tri(v_i, q_i)\}_{i = 1}^{n}$ as follows:
\begin{align*}
& v_i \eqdef b - (i - 1) \cdot \delta, &&
q_i \eqdef \frac{\R(v_i) - \R(v_{i - 1})}{v_i + \R(v_i) - \R(v_{i - 1})}, && \forall i \in [n].
\end{align*}
\end{example}

Actually, the above parameters $a$ and $b$ are carefully chosen to guarantee the next technical lemma (proved in \Cref{subapp:spm_ap_lower}), which will be useful for our lower-bound analysis.

\begin{lemma}
\label{lem:spm_ap_lower2}
Given any constant $\eps \in (0, 1)$, let $a = \min \big\{(1 + \eps), \Q^{-1}(\ln\eps^{-1})\big\}$ and $b = (1 + \eps^{-1})$. Then, $2 + \Int{a}{b} (x - 1) \cdot \big(-\Q'(x)\big) \cdot e^{-\Q(x)} \cdot \dd x \geq \C - 4 \cdot \eps$.
\end{lemma}

Notice that
\begin{enumerate*}[label = (\alph*), font = {\bfseries}]
\item $\ln(1 + x) \leq x$ when $x \geq 0$;
\item function $\R$ is a decreasing function; and
\item $\R(v_0) = \R(\infty) = 0$.
\end{enumerate*}
Because Program~\eqref{prog:spm_ap1} is derived from Program~\eqref{prog:spm_ap0} via {\em reductions}, it suffices to reason about Program~\eqref{prog:spm_ap1}. We first show that triangular instance $\{\tri(v_i, q_i)\}_{i = 1}^{n}$ satisfies constraint~\eqref{cstr:spm_ap1}:
\begin{align*}
\lhs \mbox{ of } \eqref{cstr:spm_ap1}
& = \mbox{$\sum_{i: v_i \geq p} \ln\big(1 + \frac{v_i q_i}{1 - q_i} \cdot \frac{1}{p}\big)$} \\
& \leq \mbox{$\sum_{i: v_i \geq p} \frac{v_i q_i}{1 - q_i} \cdot \frac{1}{p}$} && \mbox{\tt (by {\bf Claim~(a)} above)} \\
& = p^{-1} \cdot \R(\min_{i \in [n]} \{v_i \given v_i \geq p\}) - p^{-1} \cdot \R(0) && \mbox{\tt (by formulas of $\{q_i\}_{i = 1}^{n}$)} \\
& \leq p^{-1} \cdot \R(p)
= \rhs \mbox{ of } \eqref{cstr:spm_ap1}.  && \mbox{\tt (by {\bf Claims~(b,c)} above)}
\end{align*}
We next investigate the $\spm$ revenue extracted from triangular instance $\{\tri(v_i, q_i)\}_{i = 1}^n$, which can be summarized as the following lemma.

\begin{lemma}
\label{lem:spm_ap_lower1}
Consider the triangular instance $\{\tri(v_i, q_i)\}_{i = 1}^{n}$ defined in \Cref{exp:spm_ap}. It follows that $\spm \geq \C - 6 \cdot \eps$, for any sufficiently large integer $n \in \mathbb{N}_+$.
\end{lemma}

\begin{proof}
Note that
\begin{enumerate*}[label = (\alph*), font = {\bfseries}]
\item $\R(v_0) = \R(\infty) = 0$;
\item $a = v_{n} < v_{n - 1} < \cdots < v_1 = b$ is a uniform partition of interval $[a, b]$, with norm $\delta = \frac{b - a}{n - 1}$;
\item $\R'(p) = p \cdot \Q'(p)$ for any $p \in (1, \infty)$; and
\item $\ln(1 + x) \leq x$ when $x \geq 0$.
\end{enumerate*}
By the definition of Riemann integral, we have:
\begin{align*}
-\lim_{n \to \infty} \mbox{$\sum_{j = 1}^i \ln(1 - q_i)$}
& = \lim_{n \to \infty} \mbox{$\sum_{j = 1}^i \ln\big(1 + \frac{\R(v_j) - \R(v_{j - 1})}{v_j}\big)$}
&& \mbox{\tt (by formulas of $\{q_i\}_{i = 1}^{n}$)} \\
& = \mbox{$\ln\big(1 + \frac{\R(v_1)}{v_1}\big) + \Int{v_i}{v_1} \big(-\frac{\R'(x)}{x}\big) \cdot \dd x$}
&& \mbox{\tt (by {\bf Claims~(a, b)} above)} \\
& = \mbox{$\ln\big(1 + \frac{\R(b)}{b}\big) - \Q(b) + \Q(v_i)$}
&& \mbox{\tt (by {\bf Claim~(c)}; as $v_1 = b$)} \\
& \leq \mbox{$\frac{\R(b)}{b} - \Q(b) + \Q(v_i)$}
&& \mbox{\tt (by {\bf Claim~(d)} above)} \\
& \leq -\ln(1 - b^{-2}) + \Q(v_i),
&& \mbox{\tt (as $\R(b) = -b \cdot \ln(1 - b^{-2})$)}
\end{align*}
for each $i \in [n]$. By construction, we also have $\frac{v_i q_i}{1 - q_i} = \R(v_i) - \R(v_{i - 1})$, for each $i \in [n]$. Based on the objective function of Program~\eqref{prog:spm_ap1}, we obtain an $\spm$-revenue sequence, with the limit inferior of
\begin{align*}
\underset{n \to \infty}{\underline{\lim}} \spm\big(\{\tri(v_i, q_i)\}_{i = 1}^{n}\big)
& = 2 + \underset{n \to \infty}{\underline{\lim}} \mbox{$\sum_{i = 1}^{n} \frac{(v_i - 1) \cdot q_i}{1 - q_i} \cdot \prod_{j = 1}^{i} (1 - q_j)$} \\
& \hspace{1.955cm} \mbox{\tt (by construction and the arguments above)} \\
& \geq 2 + (1 - b^{-2}) \cdot \underset{n \to \infty}{\underline{\lim}} \mbox{$\sum_{i = 1}^{n} (1 - v_i^{-1}) \cdot \big(\R(v_i) - \R(v_{i - 1})\big) \cdot e^{-Q(v_j)}$} \\
& \hspace{1.955cm} \mbox{\tt (drop the summand with index $i = 1$)} \\
& \geq 2 + (1 - b^{-2}) \cdot \underset{n \to \infty}{\underline{\lim}} \mbox{$\sum_{i = 2}^{n} (1 - v_i^{-1}) \cdot \big(\R(v_i) - \R(v_{i - 1})\big) \cdot e^{-Q(v_j)}$} \\
& \hspace{1.955cm} \mbox{\tt (by {\bf Claims~(b,c)} above; as $v_1 = b$ and $v_n = a$)} \\
& \geq 2 + (1 - b^{-2}) \cdot \mbox{$\Int{a}{b} (x - 1) \cdot \big(-\Q'(x)\big) \cdot e^{-\Q(x)} \cdot \dd x$} \\
& \hspace{1.955cm} \mbox{\tt (by applying \Cref{lem:spm_ap_lower2})} \\
& \geq \C - (1 - b^{-2}) \cdot 4 \cdot \eps - (\C - 2) \cdot b^{-2} \\
& \hspace{1.955cm} \mbox{\tt (as $b = (1 + \eps^{-1})$ and $\C \approx 2.6202 < 3$)} \\
& \geq \C - 5 \cdot \eps
\end{align*}
Since $\eps \in (0, 1)$ is a given constant, $\spm\big(\{\tri(v_i, q_i)\}_{i = 1}^{n}\big) \geq \C - 6 \cdot \eps$, for any sufficiently large integer $n \in \NNP$. Hence, \Cref{lem:spm_ap_lower1} and the lower-bound part (in asymmetric regular setting) of \Cref{thm:ar_ap} are settled.
\end{proof}

\section{Anonymous Reserve vs. Anonymous Pricing}
\label{sec:ar_ap}
In this section, we study the revenue gap between $\ar$ and $\ap$. Sorting the buyers' values $\bids = \{b_i\}_{i = 1}^n$ such that $b_{(1)} \geq b_{(2)} \geq \cdots \geq b_{(n)}$, we begin with formulating CDF $\first$ of the highest value $b_{(1)}$ and CDF $\second$ of the second highest value $b_{(2)}$. For CDF $\first$, we have
\[
\first(p) = \Prob\big\{b_{(1)} \leq p\big\} = \Prob\big\{\forall i \in [n]:\,\, b_i \leq p\big\} = \prod_{i = 1}^n \Prob\big\{b_i \leq p\big\} = \prod_{i = 1}^n F_i(p).
\]
For CDF $\second$, the event that ``{\em the second highest value $b_{(2)}$ is at most $p$}'', i.e.\ $\big\{b_{(2)} \leq p\}$, can be partitioned into the following $(n + 1)$ disjoint sub-events: $A_i \eqdef \big\{(b_i > p) \wedge (\forall j \neq i:\,\, b_j \leq p)\big\}$ for each $i \in [n]$, and $A_0 \eqdef \big\{\forall i \in [n]:\,\, b_i \leq p\big\}$. Therefore,
\[
\second(p) = \Prob\{A_0\}+ \sum_{i = 1}^n \Prob\{A_i\} = \first(p) +  \sum_{i = 1}^n \big(1 - F_i(p)\big) \cdot \prod_{j \neq i} F_j(p) = \first(p) \cdot \bigg[1 + \sum_{i = 1}^n \Big(\frac{1}{F_i(p)} - 1\Big)\bigg].
\]
By definition, we have that $\ap(p) = p \cdot \big(1 - \first(p)\big)$. In addition, an explicit formula of the $\ar$ revenue is established in the following lemma \citep[first introduced by][]{CGM15}. For the sake of completeness, we provide a proof here. In \Cref{sec:opt_ar}, this lemma will also be useful for the lower-bound analysis of the $\opt$ vs.\ $\ar$ problem.

\begin{lemma}[\citet{CGM15}]
\label{lem:ar_rev}
For any reserve price $p \in \RRP$, the {\sf Anonymous Reserve} revenue equals $\ar(p) = p \cdot \big(1 - \first(p)\big) + \Int{p}{\infty} \big(1 - \second(x)\big) \cdot \dd x$.
\end{lemma}

\begin{proof}
Given any reserve price $p \in \RRP$, there are three possible outcomes:
\begin{enumerate*}[label = (\alph*), font = {\bfseries}]
\item a revenue of $0$ when no value reaches the reserve price of $p$, i.e.\ event $\big\{b_{(1)} < p\big\}$;
\item a revenue of $p$ when exactly one value reaches the reserve price of $p$, i.e.\ event $\big\{b_{(1)} \geq p > b_{(2)}\big\}$; and
\item a revenue of $b_{(2)}$ when two or more values reach the reserve price of $p$, i.e.\ event $\big\{b_{(2)} \geq p\big\}$.
\end{enumerate*}
Given these, we can formulate the $\ar$ revenue as follows:
\begin{align*}
\ar(p)
& = \E_{\bids \sim \distrs} \Big[p \cdot \mathbbm{1}\big\{b_{(1)} \geq p > b_{(2)}\big\} + b_{(2)} \cdot \mathbbm{1}\big\{b_{(2)} \geq p\big\}\Big] \\
& = \E_{\bids \sim \distrs} \Big[p \cdot \mathbbm{1}\big\{b_{(1)} \geq p\big\} + \big(b_{(2)} - p\big) \cdot \mathbbm{1}\big\{b_{(2)} \geq p\big\}\Big] \\
& = \mbox{$p \cdot \big(1 - \first(p)\big) + \Int{p}{\infty} (x - p) \cdot \dd \second(x)$} \\
& = \mbox{$p \cdot \big(1 - \first(p)\big) + \Int{p}{\infty} \big(1 - \second(x)\big) \cdot \dd x$}. && \mbox{\tt (integration by parts)}
\end{align*}
This completes the proof of \Cref{lem:ar_rev}.
\end{proof}

Given the above revenue formulas, the revenue gap between $\ar$ and $\ap$ can be captured by the next mathematical program. By finding the optimal solution, we attain the next theorem.

\vspace{.1in}
\noindent\fcolorbox{black}{lightgray!50}{\begin{minipage}{0.977\textwidth}
\vspace{-.1in}
\begin{align}
\label{prog:ap_ar}\tag{P4}
& \max_{\{F_i\}_{i = 1}^n, p \in \RRP} && \mbox{$\ar(p) = p \cdot \big(1 - \first(p)\big) + \Int{p}{\infty} \big(1 - \second(x)\big) \cdot \dd x$} \\
\label{cstr:ap_ar}\tag{C4}
& \mbox{subject to:} && \ap(x) = x \cdot \big(1 - \first(x)\big) \leq 1, \hspace{2cm} \forall x \in \RRP
\end{align}
\end{minipage}}

\begin{theorem}
\label{thm:ar_ap}
The supremum ratio of $\ar$ to $\ap$ is equal to $\frac{\pi^2}{6} \approx 1.6449$, which holds in each of
\begin{enumerate*}[label = (\alph*), font = {\bfseries}]
\item asymmetric general setting,
\item asymmetric regular setting, and
\item i.i.d.\ general setting.
\end{enumerate*}
\end{theorem}

\noindent
{\bf Proof Overview.}
In \Cref{subsec:ar_ap:upper}, we settle the upper-bound part of \Cref{thm:ar_ap} in asymmetric general setting, which implies the same upper bound in the other two settings. Then, we respectively construct a matching i.i.d.\ general instance in \Cref{subapp:ar_ap_lower_irregular}, and a matching asymmetric regular instance in \Cref{subapp:ar_ap_lower_regular}. As a whole, \Cref{thm:ar_ap} is accomplished.

Intuitively, we obtain \Cref{thm:ar_ap} as follows. Recall Program~\eqref{prog:spm_ap0} in \Cref{sec:spm_ap} for the $\spm$ vs.\ $\ap$ problem. The worst-case instance of Program~\eqref{prog:spm_ap0} actually makes constraint~\eqref{cstr:spm_ap0} tight everywhere, and thus in some sense ``dominates'' any other feasible instance. We simply ``guess'' that Program~\eqref{prog:ap_ar} possesses the same characterization, which turns out to be workable.

\subsection{Upper-Bound Analysis}
\label{subsec:ar_ap:upper}

Due to constraint~\eqref{cstr:ap_ar}, $\first(p) \geq \Psi_1(p) \eqdef (1 - p^{-1})_+$, for all $p \in \RRP$. Namely, the highest-value distribution $\first$ is stochastically dominated by distribution $\Psi_1$. Moreover,
\[
\begin{aligned}
\second(p)
& = \mbox{$\first(p) \cdot \Big[1 + \sum_{i = 1}^n \big(\frac{1}{F_i(p)} - 1\big)\Big]$} \\
& \geq \mbox{$\first(p) \cdot \Big[1 + \sum_{i = 1}^n \ln\big(\frac{1}{F_i(p)}\big)\Big]$}
&& \mbox{\tt (as $x \geq \ln(1 + x)$ when $x \geq 0$)} \\
& = \first(p) \cdot \big(1 - \ln \first(p)\big),
\end{aligned}
\]
It can be seen that function $d(x) \eqdef x \cdot (1 - \ln x)$ is an increasing function on interval $x \in (0, 1]$, and that $\lim\limits_{x \to 0^+} d(x) = 0$. Hence, $\second(p) \geq d\big(\Psi_1(p)\big) = \Psi_2(p)$ for all $p \in \RRP$, where $\Psi_2(p) \eqdef 0$ when $p \in [0, 1]$, and $\Psi_2(p) \eqdef (1 - p^{-1}) \cdot \big[1 - \ln(1 - p^{-1})\big]$ when $p \in (1, \infty)$.

For any reserve price $p \in \RRP$, the above arguments imply that
\begin{equation}
\label{eq:ar:upper}
\ar(p) \leq \mbox{$p \cdot \big(1 - \Psi_1(p)\big) + \Int{p}{\infty} \big(1 - \Psi_2(x)\big) \cdot \dd x$}.
\end{equation}
{\tt [When $p \leq 1$]:}
$\rhs \mbox{ of } \eqref{eq:ar:upper} = p + \Int{p}{1} 1 \cdot \dd x + \Int{1}{\infty} \big(1 - \Psi_2(x)\big) \cdot \dd x = 1 + \Int{1}{\infty} \big(1 - \Psi_2(x)\big) \cdot \dd x$.

\noindent
{\tt [When $p > 1$]:}
$\rhs \mbox{ of } \eqref{eq:ar:upper} = 1 + \Int{p}{\infty} \big(1 - \Psi_2(x)\big) \cdot \dd x < 1 + \Int{1}{\infty} \big(1 - \Psi_2(x)\big) \cdot \dd x$.

\noindent
Combining both cases together results in the upper-bound part of \Cref{thm:ar_ap}:
\begin{align}
\notag
\ar = \max_{p \in \RRP} \big\{\ar(p)\big\}
& \leq \mbox{$1 + \Int{1}{\infty} \big(1 - \Psi_2(x)\big) \cdot \dd x$} \\
\notag
& = \mbox{$1 + \Int{1}{\infty} \big[\frac{1}{x} - (1 - \frac{1}{x}) \cdot \sum_{k = 1}^{\infty} \frac{1}{k \cdot x^k}\big] \cdot \dd x$}
&& \mbox{\tt (as $\ln(1 - z) = \sum_{k = 1}^{\infty} \frac{z^k}{k}$)} \\
\notag
& = \mbox{$1 + \sum_{k = 1}^{\infty} \frac{1}{k \cdot (k + 1)} \cdot \Int{1}{\infty} \frac{\dd x}{x^{k + 1}}$} \\
\label{eq:ar_ap}
& = \mbox{$1 + \sum_{k = 1}^{\infty} \frac{1}{k^2 \cdot (k + 1)} = \sum_{k = 1}^{\infty} \frac{1}{k^2} = \frac{\pi^2}{6}$}.
\end{align}

\subsection{Lower-Bound Analysis in I.I.D.\ General Setting}
\label{subapp:ar_ap_lower_irregular}

We begin our lower-bound analysis of \Cref{thm:ar_ap} with a matching i.i.d.\ general instance, namely all of the buyers follow the same (possibly) irregular value distribution. For the instance $\{F_n\}^n$ defined in the following \Cref{exp:iid_irregular}, we will prove that:
\begin{enumerate*}[label = (\alph*), font = {\bfseries}]
\item it is feasible to Program~\eqref{prog:ap_ar}, for any $n \in \NNP$;
\item given any constant $\eps \in (0, 1)$, an $\ar$ revenue of at least $\big(\frac{\pi^2}{6} - \eps\big)$ can be extracted from it, whenever integer $n \in \NNP$ is sufficiently large.
\end{enumerate*}
For simplicity, we reuse notations defined in \Cref{subsec:ar_ap:upper}.

\begin{example}
\label{exp:iid_irregular}
There are $n \in \mathbb{N}_+$ buyers drawing i.i.d.\ values $\bids = \{b_i\}_{i = 1}^n$ from a common distribution $F_n(p) \eqdef \sqrt[n]{\Psi_1(p)}$, i.e.\ $F_n(p) = 0$ when $p \in [0, 1]$, and $F_n(p) = (1 - p^{-1})^{\frac{1}{n}}$ when $p \in (1, \infty)$.
\end{example}

First, $\ap(p) = p \cdot \Big[1 - \big(F_n(p)\big)^n\Big] = p \cdot \big(1 - \Psi_1(p)\big) \leq 1$ for all $p \in \RRP$, and thus instance $\{F_n\}^n$ is feasible to Program~\eqref{prog:ap_ar}. For CDF $\second$ of the second highest value,
\begin{align*}
\lim_{n \to \infty} \second(p)
& = \lim_{n \to \infty} \mbox{$\big(F_n(p)\big)^n \cdot \Big[1 + n \cdot \big(\frac{1}{F_n(p)} - 1\big)\Big]$} \\
& = \mbox{$\big(1 - \frac{1}{p}\big) \cdot \Big\{1 + \lim\limits_{n \to \infty} \frac{\exp[-\ln(1 - \frac{1}{p}) \cdot n^{-1}] - 1}{n^{-1}}\Big\}$}
= \Psi_2(p),
\end{align*}
for all $p \in (1, \infty)$. By choosing a fixed reserve price of $1$, we obtain a convergent $\ar$-revenue sequence, with the limitation of
\[
\lim_{n \to \infty} \ar\big(1, \{F_n\}^n\big) = \mbox{$1 + \Int{1}{\infty} \big(1 - \Psi_2(x)\big) \cdot \dd x$} \overset{\eqref{eq:ar_ap}}{=} \pi^2 / 6.
\]
In that $\eps \in (0, 1)$ is a given constant, $\ar\big(1, \{F_n\}^n\big) \geq \frac{\pi^2}{6} - \eps$ whenever integer $n \in \NNP$ is sufficiently large. This settles the lower-bound part of \Cref{thm:ar_ap} in i.i.d.\ irregular setting. (We will prove in \Cref{subapp:ar_ap_lower_formula} that distribution $F_n$ is indeed irregular, whenever $n \geq 2$.)

\begin{remark}
Actually, because $\ap\big(p, \{F_n\}^n\big) = 1$ for all $p \in (1, \infty)$, distribution $F_n$ stochastically dominates any other feasible distribution of Program~\eqref{prog:ap_ar}. This implies that instance $\{F_n\}^n$ is a worst-case instance of Program~\eqref{prog:ap_ar}, for any specific positive integer $n \in \NNP$. Below, we provide the tight ratios corresponding to some small positive integer $n \in \NNP$.

\vspace{.1in}
\begin{tabular}{|c|c|c|c|c|c|}
	\hline
	$n$ & $2$ & $3$ & $4$ & $\cdots$ & $\infty$ \\
	\hline
	\rule{0pt}{14pt}{\rm ratio} & $2\ln2 \approx 1.3863$ & $3\ln3 - \frac{\pi}{\sqrt{3}} \approx 1.4820$ & $9\ln2 - \frac{3\pi}{2} \approx 1.5259$ & $\cdots$ & $\frac{\pi^2}{6} \approx 1.6449$ \\ [4pt]
	\hline
\end{tabular}
\end{remark}

\subsection{Lower-Bound Analysis in Asymmetric Regular Setting}
\label{subapp:ar_ap_lower_regular}

We next use triangular distributions to construct an asymmetric regular instance that matches the bound in \Cref{thm:ar_ap} as well. Actually, the idea behind this instance is very similar to that behind \Cref{exp:spm_ap} (i.e.\ the lower-bound instance of the $\spm$ vs.\ $\ap$ problem in \Cref{sec:spm_ap}). For convenience, we reuse notations defined in \Cref{subsec:ar_ap:upper}. Recall that a triangular distribution $\tri(v_i, q_i)$ has a CDF of $F_i(p) = \frac{(1 - q_i) \cdot p}{(1 - q_i) \cdot p + v_i q_i}$ when $p \in [0, v_i)$, and $F_i(p) = 1$ when $p \in [v_i, \infty)$.

\begin{example}
\label{exp:ar_ap_non_iid}
Given any constant $\eps \in (0, 1)$ and integer $n \in \NNP$, let $a \eqdef (1 + \eps)$, $b \eqdef (1 + \eps^{-1})$, and $\delta \eqdef \frac{b - a}{n}$. Based on function $\V(p) \eqdef p \cdot \ln\big(\frac{p}{p - 1}\big)$, consider triangular instance $\{\tri(v_i, q_i)\}_{i = 1}^{2n}$:
\begin{align*}
& v_i \eqdef b, && q_i \eqdef \frac{\frac{1}{n} \cdot \V(v_i)}{v_i + \frac{1}{n} \cdot \V(v_i)}, && \forall i \in [n]; \\
& v_i \eqdef b - (i - n) \cdot \delta, && q_i \eqdef \frac{\V(v_i) - \V(v_{i - 1})}{v_i + \V(v_i) - \V(v_{i - 1})}, && \forall i \in [n + 1: 2n].
\end{align*}
\end{example}

For ease of notation, let $v_0 \eqdef \infty$. It can be checked that
\begin{enumerate*}[label = (\alph*), font = {\bfseries}]
\item $\ln(1 + x) \leq x$ when $x \geq 0$; and
\item function $\V$ is a decreasing function.
\end{enumerate*}
To justify the feasibility that $\ap\big(p, \{\tri(v_i, q_i)\}_{i = 1}^{2n}\big) = p \cdot \big(1 - \prod_{i = 1}^{2n} F_i(p)\big) \leq 1$, or equivalently, $\sum_{i = 1}^{2n} \ln F_i(p) \ge \ln\big(1 - \frac{1}{p}\big)$, we observe that
\begin{align*}
\mbox{$\sum_{i = 1}^{2n} \ln F_i(p)$}
& = -\mbox{$\sum_{i: v_i \geq p} \ln\big(1 + \frac{v_i q_i}{1 - q_i} \cdot \frac{1}{p}\big)$} \\
& \geq -\mbox{$\sum_{i: v_i \geq p} \frac{v_i q_i}{1 - q_i} \cdot \frac{1}{p}$}
&& \mbox{\tt (by {\bf Claim~(a)}  above)} \\
& = -\mbox{$\frac{1}{p} \cdot \V\big(\min_{i \in [2n]} \big\{v_i \given v_i \geq p\big\}\big)$}
&& \mbox{\tt (by formulas of $\{q_i\}_{i = 1}^{2n}$)} \\
& \geq -\mbox{$\frac{1}{p} \cdot \V(p) = \ln\big(1 - \frac{1}{p}\big)$}.
&& \mbox{\tt (by {\bf Claim~(b)} above)}
\end{align*}
Let us measure the $\ar$ revenue from the above triangular instance $\{\tri(v_i, q_i)\}_{i = 1}^{2n}$.

\begin{lemma}
\label{lem:ar_ap_lower}
Consider the triangular instance $\{\tri(v_i, q_i)\}_{i = 1}^{2n}$ defined in \Cref{exp:ar_ap_non_iid}. It follows that $\ar \geq \ar(a) \geq \frac{\pi^2}{6} - 3 \cdot \eps$, for any sufficiently large integer $n \in \mathbb{N}_+$.
\end{lemma}

\begin{proof}
We first prove that $\lim\limits_{n \to \infty} \first(p) = \Psi_1(p)$ for any $p \in [a, b]$.
Observe that
\begin{enumerate*}[label = (\alph*), font = {\bfseries}]
\item when $n \to \infty$, $q_i \to 0^+$ for all $i \in [2n]$; and
\item $a = v_{2n} < v_{2n - 1} < \cdots < v_{n + 1} < v_n = b$ is a uniform partition of interval $[a, b]$, with norm $\delta = \frac{b - a}{n}$.
\end{enumerate*}
As a result, for any $p \in [a, b]$,
\begin{align*}
\lim_{n \to \infty} \first(p)
& = \mbox{$\exp\Big[-\lim_{n \to \infty} \sum_{i: v_i \geq p} \ln\big(1 + \frac{v_i q_i}{1 - q_i} \cdot \frac{1}{p}\big)\Big]$} \\
& = \mbox{$\exp\Big(-\lim_{n \to \infty} \sum_{i: v_i \geq p} \frac{v_i q_i}{1 - q_i} \cdot \frac{1}{p}\Big)$}
&& \mbox{\tt (by {\bf Claim~(a)} that $q_i \to 0^+$)} \\
& = \exp\big(-p^{-1} \cdot \V(\min_{i \in [2n]} \{v_i \given v_i \geq p\})\big) && \mbox{\tt (by formulas of $\{q_i\}_{i = 1}^{2n}$)} \\
& = \exp\big(-p^{-1} \cdot \V(p)\big)
= 1 - p^{-1}. && \mbox{\tt (by {\bf Claim~(b)} that norm $\delta \to 0^+$)}
\end{align*}
Similarly, it can be checked that $\lim\limits_{n \to \infty} \second(p) = \Psi_2(p)$ for any $p \in [a, b]$. Because $\eps \in (0, 1)$ is a given constant, by choosing a fixed reserve price of $a = (1 + \eps)$ and a sufficiently large integer $n \in \NNP$, we capture an $\ar$ revenue of
\begin{align*}
\ar(a)
& \geq \mbox{$a \cdot \big(1 - \Psi_1(a)\big) + \Int{a}{b} \big(1 - \Psi_2(x)\big) \cdot \dd x - \eps$} \\
& \overset{\eqref{eq:ar_ap}}{=} \mbox{$\pi^2 / 6 - \big(\Int{1}{a} + \Int{b}{\infty}\big) \big(1 - \Psi_2(x)\big) \cdot \dd x - \eps$}.
\end{align*}
It can be checked that $\ln\big(1 - \frac{1}{p}\big) \leq -\frac{1}{p}$, and thus $\Psi_2(p) = \big(1 - \frac{1}{p}\big) \cdot \Big[1 - \ln\big(1 - \frac{1}{p}\big)\Big] \geq 1 - \frac{1}{p^2}$. Recall that $a = (1 + \eps)$ and $b = (1 + \eps^{-1})$, then
\begin{align*}
\mbox{$\big(\Int{1}{a} + \Int{b}{\infty}\big) \big(1 - \Psi_2(x)\big) \cdot \dd x$}
\leq \mbox{$\big(\Int{1}{a} + \Int{b}{\infty}\big) \frac{\dd x}{x^2}$}
= \mbox{$1 - \frac{1}{a} + \frac{1}{b}$}
= \mbox{$\frac{2 \cdot \eps}{1 + \eps}$}
\leq \mbox{$2 \cdot \eps$}.
\end{align*}
Combining the above two inequalities together completes the proof of \Cref{lem:ar_ap_lower}. This settles the lower-bound part of \Cref{thm:ar_ap} in asymmetric regular setting.
\end{proof}

\section{Myerson Auction vs. Anonymous Reserve}
\label{sec:opt_ar}
For the $\opt$ vs.\ $\ar$ problem in asymmetric regular setting, \cite{HR09} conjectured that the following two-buyer instance (i.e.\ \Cref{exp:opt_ar_2}) is the worst case. This instance gives a ratio of $2$, which remained the best known lower bound for a decade. Nonetheless, we will employ {\em triangular} instances to establish improved lower bounds.

\begin{example}
\label{exp:opt_ar_2}
Consider a two-buyer instance:
\begin{enumerate*}[label = (\alph*), font = {\bfseries}]
\item one buyer draws his value from the equal-revenue distribution $F(p) = 1 - \frac{1}{p}$; and
\item the other buyer has a deterministic value of $1$.
\end{enumerate*}

It can be checked that $\ar = \ap = 1$, whereas $\opt = \spm = 2$ (e.g. by sequentially posting a price of $\infty$ to the first buyer, and then posting a price of $1$ to the second buyer).
\end{example}

The following lemma (see \Cref{app:opt_ar} for its proof) will be helpful in reasoning about our sharper lower-bound instances. For this, we provide some intuitions in \Cref{rmk:ar_monotone}, based on the {\em revenue-equivalence theorem} of \cite{M81}.

\begin{lemma}
\label{lem:ar_rev_monotone}
Given any triangular instance $\{\tri(v_i, q_i)\}_{i = 1}^n$ that $v_1 \geq v_2 \cdots \geq v_n > v_{n + 1} \eqdef 0$, the best $\ar$ revenue is
achieved by reserve price $p = v_i$, for some $i \in [n]$.
\end{lemma}

\begin{remark}
\label{rmk:ar_monotone}
Any triangular distribution $\tri(v_i, q_i)$ is supported on interval $p \in [0, v_i]$, and only the maximum possible value of $b_i = v_i$ corresponds to a non-negative virtual value. Given any $k \in [n]$, when the seller raises the reserve price in $\ar$ mechanism from some $p \in (v_{k + 1}, v_k)$ to $p = v_k$, the item no longer gets allocated when the highest value is between $(v_{k + 1}, v_k)$. As a result, the seller rules out some possible outcomes in $\ar$ mechanism that result in negative virtual welfare. By the revenue-equivalence theorem, the new reserve price of $p = v_i$ leads to a larger $\ar$ revenue.
\end{remark}

\begin{theorem}
\label{thm:opt_ar_lower}
For the $\opt$ vs.\ $\ar$ problem, there is a three-buyer instance $\{\tri(v_i, q_i)\}_{i = 0}^2$ giving a ratio of $2.1361$, and a four-buyer triangular instance $\{\tri(\overline{v}_i, \overline{q}_i)\}_{i = 0}^3$ giving a ratio of $2.1596$.
\end{theorem}

\begin{proof}
Our three-buyer instance contains three triangular distributions $\{\tri(v_i, q_i)\}_{i = 0}^2$, where the monopoly prices $v_0 > v_1 > v_2$. Recall that a triangular distribution $\tri(v_i, q_i)$ has a CDF of $F_i(p) = \frac{(1 - q_i) \cdot p}{(1 - q_i) \cdot p + v_i q_i}$ when $p \in [0, v_i)$, and $F_i(p) = 1$ when $p \in [v_i, \infty)$. We reuse notations $\first$ and $\second$ to respectively denote the CDF's of the highest and the second highest values.

\begin{itemize}
\item Let distribution $\tri(v_0, q_0)$ be the special distribution $\tri(\infty)$ with CDF $F_0(p) = \frac{p}{p + 1}$.

    Thus, $\ar(p) = p \cdot \big(1 - F_0(p)\big) = \frac{p}{p + 1}$, for all $p \in (v_1, \infty)$; particularly, $\ar(v_0) = \ar(\infty) = 1$.
\item Let $v_1 \in (1, \infty)$ be a variable to be determined, and let $q_1 \eqdef 1 / v_1^2$.

    Because this buyer always has a value no more than $v_1$, the $\ar$ revenue under reserve price $p = v_1$ is equal to the $\ap$ revenue under posted price $p = v_1$. Because $q_1 = 1 / v_1^2$, we conclude that $\ar(v_1) = \ap(v_1) = v_1 \cdot \big[1 - \frac{v_1}{v_1 + 1} \cdot \frac{(1 - q_1) \cdot v_1}{(1 - q_1) \cdot v_1 + v_1 q_1}\big] = 1$.
\item Let $v_2$ be the root of equation $v_2 + \frac{v_1}{1 + v_1 - v_1^2} \cdot \ln\big[\frac{1 + v_1}{1 + v_2} \cdot \frac{v_2(v_1^2 - 1) + v_1}{v_1^3}\big] = 1$, and let $q_2 \eqdef 1$.

    This buyer has a deterministic value of $v_2$, i.e.\ $\ap(v_2) = v_2$. For any $x \in (v_2, v_1]$,
    \[
    1 - \second(x) = \big(1 - F_0(x)\big) \cdot \big(1 - F_1(x)\big) = \frac{v_1 q_1}{(x + 1) \cdot \big(x + \frac{v_1 q_1}{1 - q_1}\big)},
    \]
    Given these, and due to the definitions of $q_1$ and $v_2$, we have
    \begin{align*}
    \ar(v_2)
    & = \mbox{$\ap(v_2) + \Int{v_2}{v_1} \big(1 - \second(x)\big) \cdot \dd x$} \\
    & = v_2 +  \frac{v_1}{1 + v_1 - v_1^2} \cdot \Big.\ln\Big[\frac{x + 1}{(v_1^2 - 1) \cdot x + v_1}\Big]\Big|_{v_2}^{v_1} = 1.
    \end{align*}
\end{itemize}
Combining everything together, we conclude from \Cref{lem:ar_rev_monotone} that $\ar = 1$. On the other hand,
\begin{align*}
\opt
& = 1 + v_1 q_1 + v_2 q_2 \cdot (1 - q_1)
&& \mbox{\tt (by \Cref{lem:spm_tri_best})} \\
& = 1 + 1 / v_1 + v_2 \cdot (1 - 1 / v_1^2).
&& \mbox{\tt (as $q_1 = 1 / v_1^2$)}
\end{align*}
As \Cref{fig:opt_ar} shows, choosing $v_1 \approx 1.5699$ results in $v_2 \approx 0.8399$ and $\opt \approx 2.1361$.

\begin{figure}[H]
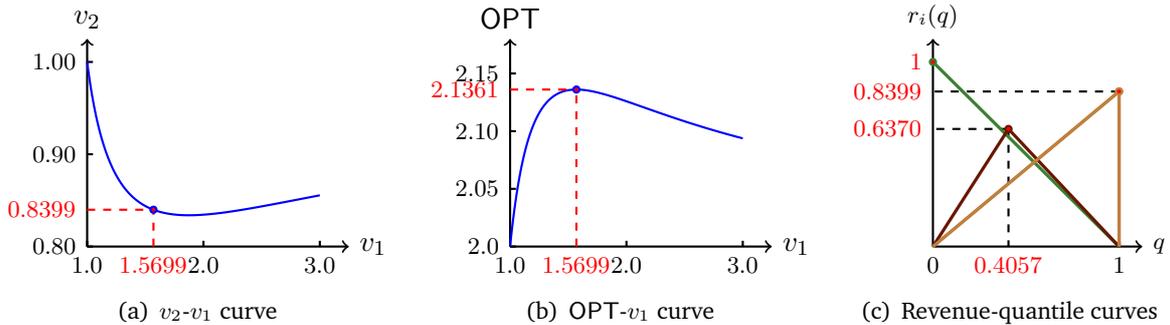

\centering
\subfigure[$v_2$-$v_1$ curve]{

}
\caption{Demonstration for the three-buyer instance $\{\tri(v_i, q_i)\}_{i = 0}^2$ defined in \Cref{thm:opt_ar_lower}.}
\label{fig:opt_ar}
\end{figure}

Similarly, the four-buyer instance is constructed as follows. The buyers' value distributions are triangular distributions whose parameters are specified in the next table. Based on \Cref{lem:spm_tri_best,lem:ar_rev_monotone} and an analysis similar to the three-buyer instance, we numerically calculate the revenue of $\ar$ and $\opt$. Details are omitted.

\vspace{.1in}
{\centering
	\begin{tabular}{|l|l|}
		\hline
		\rule{0pt}{12pt}$\tri(\overline{v}_0, \overline{q}_0) = \tri(\infty)$ & $\tri(\overline{v}_1, \overline{q}_1) = \tri(1.8512, 0.2918)$ \\ [1pt]
		\hline
		\rule{0pt}{12pt}$\tri(\overline{v}_2, \overline{q}_2) = \tri(0.9700, 0.6138)$ & $\tri(\overline{v}_3, \overline{q}_3) = \tri(0.7231, 1.0000)$ \\ [1pt]
		\hline
	\end{tabular}
\par}

\begin{itemize}
	\item $\ar(\overline{v}_i) \approx 1.0000$ for each $i \in \{0, 1, 2, 3\}$, and thus $\ar \approx 1.0000$.
	\item $\opt = 1 + \sum_{i = 1}^3 \overline{v}_i \overline{q}_i \cdot \prod_{j = 1}^{i - 1} (1 - \overline{q}_j) \approx 2.1596$.
\end{itemize}
This completes the proof of \Cref{thm:opt_ar_lower}.
\end{proof}

\begin{remark}
The reader may ask that why triangular instances are ``bad cases'' of the {\sf Myerson Auction} vs.\ {\sf Anonymous Reserve} problem. Actually, the following reasoning guides us to obtain the improved lower-bound instances involved in \Cref{thm:opt_ar_lower}.

As mentioned, in comparison with {\sf Anonymous Pricing}:
\begin{enumerate*}[label = (\alph*), font = {\bfseries}]
\item {\sf Sequential Posted-Pricing} leverages the price discrimination;
\item {\sf Anonymous Reserve} employs the buyer competition; and
\item {\sf Myerson Auction} benefits from the both advantages.
\end{enumerate*}

In addition, for any triangular instance:
\begin{enumerate*}[label = (\alph*), font = {\bfseries}]
\item \Cref{lem:spm_tri_best} suggests that {\sf Sequential Posted-Pricing} gives exactly the same revenue as {\sf Myerson Auction}; that is,
\item the price discrimination plays a dominant role in revenue maximization, whereas the buyer competition is negligible;
\item given this, we conjecture that the worst case of the {\sf Myerson Auction} vs.\ {\sf Anonymous Reserve} problem is reached by a triangular instance.
\end{enumerate*}
\end{remark}

\section{Conclusions}
\label{app:summary}
This work studies revenue gaps among $\opt$, $\spm$, $\ar$, and $\ap$. In the literature, there is another simple mechanism receiving particular attention: {\sf Order-Oblivious Posted-Pricing} ($\opm$) proposed by \cite{CHMS10}. Roughly speaking, $\opm$ is the worst-order (a.k.a.\ adversary) counterpart of $\spm$, and models the fact that the seller cannot control the order of buyers in some markets. In particular, this mechanism and the corresponding revenue are given by
\[
\opm(\distrs) \eqdef \min_{\sigma \in \Pi} \max_{\prices \in \RRP^n} \big\{\spm(\sigma, \prices, \distrs)\big\}.
\]
The next three tables summarize current progress in this research agenda. For each comparison, an interval indicates both of lower bound and upper bound, while a number means that this bound is tight. Recall the lattice structure in \Cref{fig:result}: $\spm$ and $\opm$ are identical for i.i.d.\ distributions, and are both incomparable to $\ar$.

Notice that some tight ratios in \Cref{tbl:summary1} and \Cref{tbl:summary2} are still unknown. It is an interesting direction to close these gaps, which are beneficial to our understandings of relative powers and distinctions of these mechanisms.

\begin{table}[H]
\centering
{\small
\begin{tabular}{|c|c|c|c|}
	\hline
	\rule{0pt}{12pt}Setting & Comparison & Ratio & Reference \\ [1pt]
	\hline
	\hline
	\rule{0pt}{12pt}I.I.D.\ & $\opt$ vs.\ $\spm^\diamond$ & $\alpha$ & \cite{K86,CFHOV17} \\ [1pt]
	\hline
	\hline
	\rule{0pt}{12pt}& $\opt$ vs.\ $\spm$ & $[\alpha, \beta]$ & \cite{K86,CSZ19} \\ [1pt]
	\cline{2-4}
	\rule{0pt}{12pt}Asym. & $\opt$ vs.\ $\opm$ & \multirow{2}*{\rule{0pt}{17pt}$2$} & Folklore, e.g.\ see \citet[Chapter~4.2.2]{H13} \\ [1pt]
	\cline{2-2}
	\rule{0pt}{12pt}& $\spm$ vs.\ $\opm$ & & \citet[Example~5.2]{HR09} \\ [1pt]
	\hline
	\multicolumn{4}{l}{\rule{0pt}{12pt}$(\diamond)$: Note that $\spm$ and $\opm$ are equivalent in the i.i.d.\ setting.}
\end{tabular}}
\caption{Summary I of revenue gaps; each comparison here admits the same ratio for general distributions as for regular distributions, due to the {\em ironing} technique \citep[see][]{M81}; for definitions of constants $\alpha \approx 1.34$ and $\beta \approx 1.49$, recall \Cref{footnote:alpha,footnote:beta}, respectively.}
\label{tbl:summary1}
\end{table}

\begin{table}[H]
\centering
{\small\begin{tabular}{|c|c|c|c|}
\hline
\rule{0pt}{12pt}Setting & Comparison & Ratio & Reference \\ [1pt]
\hline
\hline
\rule{0pt}{12pt}\multirow{2}*{I.I.D.\ Regular} & $\opt$ vs.\ $\ap$ & \multirow{2}*{\rule{0pt}{13pt}$e / (e - 1)$} & \cite{CHMS10} \\ [1pt]
\cline{2-2}
\rule{0pt}{12pt}& $\spm^\diamond$ vs.\ $\ap$ &  & \cite{DFK16} \\ [1pt]
\hline
\hline
\rule{0pt}{12pt}& $\opt$ vs.\ $\ar$ & $[2.1596, \C]$ & \Cref{thm:opt_ar_lower}, \cite{AHNPY15} \\ [1pt]
\cline{2-4}
\rule{0pt}{12pt}& $\opt$ vs.\ $\ap$ & $\C$ & \cite{AHNPY15,JLQTX2019} \\ [1pt]
\cline{2-4}
\rule{0pt}{12pt}Asym. Regular & $\spm$ vs.\ $\ap$ & $\C$ & \Cref{thm:spm_ap} \\ [1pt]
\cline{2-4}
\rule{0pt}{12pt}& $\opm$ vs.\ $\ap$ & $[e / (e - 1), \C]$ & \Cref{thm:spm_ap}, \cite{DFK16} \\ [1pt]
\cline{2-4}
\rule{0pt}{12pt}& $\ar$ vs.\ $\ap$ & $\pi^2 / 6$ & \Cref{thm:ar_ap} \\ [1pt]
\hline
\multicolumn{4}{l}{\rule{0pt}{12pt}$(\diamond)$: Note that $\spm$ and $\opm$ are equivalent in the i.i.d.\ setting.}
\end{tabular}}
\caption{Summary II of revenue gaps.}
\label{tbl:summary2}
\end{table}

\begin{table}[H]
\centering
{\small\begin{tabular}{|c|c|c|c|}
\hline
\rule{0pt}{12pt}Setting & Comparison & Ratio & Reference \\
\hline
\hline
\rule{0pt}{12pt}& $\opt$ vs.\ $\ar$ & & \cite{CHMS10} \\ [1pt]
\cline{2-2}
\rule{0pt}{12pt}\multirow{2}*{\rule{0pt}{14pt}I.I.D.\ General} & $\opt$ vs.\ $\ap$ & $2$ & \cite{H13} \\ [1pt]
\cline{2-2}
\rule{0pt}{12pt}& $\spm^\diamond$ vs.\ $\ap$ & & \cite{DFK16} \\ [1pt]
\cline{2-4}
\rule{0pt}{12pt}& $\ar$ vs.\ $\ap$ & $\pi^2 / 6$ & \Cref{thm:ar_ap} \\ [1pt]
\hline
\hline
\rule{0pt}{12pt}& $\opt$ vs.\ $\ar$ & & \\ [1pt]
\cline{2-2}
\rule{0pt}{12pt}& $\opt$ vs.\ $\ap$ & \multirow{2}*{\rule{0pt}{14pt}$n$} & \multirow{2}*{\cite{AHNPY15}} \\ [1pt]
\cline{2-2}
\rule{0pt}{12pt}Asym. General & $\spm$ vs.\ $\ap$ & & \\ [1pt]
\cline{2-2}
\rule{0pt}{12pt}& $\opm$ vs.\ $\ap$ & & \\ [1pt]
\cline{2-4}
\rule{0pt}{12pt}& $\ar$ vs.\ $\ap$ & $\pi^2 / 6$ & \Cref{thm:ar_ap} \\ [1pt]
\hline
\multicolumn{4}{l}{\rule{0pt}{12pt}$(\diamond)$: Note that $\spm$ and $\opm$ are equivalent in the i.i.d.\ setting.}
\end{tabular}}
\caption{Summary III of revenue gaps; note that all ratios here are well understood.}
\label{tbl:summary3}
\end{table}

\vspace{.1in}
\noindent
{\bf Acknowledgement.}
We are grateful to Qi Qi and Simai He for many helpful discussions, and would like to thank the anonymous reviewers for providing \Cref{rmk:lem:spm_tri_best,rmk:ar_monotone}, as well as other invaluable suggestions.

\appendix
\renewcommand{\appendixname}{Appendix~\Alph{section}}

\section{Mathematical Facts}
\label{subapp:math}
\subsection{Proof of \texorpdfstring{\Cref{lem:spm_ap_RQ}}{}}

{\bf \Cref{lem:spm_ap_RQ}.}
{\em For functions $\R(p) = -p \cdot \ln(1 - p^{-2})$ and $\Q(p) = -\ln(1 - p^{-2}) - \frac{1}{2} \cdot \sum_{k = 1}^{\infty} k^{-2} \cdot p^{-2k}$, the following holds:
\begin{enumerate}[font = {\em\bfseries}]
\item $\R'(p) = p \cdot \Q'(p) < 0$ for any $p \in (1, \infty)$.
\item $\lim\limits_{p \rightarrow 1^+} \R(p) = \lim\limits_{p \rightarrow 1^+} \Q(p) = \infty$ and $\lim\limits_{p \rightarrow \infty} \R(p) = \lim\limits_{p \rightarrow \infty} \Q(p) = 0$.
\end{enumerate}}

\begin{proof}
To see \Cref{lem:spm_ap_RQ}, we note that
\begin{enumerate*}[label = (\alph*), font = {\bfseries}]
\item $\ln(1 - z) = -\sum_{k = 1}^n k^{-1} \cdot z^k$;
\item $\sum_{k = 1}^{\infty} k^{-2} = \frac{\pi^2}{6}$; and
\item $\ln(1 + x) \leq x$ when $x \geq 0$.
\end{enumerate*}
Given these, for Part 1 of the lemma:
\begin{align}
\label{eq:R_derivative}
& \R(p) \overset{\bf (a)}{=} \sum_{k = 1}^n k^{-1} \cdot p^{-(2k - 1)} > 0
&& \Rightarrow \quad
\R'(p) = -\sum_{k = 1}^n (2 - k^{-1}) \cdot p^{-2k} < 0; \\
\label{eq:Q_derivative}
& \Q(p) \overset{\bf (a)}{=} \sum_{k = 1}^n \Big(k^{-1} - \frac{k^{-2}}{2}\Big) \cdot p^{-2k} > 0
&& \Rightarrow \quad
\Q'(p) = -\sum_{k = 1}^n (2 - k^{-1}) \cdot p^{-(2k + 1)} \overset{\eqref{eq:R_derivative}}{=} \frac{\R'(p)}{p}.
\end{align}
For the first chain of equalities in Part 2:
\begin{align*}
& \lim_{p \rightarrow 1^+} \R(p) \geq -\ln(1 - \lim_{p \rightarrow 1^+} p^{-2}) = \infty; \\
& \lim_{p \rightarrow 1^+} \Q(p) \geq -\ln(1 - \lim_{p \rightarrow 1^+} p^{-2}) - \frac{1}{2} \cdot \sum_{k = 1}^{\infty} k^{-2} \overset{\bf (b)}{=} \infty - \frac{\pi^2}{12}
= \infty.
\end{align*}
For the second chain of equalities in Part 2:
\[
\begin{aligned}
& 0 \overset{\eqref{eq:R_derivative}}{\leq}
\lim_{p \rightarrow \infty} \R(p)
= \lim_{p \rightarrow \infty} p \cdot \ln\Big(1 + \frac{1}{p^2 - 1}\Big)
\overset{\bf (c)}{\leq} \lim_{p \rightarrow \infty} \frac{p}{p^2 - 1} = 0; \\
& 0 \overset{\eqref{eq:Q_derivative}}{\leq}
\lim_{p \rightarrow \infty} \Q(p)
\leq -\ln(1 - \lim_{p \rightarrow \infty} p^{-2})
= 0.
\end{aligned}
\]
This completes the proof of \Cref{lem:spm_ap_RQ}.
\end{proof}

\subsection{Proof of \texorpdfstring{\Cref{lem:ineq1}}{}}
\label{subapp:ineq1}

\begin{fact}
\label{lem:ineq1}
$G(x, y) \leq 0$ for any $y \geq x > 1$, where
\[
G(x, y) \eqdef (1 - x^{-1}) \cdot (e^{\R(x) - \R(y)} - 1) + \big(\R(y) - \R(x)\big) - \big(\Q(y) - \Q(x)\big).
\]
\end{fact}

\begin{proof}
Since $G(x, x) \equiv 0$, it suffices to prove that $\frac{\partial G}{\partial y} \leq 0$ when $y \geq x > 1$. To this end,
\begin{enumerate*}[label = (\alph*), font = {\bfseries}]
\item recall Part 1 of \Cref{lem:spm_ap_RQ} that $\R'(p) = p \cdot \Q'(p) < 0$; and
\item define $g(x) \eqdef \ln(1 - x^{-1}) - x \cdot \ln(1 - x^{-2})$.
\end{enumerate*}
Then,
\begin{align*}
\frac{\partial G}{\partial y}
& = -(1 - x^{-1}) \cdot e^{\R(x) - \R(y)} \cdot \R'(y) + \R'(y) - \Q'(y) \\
& \overset{\bf (a)}{=} \big(-\R'(y)\big) \cdot e^{-\R(y)} \cdot \big[(1 - x^{-1}) \cdot e^{\R(x)} - (1 - y^{-1}) \cdot e^{\R(y)}\big] \\
& \overset{\bf (b)}{=} \big(-\R'(y)\big) \cdot e^{-\R(y)} \cdot (e^{g(x)} - e^{g(y)}).
\end{align*}
Since $\ln(1 - z) \leq -z$ for any $z \in (0, 1)$, we observe that
\[
g'(x) = -\ln(1 - x^{-2}) - x^{-1} \cdot (x + 1)^{-1} \geq x^{-2} - x^{-1} \cdot (x + 1)^{-1} \geq 0,
\]
which implies that $\frac{\partial G}{\partial y} \leq 0$ when $y \geq x > 1$. This completes the proof of \Cref{lem:ineq1}.
\end{proof}

\subsection{Proof of \texorpdfstring{\Cref{lem:ineq2}}{}}
\label{subapp:ineq2}

\begin{fact}
\label{lem:ineq2}
$H(x, y) \geq 0$ for any $y \geq x > 1$, where
\[
H(x, y) \eqdef x^{-1} \cdot (e^{\R(x) - \R(y)} - 1) - (e^{\Q(x) - \Q(y)} - 1).
\]
\end{fact}

\begin{proof}
Since $H(x, x) \equiv 0$, it suffices to prove that $\frac{\partial H}{\partial y} \geq 0$ when $y \geq x > 1$. To this end,
\begin{enumerate*}[label = (\alph*), font = {\bfseries}]
\item recall Proof 1 of \Cref{lem:spm_ap_RQ} that $\R'(p) = p \cdot \Q'(p) < 0$; and
\item define $h(x) \eqdef x \cdot e^{\Q(x) - \R(x)}$.
\end{enumerate*}
Then,
\begin{align*}
\frac{\partial H}{\partial y}
& = x^{-1} \cdot e^{\R(x) - \R(y)} \cdot \big(-\R'(y)\big) - e^{\Q(x) - \Q(y)} \cdot \big(-\Q'(y)\big) \\
& \overset{\bf (a)}{=} x^{-1} \cdot e^{\R(x) - \Q(y)}  \cdot \big(-\Q'(y)\big) \cdot (y \cdot e^{\Q(y) - \R(y)} - x \cdot e^{\Q(x) - \R(x)}) \\
& \overset{\bf (b)}{=} x^{-1} \cdot e^{\R(x) - \Q(y)}  \cdot \big(-\Q'(y)\big) \cdot \big(h(y) - h(x)\big)
\end{align*}
We observe that, for any $x \geq 1$,
\begin{align*}
h'(x)
& = e^{\Q(x) - \R(x)} \cdot \big(1 + x \cdot \Q'(x) - x \cdot \R'(x)\big) \\
& \overset{\bf (a)}{=} e^{\Q(x) - \R(x)} \cdot \Big[1 + (x - 1) \cdot \big(-\R'(x)\big)\Big]
\geq 0.
\end{align*}
which implies that $\frac{\partial H}{\partial y} \geq 0$ when $y \geq x > 1$. This completes the proof of \Cref{lem:ineq2}.
\end{proof}

\section{Proof of Lemma~\ref{lem:spm_tri_best}}
\label{app:prelim}
{\bf \Cref{lem:spm_tri_best}.}
{\em For any triangular instance $\{\tri(v_i, q_i)\}_{i = 1}^n$ that $v_1 \geq v_2 \geq \cdots \geq v_n$, it follows that:
\begin{enumerate}[font = {\em\bfseries}]
\item $\opt = \spm = \sum_{i = 1}^n v_i q_i \cdot \prod_{j = 1}^{i - 1} (1 - q_j)$.
\item An optimal $\spm$ lets the buyers come in lexicographic order, and posts price $p_i = v_i$ to each buyer $i \in [n]$.
\end{enumerate}}

\begin{proof}
For convenience, we assume that parameters $\{v_i\}_{i = 1}^n$ are all finite. (Otherwise, we can slightly modify the proof via a standard argument from measure theory.) Any triangular distribution $\tri(v_i, q_i)$ is supported on interval $p \in [0,v_i]$, and has a CDF of $F_i(p) = \frac{(1 - q_i) \cdot p}{(1 - q_i) \cdot p + v_i q_i}$ when $p \in [0, v_i)$, and $F_i(p) = 1$ when $p \in [v_i, \infty)$. Thus, any value $p \in (0, v_i)$ corresponds to a virtual value of $\varphi_i(p) = p - \frac{1 - F_i(p)}{f_i(p)} = -\frac{v_i q_i}{1 - q_i}$, and value $p = v_i$ corresponds to a virtual value of $v_i$.

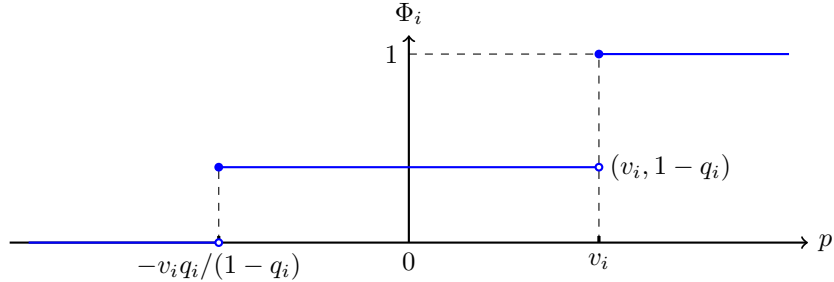
\begin{figure}[h!]
\centering
\begin{tikzpicture}[thick, smooth, scale = 2.5]
\draw[->] (-2.1, 0) -- (2.1, 0);
\draw[->] (0, 0) -- (0, 1.1);
\node[above] at (0, 1.1) {\small $\Phi_i$};
\node[right] at (2.1, 0) {\small $p$};
\node[below] at (0, 0) {\small $0$};
\draw[color = blue] (-2, 0) -- (-1, 0) (-1, 0.4) -- (1, 0.4) (1, 1) -- (2, 1);
\draw[thin, color = black, dashed] (1, 0) -- (1, 1);
\draw[thin, color = black, dashed] (0, 1) -- (1, 1);
\draw[thin, color = black, dashed] (-1, 0) -- (-1, 0.4);
\draw[very thick] (-1, 0) -- (-1, 1pt);
\node[below] at (-1, 0) {\small $-v_i q_i / (1 - q_i)$};
\draw[very thick] (1, 0) -- (1, 1pt);
\node[below] at (1, 0) {\small $v_i$};
\node[right] at (1, 0.4) {\small $(v_i, 1 - q_i)$};
\node[left] at (0, 1) {\small $1$};
\draw[color = blue, fill = white] (1, 0.4) circle(0.5pt);
\draw[color = blue, fill = blue] (1, 1) circle(0.5pt);
\draw[color = blue, fill = white] (-1, 0) circle(0.5pt);
\draw[color = blue, fill = blue] (-1, 0.4) circle(0.5pt);
\end{tikzpicture}
\caption{Demonstration for virtual value CDF $\Phi_i$ of triangular distribution $\tri(v_i, q_i)$.}
\label{fig:lem:spm_tri_best}
\end{figure}

\noindent
As \Cref{fig:lem:spm_tri_best} demonstrates, the resulting virtual value CDF $\Phi_i$ is given by
\[
\Phi_i(p) =
(1 - q_i) \cdot \mathbbm{1}\big\{p \geq -v_i q_i / (1 - q_i)\big\} + q_i \cdot \mathbbm{1}\big\{p \geq v_i\big\}, \quad\quad\quad\quad \forall p \in \RR,
\]
So, for any non-negative virtual value $p \in \RRP$ the product CDF $\Phi(\price) \eqdef \prod_{i = 1}^n \Phi_i(\price)$ equals
\[
\Phi(\price) = \prod_{i: v_i > p} (1 - q_i).
\]
Let $v_{n + 1} \eqdef 0$ for notational brevity. By the {\em revenue-equivalence theorem} \citep[see][]{M81}, {\sf Myerson Auction} gives a revenue of
\begin{align*}
\opt
& = \mbox{$\E_{\bids \sim \distrs} \Big[\max_{i \in [n]} \big(\varphi_i(b_i)\big)_+\Big]$}
= \mbox{$\Int{0}{\infty} \big(1 - \Phi(x)\big) \cdot \dd x$} \\
& = \sum_{i = 1}^n \big[1 - \prod_{j = 1}^{i} (1 - q_j)\big] \cdot (v_i - v_{i + 1})
= \sum_{i = 1}^n v_i q_i \cdot \prod_{j = 1}^{i - 1} (1 - q_j).
\end{align*}
We then prove that the {\sf Sequential Posted-Pricing} proposed in Part 2 of \Cref{lem:spm_tri_best} actually extracts the same revenue. This follows from the next two observations: when each buyer $k \in [n]$ comes,
\begin{itemize}
\item The item remains unsold with a probability of $\prod_{j = 1}^{k - 1} (1 - q_j)$.
\item If so, buyer $k$ purchases the item at the posted price of $p_k = v_k$ with a probability of $q_k$.
\end{itemize}
In that this {\sf Sequential Posted-Pricing} generates the same revenue as {\sf Myerson Auction}, its optimality trivially holds. This completes the proof of \Cref{lem:spm_tri_best}.
\end{proof}

\section{Missing Proofs in Section~\ref{sec:spm_ap}}
\label{app:spm_ap}
\subsection{Numeric Calculation about Constant \texorpdfstring{$\C \approx 2.6202$}{}}
\label{subapp:spm_ap_calculations}

Define function ${\cal F}(z) \eqdef z^{-2} - (z^{-2} - 1) \cdot e^{\frac{1}{2} \cdot \sum_{k = 1}^{\infty} k^{-2} \cdot z^{2k}}$ on interval $z \in (0, 1)$. Recall \Cref{thm:spm_ap} for function $\Q$ and constant $\C$. Applying integration by substitution (i.e.\ $z = x^{-1}$), we have
\[
\C
= \mbox{$2 + \Int{1}{\infty} \big(1 - e^{-\Q(x)}\big) \cdot \dd x$}
= \mbox{$2 + \Int{0}{1} {\cal F}(z) \cdot \dd z$}.
\]
To make things mimic, we demonstrate the both integrands in the following figure; it is easy to infer from \Cref{fig:spm_ap2} that the integral is finite.

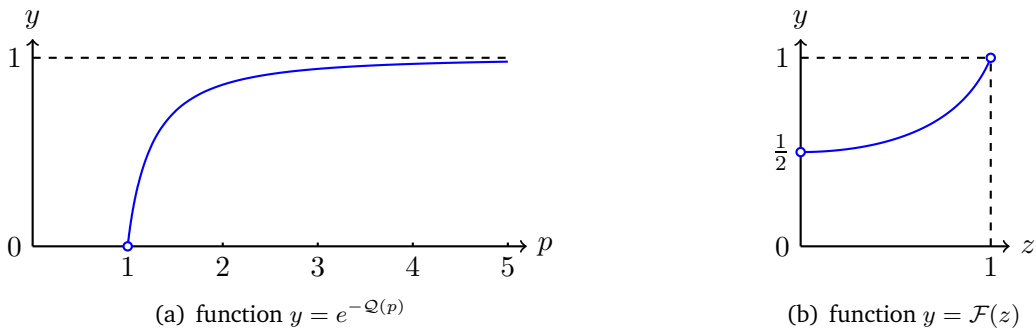
\begin{figure}[h!]
\centering
\subfigure[function $y = e^{-\Q(p)}$]{
\begin{tikzpicture}[thick, smooth, scale = 2.5]
\draw[->] (0, 0) -- (2.6, 0);
\draw[->] (0, 0) -- (0, 1.1);
\node[above] at (0, 1.1) {$y$};
\node[right] at (2.6, 0) {$p$};
\node[left] at (0, 1) {$1$};
\node[left] at (0, 0) {$0$};
\draw[dashed] (0, 1) -- (2.5, 1);
\draw[thick, color = blue] (0.500, 0.0000) -- (0.502, 0.0220) -- (0.505, 0.0427) -- (0.507, 0.0624) -- (0.510, 0.0813) -- (0.512, 0.0993) -- (0.515, 0.1166) -- (0.517, 0.1333) -- (0.520, 0.1493) -- (0.522, 0.1648) -- (0.525, 0.1797) -- (0.527, 0.1941) -- (0.530, 0.2080) -- (0.532, 0.2215) -- (0.535, 0.2346) -- (0.537, 0.2473) -- (0.540, 0.2596) -- (0.542, 0.2715) -- (0.545, 0.2831) -- (0.547, 0.2943) -- (0.550, 0.3053) -- (0.552, 0.3159) -- (0.555, 0.3263) -- (0.558, 0.3363) -- (0.560, 0.3461) -- (0.563, 0.3557) -- (0.565, 0.3650) -- (0.568, 0.3741) -- (0.570, 0.3829) -- (0.573, 0.3916) -- (0.575, 0.4000) -- (0.578, 0.4083) -- (0.580, 0.4163) -- (0.583, 0.4241) -- (0.585, 0.4318) -- (0.588, 0.4393) -- (0.590, 0.4466) -- (0.593, 0.4538) -- (0.595, 0.4608) -- (0.598, 0.4677) -- (0.600, 0.4744) -- (0.603, 0.4810) -- (0.605, 0.4874) -- (0.608, 0.4937) -- (0.610, 0.4999) -- (0.613, 0.5059) -- (0.615, 0.5118) -- (0.618, 0.5176) -- (0.620, 0.5233) -- (0.623, 0.5289) -- (0.625, 0.5344) -- (0.627, 0.5397) -- (0.630, 0.5450) -- (0.632, 0.5502) -- (0.635, 0.5552) -- (0.637, 0.5602) -- (0.640, 0.5651) -- (0.642, 0.5699) -- (0.645, 0.5746) -- (0.647, 0.5792) -- (0.650, 0.5837) -- (0.652, 0.5882) -- (0.655, 0.5926) -- (0.657, 0.5969) -- (0.660, 0.6011) -- (0.662, 0.6053) -- (0.665, 0.6093) -- (0.667, 0.6134) -- (0.670, 0.6173) -- (0.672, 0.6212) -- (0.675, 0.6250) -- (0.677, 0.6288) -- (0.680, 0.6325) -- (0.683, 0.6361) -- (0.685, 0.6397) -- (0.688, 0.6432) -- (0.690, 0.6467) -- (0.693, 0.6501) -- (0.695, 0.6534) -- (0.698, 0.6567) -- (0.700, 0.6600) -- (0.703, 0.6632) -- (0.705, 0.6663) -- (0.708, 0.6694) -- (0.710, 0.6725) -- (0.713, 0.6755) -- (0.715, 0.6785) -- (0.718, 0.6814) -- (0.720, 0.6843) -- (0.723, 0.6871) -- (0.725, 0.6899) -- (0.728, 0.6927) -- (0.730, 0.6954) -- (0.733, 0.6981) -- (0.735, 0.7008) -- (0.738, 0.7034) -- (0.740, 0.7059) -- (0.743, 0.7085) -- (0.745, 0.7110) -- (0.748, 0.7134) -- (0.750, 0.7159) -- (0.752, 0.7183) -- (0.755, 0.7206) -- (0.757, 0.7230) -- (0.760, 0.7253) -- (0.762, 0.7275) -- (0.765, 0.7298) -- (0.767, 0.7320) -- (0.770, 0.7342) -- (0.772, 0.7363) -- (0.775, 0.7385) -- (0.777, 0.7406) -- (0.780, 0.7426) -- (0.782, 0.7447) -- (0.785, 0.7467) -- (0.787, 0.7487) -- (0.790, 0.7507) -- (0.792, 0.7526) -- (0.795, 0.7545) -- (0.797, 0.7564) -- (0.800, 0.7583) -- (0.802, 0.7602) -- (0.805, 0.7620) -- (0.808, 0.7638) -- (0.810, 0.7656) -- (0.813, 0.7674) -- (0.815, 0.7691) -- (0.818, 0.7708) -- (0.820, 0.7725) -- (0.823, 0.7742) -- (0.825, 0.7759) -- (0.828, 0.7775) -- (0.830, 0.7792) -- (0.833, 0.7808) -- (0.835, 0.7823) -- (0.838, 0.7839) -- (0.840, 0.7855) -- (0.843, 0.7870) -- (0.845, 0.7885) -- (0.848, 0.7900) -- (0.850, 0.7915) -- (0.853, 0.7930) -- (0.855, 0.7944) -- (0.858, 0.7959) -- (0.860, 0.7973) -- (0.863, 0.7987) -- (0.865, 0.8001) -- (0.868, 0.8014) -- (0.870, 0.8028) -- (0.873, 0.8041) -- (0.875, 0.8055) -- (0.877, 0.8068) -- (0.880, 0.8081) -- (0.882, 0.8094) -- (0.885, 0.8106) -- (0.887, 0.8119) -- (0.890, 0.8131) -- (0.892, 0.8144) -- (0.895, 0.8156) -- (0.897, 0.8168) -- (0.900, 0.8180) -- (0.902, 0.8192) -- (0.905, 0.8204) -- (0.907, 0.8215) -- (0.910, 0.8227) -- (0.912, 0.8238) -- (0.915, 0.8249) -- (0.917, 0.8261) -- (0.920, 0.8272) -- (0.922, 0.8282) -- (0.925, 0.8293) -- (0.927, 0.8304) -- (0.930, 0.8315) -- (0.933, 0.8325) -- (0.935, 0.8335) -- (0.938, 0.8346) -- (0.940, 0.8356) -- (0.943, 0.8366) -- (0.945, 0.8376) -- (0.948, 0.8386) -- (0.950, 0.8396) -- (0.953, 0.8405) -- (0.955, 0.8415) -- (0.958, 0.8425) -- (0.960, 0.8434) -- (0.963, 0.8443) -- (0.965, 0.8453) -- (0.968, 0.8462) -- (0.970, 0.8471) -- (0.973, 0.8480) -- (0.975, 0.8489) -- (0.978, 0.8498) -- (0.980, 0.8507) -- (0.983, 0.8515) -- (0.985, 0.8524) -- (0.988, 0.8532) -- (0.990, 0.8541) -- (0.993, 0.8549) -- (0.995, 0.8558) -- (0.998, 0.8566) -- (1.000, 0.8574) -- (1.002, 0.8582) -- (1.005, 0.8590) -- (1.008, 0.8598) -- (1.010, 0.8606) -- (1.013, 0.8614) -- (1.015, 0.8621) -- (1.018, 0.8629) -- (1.020, 0.8637) -- (1.023, 0.8644) -- (1.025, 0.8652) -- (1.028, 0.8659) -- (1.030, 0.8666) -- (1.033, 0.8674) -- (1.035, 0.8681) -- (1.038, 0.8688) -- (1.040, 0.8695) -- (1.043, 0.8702) -- (1.045, 0.8709) -- (1.048, 0.8716) -- (1.050, 0.8723) -- (1.053, 0.8730) -- (1.055, 0.8737) -- (1.058, 0.8743) -- (1.060, 0.8750) -- (1.063, 0.8757) -- (1.065, 0.8763) -- (1.067, 0.8770) -- (1.070, 0.8776) -- (1.073, 0.8782) -- (1.075, 0.8789) -- (1.077, 0.8795) -- (1.080, 0.8801) -- (1.083, 0.8807) -- (1.085, 0.8814) -- (1.087, 0.8820) -- (1.090, 0.8826) -- (1.093, 0.8832) -- (1.095, 0.8838) -- (1.097, 0.8843) -- (1.100, 0.8849) -- (1.103, 0.8855) -- (1.105, 0.8861) -- (1.107, 0.8867) -- (1.110, 0.8872) -- (1.113, 0.8878) -- (1.115, 0.8883) -- (1.117, 0.8889) -- (1.120, 0.8894) -- (1.123, 0.8900) -- (1.125, 0.8905) -- (1.127, 0.8911) -- (1.130, 0.8916) -- (1.133, 0.8921) -- (1.135, 0.8927) -- (1.138, 0.8932) -- (1.140, 0.8937) -- (1.143, 0.8942) -- (1.145, 0.8947) -- (1.148, 0.8952) -- (1.150, 0.8957) -- (1.153, 0.8962) -- (1.155, 0.8967) -- (1.158, 0.8972) -- (1.160, 0.8977) -- (1.163, 0.8982) -- (1.165, 0.8987) -- (1.168, 0.8991) -- (1.170, 0.8996) -- (1.173, 0.9001) -- (1.175, 0.9005) -- (1.178, 0.9010) -- (1.180, 0.9015) -- (1.183, 0.9019) -- (1.185, 0.9024) -- (1.188, 0.9028) -- (1.190, 0.9033) -- (1.192, 0.9037) -- (1.195, 0.9042) -- (1.198, 0.9046) -- (1.200, 0.9050) -- (1.202, 0.9055) -- (1.205, 0.9059) -- (1.208, 0.9063) -- (1.210, 0.9067) -- (1.212, 0.9071) -- (1.215, 0.9076) -- (1.218, 0.9080) -- (1.220, 0.9084) -- (1.222, 0.9088) -- (1.225, 0.9092) -- (1.228, 0.9096) -- (1.230, 0.9100) -- (1.232, 0.9104) -- (1.235, 0.9108) -- (1.238, 0.9112) -- (1.240, 0.9116) -- (1.242, 0.9120) -- (1.245, 0.9123) -- (1.248, 0.9127) -- (1.250, 0.9131) -- (1.252, 0.9135) -- (1.255, 0.9139) -- (1.258, 0.9142) -- (1.260, 0.9146) -- (1.263, 0.9150) -- (1.265, 0.9153) -- (1.268, 0.9157) -- (1.270, 0.9160) -- (1.273, 0.9164) -- (1.275, 0.9168) -- (1.278, 0.9171) -- (1.280, 0.9175) -- (1.283, 0.9178) -- (1.285, 0.9182) -- (1.288, 0.9185) -- (1.290, 0.9188) -- (1.293, 0.9192) -- (1.295, 0.9195) -- (1.298, 0.9198) -- (1.300, 0.9202) -- (1.303, 0.9205) -- (1.305, 0.9208) -- (1.308, 0.9212) -- (1.310, 0.9215) -- (1.313, 0.9218) -- (1.315, 0.9221) -- (1.317, 0.9224) -- (1.320, 0.9228) -- (1.323, 0.9231) -- (1.325, 0.9234) -- (1.327, 0.9237) -- (1.330, 0.9240) -- (1.333, 0.9243) -- (1.335, 0.9246) -- (1.337, 0.9249) -- (1.340, 0.9252) -- (1.343, 0.9255) -- (1.345, 0.9258) -- (1.347, 0.9261) -- (1.350, 0.9264) -- (1.353, 0.9267) -- (1.355, 0.9270) -- (1.357, 0.9273) -- (1.360, 0.9276) -- (1.363, 0.9278) -- (1.365, 0.9281) -- (1.367, 0.9284) -- (1.370, 0.9287) -- (1.373, 0.9290) -- (1.375, 0.9292) -- (1.377, 0.9295) -- (1.380, 0.9298) -- (1.383, 0.9301) -- (1.385, 0.9303) -- (1.388, 0.9306) -- (1.390, 0.9309) -- (1.393, 0.9311) -- (1.395, 0.9314) -- (1.398, 0.9317) -- (1.400, 0.9319) -- (1.403, 0.9322) -- (1.405, 0.9324) -- (1.408, 0.9327) -- (1.410, 0.9329) -- (1.413, 0.9332) -- (1.415, 0.9334) -- (1.418, 0.9337) -- (1.420, 0.9339) -- (1.423, 0.9342) -- (1.425, 0.9344) -- (1.428, 0.9347) -- (1.430, 0.9349) -- (1.433, 0.9352) -- (1.435, 0.9354) -- (1.438, 0.9356) -- (1.440, 0.9359) -- (1.442, 0.9361) -- (1.445, 0.9363) -- (1.448, 0.9366) -- (1.450, 0.9368) -- (1.452, 0.9370) -- (1.455, 0.9373) -- (1.458, 0.9375) -- (1.460, 0.9377) -- (1.462, 0.9380) -- (1.465, 0.9382) -- (1.468, 0.9384) -- (1.470, 0.9386) -- (1.472, 0.9388) -- (1.475, 0.9391) -- (1.478, 0.9393) -- (1.480, 0.9395) -- (1.482, 0.9397) -- (1.485, 0.9399) -- (1.488, 0.9401) -- (1.490, 0.9404) -- (1.492, 0.9406) -- (1.495, 0.9408) -- (1.498, 0.9410) -- (1.500, 0.9412) -- (1.502, 0.9414) -- (1.505, 0.9416) -- (1.508, 0.9418) -- (1.510, 0.9420) -- (1.513, 0.9422) -- (1.515, 0.9424) -- (1.518, 0.9426) -- (1.520, 0.9428) -- (1.523, 0.9430) -- (1.525, 0.9432) -- (1.528, 0.9434) -- (1.530, 0.9436) -- (1.533, 0.9438) -- (1.535, 0.9440) -- (1.538, 0.9442) -- (1.540, 0.9444) -- (1.543, 0.9446) -- (1.545, 0.9448) -- (1.548, 0.9449) -- (1.550, 0.9451) -- (1.553, 0.9453) -- (1.555, 0.9455) -- (1.558, 0.9457) -- (1.560, 0.9459) -- (1.563, 0.9461) -- (1.565, 0.9462) -- (1.567, 0.9464) -- (1.570, 0.9466) -- (1.573, 0.9468) -- (1.575, 0.9470) -- (1.577, 0.9471) -- (1.580, 0.9473) -- (1.583, 0.9475) -- (1.585, 0.9477) -- (1.587, 0.9478) -- (1.590, 0.9480) -- (1.593, 0.9482) -- (1.595, 0.9483) -- (1.597, 0.9485) -- (1.600, 0.9487) -- (1.603, 0.9488) -- (1.605, 0.9490) -- (1.607, 0.9492) -- (1.610, 0.9493) -- (1.613, 0.9495) -- (1.615, 0.9497) -- (1.617, 0.9498) -- (1.620, 0.9500) -- (1.623, 0.9502) -- (1.625, 0.9503) -- (1.627, 0.9505) -- (1.630, 0.9506) -- (1.633, 0.9508) -- (1.635, 0.9510) -- (1.638, 0.9511) -- (1.640, 0.9513) -- (1.643, 0.9514) -- (1.645, 0.9516) -- (1.648, 0.9517) -- (1.650, 0.9519) -- (1.653, 0.9520) -- (1.655, 0.9522) -- (1.658, 0.9523) -- (1.660, 0.9525) -- (1.663, 0.9526) -- (1.665, 0.9528) -- (1.668, 0.9529) -- (1.670, 0.9531) -- (1.673, 0.9532) -- (1.675, 0.9534) -- (1.678, 0.9535) -- (1.680, 0.9537) -- (1.683, 0.9538) -- (1.685, 0.9540) -- (1.688, 0.9541) -- (1.690, 0.9542) -- (1.692, 0.9544) -- (1.695, 0.9545) -- (1.698, 0.9547) -- (1.700, 0.9548) -- (1.702, 0.9549) -- (1.705, 0.9551) -- (1.708, 0.9552) -- (1.710, 0.9554) -- (1.712, 0.9555) -- (1.715, 0.9556) -- (1.718, 0.9558) -- (1.720, 0.9559) -- (1.722, 0.9560) -- (1.725, 0.9562) -- (1.728, 0.9563) -- (1.730, 0.9564) -- (1.732, 0.9566) -- (1.735, 0.9567) -- (1.738, 0.9568) -- (1.740, 0.9569) -- (1.742, 0.9571) -- (1.745, 0.9572) -- (1.748, 0.9573) -- (1.750, 0.9575) -- (1.752, 0.9576) -- (1.755, 0.9577) -- (1.758, 0.9578) -- (1.760, 0.9580) -- (1.763, 0.9581) -- (1.765, 0.9582) -- (1.768, 0.9583) -- (1.770, 0.9585) -- (1.773, 0.9586) -- (1.775, 0.9587) -- (1.778, 0.9588) -- (1.780, 0.9589) -- (1.783, 0.9591) -- (1.785, 0.9592) -- (1.788, 0.9593) -- (1.790, 0.9594) -- (1.793, 0.9595) -- (1.795, 0.9596) -- (1.798, 0.9598) -- (1.800, 0.9599) -- (1.803, 0.9600) -- (1.805, 0.9601) -- (1.808, 0.9602) -- (1.810, 0.9603) -- (1.813, 0.9605) -- (1.815, 0.9606) -- (1.817, 0.9607) -- (1.820, 0.9608) -- (1.823, 0.9609) -- (1.825, 0.9610) -- (1.827, 0.9611) -- (1.830, 0.9612) -- (1.833, 0.9613) -- (1.835, 0.9615) -- (1.837, 0.9616) -- (1.840, 0.9617) -- (1.843, 0.9618) -- (1.845, 0.9619) -- (1.847, 0.9620) -- (1.850, 0.9621) -- (1.853, 0.9622) -- (1.855, 0.9623) -- (1.857, 0.9624) -- (1.860, 0.9625) -- (1.863, 0.9626) -- (1.865, 0.9627) -- (1.867, 0.9628) -- (1.870, 0.9629) -- (1.873, 0.9630) -- (1.875, 0.9631) -- (1.877, 0.9632) -- (1.880, 0.9633) -- (1.883, 0.9634) -- (1.885, 0.9635) -- (1.888, 0.9636) -- (1.890, 0.9637) -- (1.893, 0.9638) -- (1.895, 0.9639) -- (1.898, 0.9640) -- (1.900, 0.9641) -- (1.903, 0.9642) -- (1.905, 0.9643) -- (1.908, 0.9644) -- (1.910, 0.9645) -- (1.913, 0.9646) -- (1.915, 0.9647) -- (1.918, 0.9648) -- (1.920, 0.9649) -- (1.923, 0.9650) -- (1.925, 0.9651) -- (1.928, 0.9652) -- (1.930, 0.9653) -- (1.933, 0.9654) -- (1.935, 0.9655) -- (1.938, 0.9656) -- (1.940, 0.9657) -- (1.942, 0.9657) -- (1.945, 0.9658) -- (1.948, 0.9659) -- (1.950, 0.9660) -- (1.952, 0.9661) -- (1.955, 0.9662) -- (1.958, 0.9663) -- (1.960, 0.9664) -- (1.962, 0.9665) -- (1.965, 0.9665) -- (1.968, 0.9666) -- (1.970, 0.9667) -- (1.972, 0.9668) -- (1.975, 0.9669) -- (1.978, 0.9670) -- (1.980, 0.9671) -- (1.982, 0.9672) -- (1.985, 0.9672) -- (1.988, 0.9673) -- (1.990, 0.9674) -- (1.992, 0.9675) -- (1.995, 0.9676) -- (1.998, 0.9677) -- (2.000, 0.9677) -- (2.002, 0.9678) -- (2.005, 0.9679) -- (2.007, 0.9680) -- (2.010, 0.9681) -- (2.013, 0.9682) -- (2.015, 0.9682) -- (2.018, 0.9683) -- (2.020, 0.9684) -- (2.023, 0.9685) -- (2.025, 0.9686) -- (2.027, 0.9686) -- (2.030, 0.9687) -- (2.033, 0.9688) -- (2.035, 0.9689) -- (2.038, 0.9690) -- (2.040, 0.9690) -- (2.043, 0.9691) -- (2.045, 0.9692) -- (2.047, 0.9693) -- (2.050, 0.9693) -- (2.053, 0.9694) -- (2.055, 0.9695) -- (2.058, 0.9696) -- (2.060, 0.9697) -- (2.063, 0.9697) -- (2.065, 0.9698) -- (2.067, 0.9699) -- (2.070, 0.9700) -- (2.072, 0.9700) -- (2.075, 0.9701) -- (2.078, 0.9702) -- (2.080, 0.9703) -- (2.083, 0.9703) -- (2.085, 0.9704) -- (2.087, 0.9705) -- (2.090, 0.9705) -- (2.092, 0.9706) -- (2.095, 0.9707) -- (2.098, 0.9708) -- (2.100, 0.9708) -- (2.103, 0.9709) -- (2.105, 0.9710) -- (2.107, 0.9710) -- (2.110, 0.9711) -- (2.112, 0.9712) -- (2.115, 0.9713) -- (2.118, 0.9713) -- (2.120, 0.9714) -- (2.123, 0.9715) -- (2.125, 0.9715) -- (2.127, 0.9716) -- (2.130, 0.9717) -- (2.132, 0.9717) -- (2.135, 0.9718) -- (2.138, 0.9719) -- (2.140, 0.9719) -- (2.143, 0.9720) -- (2.145, 0.9721) -- (2.148, 0.9721) -- (2.150, 0.9722) -- (2.152, 0.9723) -- (2.155, 0.9723) -- (2.158, 0.9724) -- (2.160, 0.9725) -- (2.163, 0.9725) -- (2.165, 0.9726) -- (2.168, 0.9727) -- (2.170, 0.9727) -- (2.172, 0.9728) -- (2.175, 0.9729) -- (2.178, 0.9729) -- (2.180, 0.9730) -- (2.183, 0.9731) -- (2.185, 0.9731) -- (2.188, 0.9732) -- (2.190, 0.9732) -- (2.192, 0.9733) -- (2.195, 0.9734) -- (2.197, 0.9734) -- (2.200, 0.9735) -- (2.203, 0.9736) -- (2.205, 0.9736) -- (2.208, 0.9737) -- (2.210, 0.9737) -- (2.212, 0.9738) -- (2.215, 0.9739) -- (2.217, 0.9739) -- (2.220, 0.9740) -- (2.223, 0.9740) -- (2.225, 0.9741) -- (2.228, 0.9742) -- (2.230, 0.9742) -- (2.232, 0.9743) -- (2.235, 0.9743) -- (2.237, 0.9744) -- (2.240, 0.9745) -- (2.243, 0.9745) -- (2.245, 0.9746) -- (2.248, 0.9746) -- (2.250, 0.9747) -- (2.252, 0.9747) -- (2.255, 0.9748) -- (2.257, 0.9749) -- (2.260, 0.9749) -- (2.263, 0.9750) -- (2.265, 0.9750) -- (2.268, 0.9751) -- (2.270, 0.9751) -- (2.273, 0.9752) -- (2.275, 0.9753) -- (2.277, 0.9753) -- (2.280, 0.9754) -- (2.283, 0.9754) -- (2.285, 0.9755) -- (2.288, 0.9755) -- (2.290, 0.9756) -- (2.293, 0.9756) -- (2.295, 0.9757) -- (2.297, 0.9757) -- (2.300, 0.9758) -- (2.303, 0.9759) -- (2.305, 0.9759) -- (2.308, 0.9760) -- (2.310, 0.9760) -- (2.313, 0.9761) -- (2.315, 0.9761) -- (2.317, 0.9762) -- (2.320, 0.9762) -- (2.322, 0.9763) -- (2.325, 0.9763) -- (2.328, 0.9764) -- (2.330, 0.9764) -- (2.333, 0.9765) -- (2.335, 0.9765) -- (2.337, 0.9766) -- (2.340, 0.9766) -- (2.342, 0.9767) -- (2.345, 0.9767) -- (2.348, 0.9768) -- (2.350, 0.9768) -- (2.353, 0.9769) -- (2.355, 0.9769) -- (2.357, 0.9770) -- (2.360, 0.9770) -- (2.362, 0.9771) -- (2.365, 0.9771) -- (2.368, 0.9772) -- (2.370, 0.9772) -- (2.373, 0.9773) -- (2.375, 0.9773) -- (2.377, 0.9774) -- (2.380, 0.9774) -- (2.382, 0.9775) -- (2.385, 0.9775) -- (2.388, 0.9776) -- (2.390, 0.9776) -- (2.393, 0.9777) -- (2.395, 0.9777) -- (2.398, 0.9778) -- (2.400, 0.9778) -- (2.402, 0.9779) -- (2.405, 0.9779) -- (2.408, 0.9780) -- (2.410, 0.9780) -- (2.413, 0.9781) -- (2.415, 0.9781) -- (2.418, 0.9781) -- (2.420, 0.9782) -- (2.422, 0.9782) -- (2.425, 0.9783) -- (2.428, 0.9783) -- (2.430, 0.9784) -- (2.433, 0.9784) -- (2.435, 0.9785) -- (2.438, 0.9785) -- (2.440, 0.9786) -- (2.442, 0.9786) -- (2.445, 0.9786) -- (2.447, 0.9787) -- (2.450, 0.9787) -- (2.453, 0.9788) -- (2.455, 0.9788) -- (2.458, 0.9789) -- (2.460, 0.9789) -- (2.462, 0.9790) -- (2.465, 0.9790) -- (2.467, 0.9790) -- (2.470, 0.9791) -- (2.473, 0.9791) -- (2.475, 0.9792) -- (2.478, 0.9792) -- (2.480, 0.9793) -- (2.482, 0.9793) -- (2.485, 0.9793) -- (2.487, 0.9794) -- (2.490, 0.9794) -- (2.493, 0.9795) -- (2.495, 0.9795) -- (2.498, 0.9796) -- (2.500, 0.9796);
\draw[thick] (0.5, 0) -- (0.5, 0.02);
\node[below] at (0.5, 0) {$1$};
\draw[thick] (1.0, 0) -- (1.0, 0.02);
\node[below] at (1.0, 0) {$2$};
\draw[thick] (1.5, 0) -- (1.5, 0.02);
\node[below] at (1.5, 0) {$3$};
\draw[thick] (2.0, 0) -- (2.0, 0.02);
\node[below] at (2.0, 0) {$4$};
\draw[thick] (2.5, 0) -- (2.5, 0.02);
\node[below] at (2.5, 0) {$5$};
\draw[color = blue, fill = white, thick] (0.5, 0) circle(0.625pt);
\end{tikzpicture}
\label{fig:spm_ap1}
}
\quad\quad\quad\quad\quad\quad
\subfigure[function $y = {\cal F}(z)$]{
\begin{tikzpicture}[thick, smooth, scale = 2.5]
\draw[->] (0, 0) -- (1.1, 0);
\draw[->] (0, 0) -- (0, 1.1);
\node[above] at (0, 1.1) {$y$};
\node[right] at (1.1, 0) {$z$};
\node[left] at (0, 1) {$1$};
\node[left] at (0, 0.5) {$\frac{1}{2}$};
\node[left] at (0, 0) {$0$};
\node[below] at (1, 0) {$1$};
\draw[dashed] (0, 1) -- (1, 1);
\draw[dashed] (1, 0) -- (1, 1);
\draw[thick, color = blue] (0.000, 0.5000) -- (0.001, 0.5000) -- (0.002, 0.5000) -- (0.003, 0.5000) -- (0.004, 0.5000) -- (0.005, 0.5000) -- (0.006, 0.5000) -- (0.007, 0.5000) -- (0.008, 0.5000) -- (0.009, 0.5000) -- (0.010, 0.5000) -- (0.011, 0.5000) -- (0.012, 0.5000) -- (0.013, 0.5000) -- (0.014, 0.5000) -- (0.015, 0.5001) -- (0.016, 0.5001) -- (0.017, 0.5001) -- (0.018, 0.5001) -- (0.019, 0.5001) -- (0.020, 0.5001) -- (0.021, 0.5001) -- (0.022, 0.5001) -- (0.023, 0.5001) -- (0.024, 0.5001) -- (0.025, 0.5002) -- (0.026, 0.5002) -- (0.027, 0.5002) -- (0.028, 0.5002) -- (0.029, 0.5002) -- (0.030, 0.5002) -- (0.031, 0.5002) -- (0.032, 0.5003) -- (0.033, 0.5003) -- (0.034, 0.5003) -- (0.035, 0.5003) -- (0.036, 0.5003) -- (0.037, 0.5003) -- (0.038, 0.5004) -- (0.039, 0.5004) -- (0.040, 0.5004) -- (0.041, 0.5004) -- (0.042, 0.5004) -- (0.043, 0.5005) -- (0.044, 0.5005) -- (0.045, 0.5005) -- (0.046, 0.5005) -- (0.047, 0.5006) -- (0.048, 0.5006) -- (0.049, 0.5006) -- (0.050, 0.5006) -- (0.051, 0.5007) -- (0.052, 0.5007) -- (0.053, 0.5007) -- (0.054, 0.5007) -- (0.055, 0.5008) -- (0.056, 0.5008) -- (0.057, 0.5008) -- (0.058, 0.5008) -- (0.059, 0.5009) -- (0.060, 0.5009) -- (0.061, 0.5009) -- (0.062, 0.5010) -- (0.063, 0.5010) -- (0.064, 0.5010) -- (0.065, 0.5011) -- (0.066, 0.5011) -- (0.067, 0.5011) -- (0.068, 0.5012) -- (0.069, 0.5012) -- (0.070, 0.5012) -- (0.071, 0.5013) -- (0.072, 0.5013) -- (0.073, 0.5013) -- (0.074, 0.5014) -- (0.075, 0.5014) -- (0.076, 0.5014) -- (0.077, 0.5015) -- (0.078, 0.5015) -- (0.079, 0.5016) -- (0.080, 0.5016) -- (0.081, 0.5016) -- (0.082, 0.5017) -- (0.083, 0.5017) -- (0.084, 0.5018) -- (0.085, 0.5018) -- (0.086, 0.5019) -- (0.087, 0.5019) -- (0.088, 0.5019) -- (0.089, 0.5020) -- (0.090, 0.5020) -- (0.091, 0.5021) -- (0.092, 0.5021) -- (0.093, 0.5022) -- (0.094, 0.5022) -- (0.095, 0.5023) -- (0.096, 0.5023) -- (0.097, 0.5024) -- (0.098, 0.5024) -- (0.099, 0.5025) -- (0.100, 0.5025) -- (0.101, 0.5026) -- (0.102, 0.5026) -- (0.103, 0.5027) -- (0.104, 0.5027) -- (0.105, 0.5028) -- (0.106, 0.5028) -- (0.107, 0.5029) -- (0.108, 0.5029) -- (0.109, 0.5030) -- (0.110, 0.5030) -- (0.111, 0.5031) -- (0.112, 0.5032) -- (0.113, 0.5032) -- (0.114, 0.5033) -- (0.115, 0.5033) -- (0.116, 0.5034) -- (0.117, 0.5034) -- (0.118, 0.5035) -- (0.119, 0.5036) -- (0.120, 0.5036) -- (0.121, 0.5037) -- (0.122, 0.5037) -- (0.123, 0.5038) -- (0.124, 0.5039) -- (0.125, 0.5039) -- (0.126, 0.5040) -- (0.127, 0.5041) -- (0.128, 0.5041) -- (0.129, 0.5042) -- (0.130, 0.5043) -- (0.131, 0.5043) -- (0.132, 0.5044) -- (0.133, 0.5045) -- (0.134, 0.5045) -- (0.135, 0.5046) -- (0.136, 0.5047) -- (0.137, 0.5047) -- (0.138, 0.5048) -- (0.139, 0.5049) -- (0.140, 0.5049) -- (0.141, 0.5050) -- (0.142, 0.5051) -- (0.143, 0.5052) -- (0.144, 0.5052) -- (0.145, 0.5053) -- (0.146, 0.5054) -- (0.147, 0.5055) -- (0.148, 0.5055) -- (0.149, 0.5056) -- (0.150, 0.5057) -- (0.151, 0.5058) -- (0.152, 0.5058) -- (0.153, 0.5059) -- (0.154, 0.5060) -- (0.155, 0.5061) -- (0.156, 0.5062) -- (0.157, 0.5062) -- (0.158, 0.5063) -- (0.159, 0.5064) -- (0.160, 0.5065) -- (0.161, 0.5066) -- (0.162, 0.5066) -- (0.163, 0.5067) -- (0.164, 0.5068) -- (0.165, 0.5069) -- (0.166, 0.5070) -- (0.167, 0.5071) -- (0.168, 0.5071) -- (0.169, 0.5072) -- (0.170, 0.5073) -- (0.171, 0.5074) -- (0.172, 0.5075) -- (0.173, 0.5076) -- (0.174, 0.5077) -- (0.175, 0.5078) -- (0.176, 0.5079) -- (0.177, 0.5079) -- (0.178, 0.5080) -- (0.179, 0.5081) -- (0.180, 0.5082) -- (0.181, 0.5083) -- (0.182, 0.5084) -- (0.183, 0.5085) -- (0.184, 0.5086) -- (0.185, 0.5087) -- (0.186, 0.5088) -- (0.187, 0.5089) -- (0.188, 0.5090) -- (0.189, 0.5091) -- (0.190, 0.5092) -- (0.191, 0.5093) -- (0.192, 0.5094) -- (0.193, 0.5095) -- (0.194, 0.5096) -- (0.195, 0.5097) -- (0.196, 0.5098) -- (0.197, 0.5099) -- (0.198, 0.5100) -- (0.199, 0.5101) -- (0.200, 0.5102) -- (0.201, 0.5103) -- (0.202, 0.5104) -- (0.203, 0.5105) -- (0.204, 0.5106) -- (0.205, 0.5107) -- (0.206, 0.5108) -- (0.207, 0.5109) -- (0.208, 0.5110) -- (0.209, 0.5111) -- (0.210, 0.5112) -- (0.211, 0.5114) -- (0.212, 0.5115) -- (0.213, 0.5116) -- (0.214, 0.5117) -- (0.215, 0.5118) -- (0.216, 0.5119) -- (0.217, 0.5120) -- (0.218, 0.5121) -- (0.219, 0.5123) -- (0.220, 0.5124) -- (0.221, 0.5125) -- (0.222, 0.5126) -- (0.223, 0.5127) -- (0.224, 0.5128) -- (0.225, 0.5129) -- (0.226, 0.5131) -- (0.227, 0.5132) -- (0.228, 0.5133) -- (0.229, 0.5134) -- (0.230, 0.5135) -- (0.231, 0.5137) -- (0.232, 0.5138) -- (0.233, 0.5139) -- (0.234, 0.5140) -- (0.235, 0.5142) -- (0.236, 0.5143) -- (0.237, 0.5144) -- (0.238, 0.5145) -- (0.239, 0.5147) -- (0.240, 0.5148) -- (0.241, 0.5149) -- (0.242, 0.5150) -- (0.243, 0.5152) -- (0.244, 0.5153) -- (0.245, 0.5154) -- (0.246, 0.5155) -- (0.247, 0.5157) -- (0.248, 0.5158) -- (0.249, 0.5159) -- (0.250, 0.5161) -- (0.251, 0.5162) -- (0.252, 0.5163) -- (0.253, 0.5165) -- (0.254, 0.5166) -- (0.255, 0.5167) -- (0.256, 0.5169) -- (0.257, 0.5170) -- (0.258, 0.5171) -- (0.259, 0.5173) -- (0.260, 0.5174) -- (0.261, 0.5176) -- (0.262, 0.5177) -- (0.263, 0.5178) -- (0.264, 0.5180) -- (0.265, 0.5181) -- (0.266, 0.5183) -- (0.267, 0.5184) -- (0.268, 0.5185) -- (0.269, 0.5187) -- (0.270, 0.5188) -- (0.271, 0.5190) -- (0.272, 0.5191) -- (0.273, 0.5193) -- (0.274, 0.5194) -- (0.275, 0.5196) -- (0.276, 0.5197) -- (0.277, 0.5199) -- (0.278, 0.5200) -- (0.279, 0.5202) -- (0.280, 0.5203) -- (0.281, 0.5205) -- (0.282, 0.5206) -- (0.283, 0.5208) -- (0.284, 0.5209) -- (0.285, 0.5211) -- (0.286, 0.5212) -- (0.287, 0.5214) -- (0.288, 0.5215) -- (0.289, 0.5217) -- (0.290, 0.5218) -- (0.291, 0.5220) -- (0.292, 0.5222) -- (0.293, 0.5223) -- (0.294, 0.5225) -- (0.295, 0.5226) -- (0.296, 0.5228) -- (0.297, 0.5230) -- (0.298, 0.5231) -- (0.299, 0.5233) -- (0.300, 0.5234) -- (0.301, 0.5236) -- (0.302, 0.5238) -- (0.303, 0.5239) -- (0.304, 0.5241) -- (0.305, 0.5243) -- (0.306, 0.5244) -- (0.307, 0.5246) -- (0.308, 0.5248) -- (0.309, 0.5249) -- (0.310, 0.5251) -- (0.311, 0.5253) -- (0.312, 0.5254) -- (0.313, 0.5256) -- (0.314, 0.5258) -- (0.315, 0.5260) -- (0.316, 0.5261) -- (0.317, 0.5263) -- (0.318, 0.5265) -- (0.319, 0.5267) -- (0.320, 0.5268) -- (0.321, 0.5270) -- (0.322, 0.5272) -- (0.323, 0.5274) -- (0.324, 0.5275) -- (0.325, 0.5277) -- (0.326, 0.5279) -- (0.327, 0.5281) -- (0.328, 0.5283) -- (0.329, 0.5284) -- (0.330, 0.5286) -- (0.331, 0.5288) -- (0.332, 0.5290) -- (0.333, 0.5292) -- (0.334, 0.5294) -- (0.335, 0.5295) -- (0.336, 0.5297) -- (0.337, 0.5299) -- (0.338, 0.5301) -- (0.339, 0.5303) -- (0.340, 0.5305) -- (0.341, 0.5307) -- (0.342, 0.5309) -- (0.343, 0.5310) -- (0.344, 0.5312) -- (0.345, 0.5314) -- (0.346, 0.5316) -- (0.347, 0.5318) -- (0.348, 0.5320) -- (0.349, 0.5322) -- (0.350, 0.5324) -- (0.351, 0.5326) -- (0.352, 0.5328) -- (0.353, 0.5330) -- (0.354, 0.5332) -- (0.355, 0.5334) -- (0.356, 0.5336) -- (0.357, 0.5338) -- (0.358, 0.5340) -- (0.359, 0.5342) -- (0.360, 0.5344) -- (0.361, 0.5346) -- (0.362, 0.5348) -- (0.363, 0.5350) -- (0.364, 0.5352) -- (0.365, 0.5354) -- (0.366, 0.5356) -- (0.367, 0.5358) -- (0.368, 0.5360) -- (0.369, 0.5362) -- (0.370, 0.5365) -- (0.371, 0.5367) -- (0.372, 0.5369) -- (0.373, 0.5371) -- (0.374, 0.5373) -- (0.375, 0.5375) -- (0.376, 0.5377) -- (0.377, 0.5379) -- (0.378, 0.5382) -- (0.379, 0.5384) -- (0.380, 0.5386) -- (0.381, 0.5388) -- (0.382, 0.5390) -- (0.383, 0.5392) -- (0.384, 0.5395) -- (0.385, 0.5397) -- (0.386, 0.5399) -- (0.387, 0.5401) -- (0.388, 0.5404) -- (0.389, 0.5406) -- (0.390, 0.5408) -- (0.391, 0.5410) -- (0.392, 0.5413) -- (0.393, 0.5415) -- (0.394, 0.5417) -- (0.395, 0.5419) -- (0.396, 0.5422) -- (0.397, 0.5424) -- (0.398, 0.5426) -- (0.399, 0.5429) -- (0.400, 0.5431) -- (0.401, 0.5433) -- (0.402, 0.5436) -- (0.403, 0.5438) -- (0.404, 0.5440) -- (0.405, 0.5443) -- (0.406, 0.5445) -- (0.407, 0.5447) -- (0.408, 0.5450) -- (0.409, 0.5452) -- (0.410, 0.5454) -- (0.411, 0.5457) -- (0.412, 0.5459) -- (0.413, 0.5462) -- (0.414, 0.5464) -- (0.415, 0.5467) -- (0.416, 0.5469) -- (0.417, 0.5471) -- (0.418, 0.5474) -- (0.419, 0.5476) -- (0.420, 0.5479) -- (0.421, 0.5481) -- (0.422, 0.5484) -- (0.423, 0.5486) -- (0.424, 0.5489) -- (0.425, 0.5491) -- (0.426, 0.5494) -- (0.427, 0.5496) -- (0.428, 0.5499) -- (0.429, 0.5501) -- (0.430, 0.5504) -- (0.431, 0.5507) -- (0.432, 0.5509) -- (0.433, 0.5512) -- (0.434, 0.5514) -- (0.435, 0.5517) -- (0.436, 0.5520) -- (0.437, 0.5522) -- (0.438, 0.5525) -- (0.439, 0.5527) -- (0.440, 0.5530) -- (0.441, 0.5533) -- (0.442, 0.5535) -- (0.443, 0.5538) -- (0.444, 0.5541) -- (0.445, 0.5543) -- (0.446, 0.5546) -- (0.447, 0.5549) -- (0.448, 0.5551) -- (0.449, 0.5554) -- (0.450, 0.5557) -- (0.451, 0.5560) -- (0.452, 0.5562) -- (0.453, 0.5565) -- (0.454, 0.5568) -- (0.455, 0.5571) -- (0.456, 0.5573) -- (0.457, 0.5576) -- (0.458, 0.5579) -- (0.459, 0.5582) -- (0.460, 0.5585) -- (0.461, 0.5587) -- (0.462, 0.5590) -- (0.463, 0.5593) -- (0.464, 0.5596) -- (0.465, 0.5599) -- (0.466, 0.5602) -- (0.467, 0.5604) -- (0.468, 0.5607) -- (0.469, 0.5610) -- (0.470, 0.5613) -- (0.471, 0.5616) -- (0.472, 0.5619) -- (0.473, 0.5622) -- (0.474, 0.5625) -- (0.475, 0.5628) -- (0.476, 0.5631) -- (0.477, 0.5634) -- (0.478, 0.5637) -- (0.479, 0.5640) -- (0.480, 0.5643) -- (0.481, 0.5645) -- (0.482, 0.5649) -- (0.483, 0.5652) -- (0.484, 0.5655) -- (0.485, 0.5658) -- (0.486, 0.5661) -- (0.487, 0.5664) -- (0.488, 0.5667) -- (0.489, 0.5670) -- (0.490, 0.5673) -- (0.491, 0.5676) -- (0.492, 0.5679) -- (0.493, 0.5682) -- (0.494, 0.5685) -- (0.495, 0.5688) -- (0.496, 0.5692) -- (0.497, 0.5695) -- (0.498, 0.5698) -- (0.499, 0.5701) -- (0.500, 0.5704) -- (0.501, 0.5707) -- (0.502, 0.5711) -- (0.503, 0.5714) -- (0.504, 0.5717) -- (0.505, 0.5720) -- (0.506, 0.5723) -- (0.507, 0.5727) -- (0.508, 0.5730) -- (0.509, 0.5733) -- (0.510, 0.5736) -- (0.511, 0.5740) -- (0.512, 0.5743) -- (0.513, 0.5746) -- (0.514, 0.5750) -- (0.515, 0.5753) -- (0.516, 0.5756) -- (0.517, 0.5760) -- (0.518, 0.5763) -- (0.519, 0.5766) -- (0.520, 0.5770) -- (0.521, 0.5773) -- (0.522, 0.5776) -- (0.523, 0.5780) -- (0.524, 0.5783) -- (0.525, 0.5787) -- (0.526, 0.5790) -- (0.527, 0.5794) -- (0.528, 0.5797) -- (0.529, 0.5801) -- (0.530, 0.5804) -- (0.531, 0.5807) -- (0.532, 0.5811) -- (0.533, 0.5814) -- (0.534, 0.5818) -- (0.535, 0.5822) -- (0.536, 0.5825) -- (0.537, 0.5829) -- (0.538, 0.5832) -- (0.539, 0.5836) -- (0.540, 0.5839) -- (0.541, 0.5843) -- (0.542, 0.5847) -- (0.543, 0.5850) -- (0.544, 0.5854) -- (0.545, 0.5857) -- (0.546, 0.5861) -- (0.547, 0.5865) -- (0.548, 0.5868) -- (0.549, 0.5872) -- (0.550, 0.5876) -- (0.551, 0.5879) -- (0.552, 0.5883) -- (0.553, 0.5887) -- (0.554, 0.5891) -- (0.555, 0.5894) -- (0.556, 0.5898) -- (0.557, 0.5902) -- (0.558, 0.5906) -- (0.559, 0.5909) -- (0.560, 0.5913) -- (0.561, 0.5917) -- (0.562, 0.5921) -- (0.563, 0.5925) -- (0.564, 0.5929) -- (0.565, 0.5932) -- (0.566, 0.5936) -- (0.567, 0.5940) -- (0.568, 0.5944) -- (0.569, 0.5948) -- (0.570, 0.5952) -- (0.571, 0.5956) -- (0.572, 0.5960) -- (0.573, 0.5964) -- (0.574, 0.5968) -- (0.575, 0.5972) -- (0.576, 0.5976) -- (0.577, 0.5980) -- (0.578, 0.5984) -- (0.579, 0.5988) -- (0.580, 0.5992) -- (0.581, 0.5996) -- (0.582, 0.6000) -- (0.583, 0.6004) -- (0.584, 0.6008) -- (0.585, 0.6012) -- (0.586, 0.6016) -- (0.587, 0.6020) -- (0.588, 0.6025) -- (0.589, 0.6029) -- (0.590, 0.6033) -- (0.591, 0.6037) -- (0.592, 0.6041) -- (0.593, 0.6045) -- (0.594, 0.6050) -- (0.595, 0.6054) -- (0.596, 0.6058) -- (0.597, 0.6062) -- (0.598, 0.6067) -- (0.599, 0.6071) -- (0.600, 0.6075) -- (0.601, 0.6080) -- (0.602, 0.6084) -- (0.603, 0.6088) -- (0.604, 0.6093) -- (0.605, 0.6097) -- (0.606, 0.6101) -- (0.607, 0.6106) -- (0.608, 0.6110) -- (0.609, 0.6114) -- (0.610, 0.6119) -- (0.611, 0.6123) -- (0.612, 0.6128) -- (0.613, 0.6132) -- (0.614, 0.6137) -- (0.615, 0.6141) -- (0.616, 0.6146) -- (0.617, 0.6150) -- (0.618, 0.6155) -- (0.619, 0.6159) -- (0.620, 0.6164) -- (0.621, 0.6168) -- (0.622, 0.6173) -- (0.623, 0.6178) -- (0.624, 0.6182) -- (0.625, 0.6187) -- (0.626, 0.6192) -- (0.627, 0.6196) -- (0.628, 0.6201) -- (0.629, 0.6206) -- (0.630, 0.6210) -- (0.631, 0.6215) -- (0.632, 0.6220) -- (0.633, 0.6224) -- (0.634, 0.6229) -- (0.635, 0.6234) -- (0.636, 0.6239) -- (0.637, 0.6244) -- (0.638, 0.6248) -- (0.639, 0.6253) -- (0.640, 0.6258) -- (0.641, 0.6263) -- (0.642, 0.6268) -- (0.643, 0.6273) -- (0.644, 0.6278) -- (0.645, 0.6283) -- (0.646, 0.6288) -- (0.647, 0.6292) -- (0.648, 0.6297) -- (0.649, 0.6302) -- (0.650, 0.6307) -- (0.651, 0.6312) -- (0.652, 0.6317) -- (0.653, 0.6323) -- (0.654, 0.6328) -- (0.655, 0.6333) -- (0.656, 0.6338) -- (0.657, 0.6343) -- (0.658, 0.6348) -- (0.659, 0.6353) -- (0.660, 0.6358) -- (0.661, 0.6363) -- (0.662, 0.6369) -- (0.663, 0.6374) -- (0.664, 0.6379) -- (0.665, 0.6384) -- (0.666, 0.6390) -- (0.667, 0.6395) -- (0.668, 0.6400) -- (0.669, 0.6405) -- (0.670, 0.6411) -- (0.671, 0.6416) -- (0.672, 0.6422) -- (0.673, 0.6427) -- (0.674, 0.6432) -- (0.675, 0.6438) -- (0.676, 0.6443) -- (0.677, 0.6449) -- (0.678, 0.6454) -- (0.679, 0.6459) -- (0.680, 0.6465) -- (0.681, 0.6470) -- (0.682, 0.6476) -- (0.683, 0.6482) -- (0.684, 0.6487) -- (0.685, 0.6493) -- (0.686, 0.6498) -- (0.687, 0.6504) -- (0.688, 0.6510) -- (0.689, 0.6515) -- (0.690, 0.6521) -- (0.691, 0.6527) -- (0.692, 0.6532) -- (0.693, 0.6538) -- (0.694, 0.6544) -- (0.695, 0.6549) -- (0.696, 0.6555) -- (0.697, 0.6561) -- (0.698, 0.6567) -- (0.699, 0.6573) -- (0.700, 0.6579) -- (0.701, 0.6584) -- (0.702, 0.6590) -- (0.703, 0.6596) -- (0.704, 0.6602) -- (0.705, 0.6608) -- (0.706, 0.6614) -- (0.707, 0.6620) -- (0.708, 0.6626) -- (0.709, 0.6632) -- (0.710, 0.6638) -- (0.711, 0.6644) -- (0.712, 0.6650) -- (0.713, 0.6656) -- (0.714, 0.6663) -- (0.715, 0.6669) -- (0.716, 0.6675) -- (0.717, 0.6681) -- (0.718, 0.6687) -- (0.719, 0.6694) -- (0.720, 0.6700) -- (0.721, 0.6706) -- (0.722, 0.6712) -- (0.723, 0.6719) -- (0.724, 0.6725) -- (0.725, 0.6731) -- (0.726, 0.6738) -- (0.727, 0.6744) -- (0.728, 0.6750) -- (0.729, 0.6757) -- (0.730, 0.6763) -- (0.731, 0.6770) -- (0.732, 0.6776) -- (0.733, 0.6783) -- (0.734, 0.6789) -- (0.735, 0.6796) -- (0.736, 0.6803) -- (0.737, 0.6809) -- (0.738, 0.6816) -- (0.739, 0.6822) -- (0.740, 0.6829) -- (0.741, 0.6836) -- (0.742, 0.6843) -- (0.743, 0.6849) -- (0.744, 0.6856) -- (0.745, 0.6863) -- (0.746, 0.6870) -- (0.747, 0.6877) -- (0.748, 0.6883) -- (0.749, 0.6890) -- (0.750, 0.6897) -- (0.751, 0.6904) -- (0.752, 0.6911) -- (0.753, 0.6918) -- (0.754, 0.6925) -- (0.755, 0.6932) -- (0.756, 0.6939) -- (0.757, 0.6946) -- (0.758, 0.6953) -- (0.759, 0.6960) -- (0.760, 0.6968) -- (0.761, 0.6975) -- (0.762, 0.6982) -- (0.763, 0.6989) -- (0.764, 0.6996) -- (0.765, 0.7004) -- (0.766, 0.7011) -- (0.767, 0.7018) -- (0.768, 0.7026) -- (0.769, 0.7033) -- (0.770, 0.7040) -- (0.771, 0.7048) -- (0.772, 0.7055) -- (0.773, 0.7063) -- (0.774, 0.7070) -- (0.775, 0.7078) -- (0.776, 0.7085) -- (0.777, 0.7093) -- (0.778, 0.7101) -- (0.779, 0.7108) -- (0.780, 0.7116) -- (0.781, 0.7124) -- (0.782, 0.7131) -- (0.783, 0.7139) -- (0.784, 0.7147) -- (0.785, 0.7155) -- (0.786, 0.7163) -- (0.787, 0.7170) -- (0.788, 0.7178) -- (0.789, 0.7186) -- (0.790, 0.7194) -- (0.791, 0.7202) -- (0.792, 0.7210) -- (0.793, 0.7218) -- (0.794, 0.7226) -- (0.795, 0.7234) -- (0.796, 0.7242) -- (0.797, 0.7251) -- (0.798, 0.7259) -- (0.799, 0.7267) -- (0.800, 0.7275) -- (0.801, 0.7284) -- (0.802, 0.7292) -- (0.803, 0.7300) -- (0.804, 0.7309) -- (0.805, 0.7317) -- (0.806, 0.7325) -- (0.807, 0.7334) -- (0.808, 0.7342) -- (0.809, 0.7351) -- (0.810, 0.7359) -- (0.811, 0.7368) -- (0.812, 0.7377) -- (0.813, 0.7385) -- (0.814, 0.7394) -- (0.815, 0.7403) -- (0.816, 0.7411) -- (0.817, 0.7420) -- (0.818, 0.7429) -- (0.819, 0.7438) -- (0.820, 0.7447) -- (0.821, 0.7456) -- (0.822, 0.7465) -- (0.823, 0.7474) -- (0.824, 0.7483) -- (0.825, 0.7492) -- (0.826, 0.7501) -- (0.827, 0.7510) -- (0.828, 0.7519) -- (0.829, 0.7528) -- (0.830, 0.7538) -- (0.831, 0.7547) -- (0.832, 0.7556) -- (0.833, 0.7565) -- (0.834, 0.7575) -- (0.835, 0.7584) -- (0.836, 0.7594) -- (0.837, 0.7603) -- (0.838, 0.7613) -- (0.839, 0.7622) -- (0.840, 0.7632) -- (0.841, 0.7642) -- (0.842, 0.7651) -- (0.843, 0.7661) -- (0.844, 0.7671) -- (0.845, 0.7681) -- (0.846, 0.7690) -- (0.847, 0.7700) -- (0.848, 0.7710) -- (0.849, 0.7720) -- (0.850, 0.7730) -- (0.851, 0.7740) -- (0.852, 0.7750) -- (0.853, 0.7760) -- (0.854, 0.7771) -- (0.855, 0.7781) -- (0.856, 0.7791) -- (0.857, 0.7801) -- (0.858, 0.7812) -- (0.859, 0.7822) -- (0.860, 0.7833) -- (0.861, 0.7843) -- (0.862, 0.7854) -- (0.863, 0.7864) -- (0.864, 0.7875) -- (0.865, 0.7885) -- (0.866, 0.7896) -- (0.867, 0.7907) -- (0.868, 0.7918) -- (0.869, 0.7929) -- (0.870, 0.7939) -- (0.871, 0.7950) -- (0.872, 0.7961) -- (0.873, 0.7972) -- (0.874, 0.7983) -- (0.875, 0.7995) -- (0.876, 0.8006) -- (0.877, 0.8017) -- (0.878, 0.8028) -- (0.879, 0.8040) -- (0.880, 0.8051) -- (0.881, 0.8062) -- (0.882, 0.8074) -- (0.883, 0.8085) -- (0.884, 0.8097) -- (0.885, 0.8109) -- (0.886, 0.8120) -- (0.887, 0.8132) -- (0.888, 0.8144) -- (0.889, 0.8156) -- (0.890, 0.8168) -- (0.891, 0.8180) -- (0.892, 0.8192) -- (0.893, 0.8204) -- (0.894, 0.8216) -- (0.895, 0.8228) -- (0.896, 0.8240) -- (0.897, 0.8253) -- (0.898, 0.8265) -- (0.899, 0.8277) -- (0.900, 0.8290) -- (0.901, 0.8303) -- (0.902, 0.8315) -- (0.903, 0.8328) -- (0.904, 0.8341) -- (0.905, 0.8353) -- (0.906, 0.8366) -- (0.907, 0.8379) -- (0.908, 0.8392) -- (0.909, 0.8405) -- (0.910, 0.8418) -- (0.911, 0.8431) -- (0.912, 0.8445) -- (0.913, 0.8458) -- (0.914, 0.8471) -- (0.915, 0.8485) -- (0.916, 0.8498) -- (0.917, 0.8512) -- (0.918, 0.8526) -- (0.919, 0.8539) -- (0.920, 0.8553) -- (0.921, 0.8567) -- (0.922, 0.8581) -- (0.923, 0.8595) -- (0.924, 0.8609) -- (0.925, 0.8623) -- (0.926, 0.8637) -- (0.927, 0.8652) -- (0.928, 0.8666) -- (0.929, 0.8681) -- (0.930, 0.8695) -- (0.931, 0.8710) -- (0.932, 0.8725) -- (0.933, 0.8739) -- (0.934, 0.8754) -- (0.935, 0.8769) -- (0.936, 0.8784) -- (0.937, 0.8799) -- (0.938, 0.8815) -- (0.939, 0.8830) -- (0.940, 0.8845) -- (0.941, 0.8861) -- (0.942, 0.8876) -- (0.943, 0.8892) -- (0.944, 0.8908) -- (0.945, 0.8924) -- (0.946, 0.8940) -- (0.947, 0.8956) -- (0.948, 0.8972) -- (0.949, 0.8988) -- (0.950, 0.9005) -- (0.951, 0.9021) -- (0.952, 0.9038) -- (0.953, 0.9054) -- (0.954, 0.9071) -- (0.955, 0.9088) -- (0.956, 0.9105) -- (0.957, 0.9122) -- (0.958, 0.9139) -- (0.959, 0.9157) -- (0.960, 0.9174) -- (0.961, 0.9192) -- (0.962, 0.9209) -- (0.963, 0.9227) -- (0.964, 0.9245) -- (0.965, 0.9263) -- (0.966, 0.9281) -- (0.967, 0.9299) -- (0.968, 0.9318) -- (0.969, 0.9337) -- (0.970, 0.9355) -- (0.971, 0.9374) -- (0.972, 0.9393) -- (0.973, 0.9412) -- (0.974, 0.9432) -- (0.975, 0.9451) -- (0.976, 0.9471) -- (0.977, 0.9490) -- (0.978, 0.9510) -- (0.979, 0.9530) -- (0.980, 0.9551) -- (0.981, 0.9571) -- (0.982, 0.9592) -- (0.983, 0.9612) -- (0.984, 0.9633) -- (0.985, 0.9654) -- (0.986, 0.9676) -- (0.987, 0.9697) -- (0.988, 0.9719) -- (0.989, 0.9741) -- (0.990, 0.9763) -- (0.991, 0.9786) -- (0.992, 0.9808) -- (0.993, 0.9831) -- (0.994, 0.9854) -- (0.995, 0.9878) -- (0.996, 0.9902) -- (0.997, 0.9926) -- (0.998, 0.9950) -- (0.999, 0.9975) -- (1.000, 1.0000);
\draw[color = blue, fill = white, thick] (0, 0.5) circle(0.625pt);
\draw[color = blue, fill = white, thick] (1, 1) circle(0.625pt);
\end{tikzpicture}
\label{fig:spm_ap2}
}
\caption{Demonstration for numeric calculations about constant $\C \approx 2.6202$.}
\label{fig:spm_ap}
\end{figure}

\subsection{Proof of \texorpdfstring{\Cref{lem:spm_ap_reduction3}}{}}
\label{subapp:spm_ap_reductions}

\begin{lemma}
\label{lem:spm_ap_reduction3}
W.l.o.g.\ any worst-case instance of Program~\eqref{prog:spm_ap0.5} satisfies that $v_1 > v_2 > \cdots > v_n$.
\end{lemma}

\begin{proof}
W.l.o.g.\ suppose that $v_k = v_{k + 1}$ for some $k \in [n - 1]$, then consider the following new triangular instance $\{\tri(\vv_i, \qq_i)\}_{i = 1}^{n - 1}$:
\begin{align}
\notag
& \vv_i \eqdef v_i, && \qq_i \eqdef q_i, && \forall i \in [k - 1]; \\
\label{eq:distinct_v}
& \vv_k \eqdef v_k = v_{k + 1}, && \qq_k \eqdef q_k + q_{k + 1} - q_k \cdot q_{k + 1}; \\
\notag
& \vv_i \eqdef v_{i + 1}, && \qq_i \eqdef q_{i + 1}, && \forall i \in [k + 1: n - 1].
\end{align}
To establish \Cref{lem:spm_ap_reduction3}, it suffices to prove the following:
\begin{itemize}
\item $\spm\big(\{\tri(\vv_i, \qq_i)\}_{i = 1}^{n - 1}\big) = \spm\big(\{\tri(v_i, q_i)\}_{i = 1}^n\big)$;
\item $\ap\big(p, \{\tri(\vv_i, \qq_i)\}_{i = 1}^{n - 1}\big) \leq \ap\big(p, \{\tri(v_i, q_i)\}_{i = 1}^n\big)$ for any $p \in \RRP$.
\end{itemize}
Recall the $\spm$ revenue formula in \Cref{lem:spm_tri_best}. The first claim trivially holds by construction, i.e.\ $\vv_k \qq_k \overset{\eqref{eq:distinct_v}}{=} v_k q_k + v_{k + 1} q_{k + 1} \cdot (1 - q_k)$ and $1 - \qq_k \overset{\eqref{eq:distinct_v}}{=} (1 - q_k) \cdot (1 - q_{k + 1})$. For the second claim, recall the $\ap(p)$ revenue formula in \Cref{subsec:prelim:mech}. It suffices to justify that $\FF_k(p) \geq F_k(p) \cdot F_{k + 1}(p)$. This follows as $\FF_k(p) = F_k(p) = F_{k + 1}(p) = 1$ when $p \in [v_k, \infty)$, and
\begin{align*}
\big(\FF_k(p)\big)^{-1} - \big(F_k(p) \cdot F_{k + 1}(p)\big)^{-1}
& = 1 + \frac{\vv_k \qq_k}{1 - \qq_k} \cdot \frac{1}{p} - \Big(1 + \frac{v_k q_k}{1 - q_k} \cdot \frac{1}{p}\Big) \cdot \Big(1 + \frac{v_{k + 1} q_{k + 1}}{1 - q_{k + 1}} \cdot \frac{1}{p}\Big) \\
& \overset{\eqref{eq:distinct_v}}{=} -\frac{q_k}{1 - q_k} \cdot \frac{q_{k + 1}}{1 - q_{k + 1}} \cdot \Big(\frac{v_k}{p} - 1\Big) \cdot \frac{v_k}{p} \leq 0,
\end{align*}
when $p \in (0, v_k)$. This completes the proof of \Cref{lem:spm_ap_reduction3}.
\end{proof}

\subsection{Proof of \texorpdfstring{\Cref{lem:spm_ap_reduction4}}{}}
\label{subapp:spm_ap_reduction4}

\begin{lemma}
\label{lem:spm_ap_reduction4}
W.l.o.g.\ any worst-case instance of Program~\eqref{prog:spm_ap0.5} includes distribution $\tri(\infty)$, i.e.\ CDF $F_0(p) = \frac{p}{p + 1}$ for all $p \in \RRP$.
\end{lemma}

\begin{proof}
Given a triangular instance $\{\tri(v_i, q_i)\}_{i = 1}^n$ feasible to Program~\eqref{prog:spm_ap0.5}, assume w.l.o.g.\ that $\sum_{i = 1}^n v_i q_i \geq \sum_{i = 1}^n v_i q_i \cdot \prod_{j = 1}^{i - 1} (1 - q_j) = \spm > 1$, then
\begin{enumerate*}[label = (\alph*), font = {\bfseries}]
\item $k \eqdef \argmin_{1 \leq i \leq n} \big\{\sum_{j = 1}^{i} v_j q_j > 1\big\}$ is well defined. Besides,
\item $v_k q_k \leq 1$, which follows from constraint~\eqref{cstr:spm_ap0.5} as
\end{enumerate*}
\[
1 \geq \ap(v_k) = v_k \cdot \Big(1 - \prod_{j = 1}^n F_i(v_k)\Big) \geq v_k \cdot \big(1 - F_k(v_k)\big) = v_k q_k.
\]
Now, consider the next triangular instance $\tri(\vv_i, \qq_i)_{i = 1}^{n - k + 1}$:
\begin{align}
\label{eq:lem_v0q0=1-2}
& \vv_1 \eqdef v_k, && \qq_1 \eqdef \frac{1}{v_k} \cdot \Big(\sum_{i = 1}^k v_i q_i - 1\Big) \overset{\bf (a)}{\leq} q_k; \\
\label{eq:lem_v0q0=1-3}
& \vv_i \eqdef v_{i + k - 1}, && \qq_i \eqdef q_{i + k - 1}, && \forall i \in [2: n - k + 1].
\end{align}
We can infer \Cref{lem:spm_ap_reduction4} from the following:
\begin{itemize}
\item $\spm\big(\{\tri(\infty) \cup \tri(\vv_i, \qq_i)\}_{i = 1}^{n - k + 1}\big) \leq \spm\big(\{\tri(v_i, q_i)\}_{i = 1}^n\big)$.
\item $\ap\big(p, \tri(\infty) \cup \{\tri(\vv_i, \qq_i)\}_{i = 1}^{n - k + 1}\big) \leq 1$ for any $p \in \RRP$.
\end{itemize}
Recall the $\spm$ revenue formula in \Cref{lem:spm_tri_best}. The firts claim trivially holds by construction, i.e.\ $1 + \vv_k \qq_k \overset{\eqref{eq:lem_v0q0=1-2}}{=} \sum_{i = 1}^k v_i q_i \geq \sum_{i = 1}^k v_i q_i \cdot \prod_{j = 1}^{i - 1} (1 - q_j)$ and $1 - \qq_1 \overset{\eqref{eq:lem_v0q0=1-2}}{\geq} 1 - q_k \geq \prod_{i = 1}^k (1 - q_k)$.

To see the second claim, recall the $\ap(p)$ revenue formula in \Cref{subsec:prelim:mech}. As formalized in Program~\eqref{prog:spm_ap0.5}, $v_1 \geq v_2 \geq \cdots \geq v_n$. By construction, $\vv_1 \geq \vv_2 \geq \cdots \geq \vv_{n - k + 1}$, and thus
\[
\ap\big(p, \tri(\infty) \cup \{\tri(\vv_i, \qq_i)\}_{i = 1}^{n - k + 1}\big) \leq \ap\big(p, \tri(\infty)\big) = p \cdot \Big(1 - \frac{p}{p + 1}\Big) = \frac{p}{p + 1} \leq 1,
\]
for any $p \in [\vv_1, \infty)$. In the other range of $p \in (0, \vv_1) = (0, v_k)$, we only need to reveal that $\ap\big(p, \tri(\infty) \cup \{\tri(\vv_i, \qq_i)\}_{i = 1}^{n - k + 1}\big) \leq \ap\big(p, \{\tri(v_i, q_i)\}_{i = 1}^n\big)$, which is implied by the following:
\begin{align*}
\big(\FF_0(p) \cdot \FF_1(p)\big)^{-1}
& \,\,=\,\, \Big(1 + \frac{1}{p}\Big) \cdot \Big(1 + \frac{\sum_{i = 1}^k v_i q_i - 1}{1 - \qq_k} \cdot \frac{1}{p}\Big) \\
& \,\,\overset{\eqref{eq:lem_v0q0=1-2}}{\leq}\,\, \Big(1 + \frac{1}{p}\Big) \cdot \Big(1 + \frac{\sum_{i = 1}^k v_i q_i - 1}{1 - q_k} \cdot \frac{1}{p}\Big) \quad\quad\quad\quad \mbox{\tt (as $\qq_k \leq q_k$)} \\
& \overset{{\bf (a,b)}}{\leq} \Big(1 + \frac{1}{p}\Big) \cdot \Big(1 + \frac{\sum_{i = 1}^k v_i q_i - 1}{1 - q_k} \cdot \frac{1}{p}\Big) + \frac{(1 - v_k q_k) \cdot (1 - \sum_{i = 1}^{k - 1} v_i q_i)}{p^2 \cdot (1 - q_k)} \\
& \,\,=\,\, \Big(1 + \frac{\sum_{i = 1}^{k - 1} v_i q_i}{p}\Big) \cdot \Big(1 + \frac{v_k q_k}{1 - q_k} \cdot \frac{1}{p}\Big) \\
& \,\,\leq\,\, \prod_{i = 1}^k \Big(1 + \frac{v_i q_i}{1 - q_i} \cdot \frac{1}{p}\Big)
= \prod_{i = 1}^k \big(F_i(p)\big)^{-1},
\end{align*}
where the last inequality follows as $1 + \sum_i z_i \leq \prod_i (1 + z_i)$ when $z_i$'s are all non-negative. This completes the proof of \Cref{lem:spm_ap_reduction4}.
\end{proof}

\subsection{Proof of \texorpdfstring{\Cref{lem:spm_ap_relexation}}{}}
\label{subapp:spm_ap_solution}

{\bf \Cref{lem:spm_ap_relexation}.}
{\em Given any triangular instance $\{\tri(v_i, q_i)\}_{i = 1}^n$ with constraint~\eqref{cstr:spm_ap4} loose for some $i \in [n]$, there exists another triangular instance $\{\tri(\vv_i, \qq_i)\}_{i = 1}^n$ satisfying the following:
\begin{enumerate}[font = {\em\bfseries}]
\item Constraint~\eqref{cstr:spm_ap4} still holds, and is tight for each $i \in [n]$.
\item $\spm\big(\{\tri(\vv_i, \qq_i)\}_{i = 1}^n\big) \geq \spm\big(\{\tri(v_i, q_i)\}_{i = 1}^n\big)$.
\end{enumerate}}

\begin{proof}
Let $v_0 \eqdef \infty$; recall Part 2 of \Cref{lem:spm_ap_RQ} that $\R(v_0) = 0$. W.l.o.g.\ let $k \in [n]$ be the smallest index for constraint~\eqref{cstr:spm_ap4} to be loose, then
\begin{align}
\notag
& \sum_{j = 1}^i \ln\Big(1 + \frac{v_j q_j}{1 - q_j}\Big) = \R(v_i) && \Rightarrow && \ln\Big(1 + \frac{v_j q_j}{1 - q_j}\Big) = \R(v_i) - \R(v_{i - 1}), && \forall i \in [k - 1]; \\
\label{eq:v_k}
& \sum_{j = 1}^k \ln\Big(1 + \frac{v_j q_j}{1 - q_j}\Big) < \R(v_k) && \Rightarrow && \ln\Big(1 + \frac{v_k q_k}{1 - q_k}\Big) < \R(v_k) - \R(v_{k - 1}).
\end{align}
Denote $\Delta_k \eqdef \ln\big(1 + \frac{v_k q_k}{1 - q_k}\big)$ for brevity, and consider this triangular instance $\{\tri(\vv_i, \qq_i)\}_{i = 1}^n$:
\begin{align}
\notag
& \vv_i \eqdef v_i, && \qq_i \eqdef q_i, && \forall i \in ([n] \setminus \{k\}); \\
\label{eq:bar_vk}
& \vv_k \eqdef \R^{-1}\big(\Delta_k + \R(v_{k - 1})\big), && \qq_k \eqdef \frac{e^{\R(\vv_k) - \R(v_{k - 1})} - 1}{\vv_k + e^{\R(\vv_k) - \R(v_{k - 1})} - 1}.
\end{align}
In what follows, we prove that
\begin{enumerate*}[label = (\alph*), font = {\bfseries}]
\item this instance makes constraint~\eqref{cstr:spm_ap4} tight for each $i \in [k]$; and
\item $\spm\big(\{\tri(\vv_i, \qq_i)\}_{i = 1}^n\big) \geq \spm\big(\{\tri(v_i, q_i)\}_{i = 1}^n\big)$.
\end{enumerate*}
To see so, these facts are useful:
\begin{itemize}
\item $v_k < \vv_k < v_{k-1}$, since $\R(\vv_k) \overset{\eqref{eq:bar_vk}}{=} \Delta_k + \R(v_{k - 1}) \overset{\eqref{eq:v_k}}{\in} \big(\R(v_{k-1}), \R(v_k)\big)$ and function $\R$ is a decreasing function (see Part 1 of \Cref{lem:spm_ap_RQ});
\item $\ln\big(1 + \frac{\vv_i \qq_i}{1 - \qq_i}\big) = \ln\big(1 + \frac{v_i q_i}{1 - q_i}\big)$ for each $i \in [n]$, by construction of instance $\{\tri(\vv_i, \qq_i)\}_{i = 1}^n$;
\item $\qq_k < q_k$ and $\vv_k \qq_k > v_k q_k$, as a consequence of the first and second facts.
\end{itemize}
Clearly, {\bf Claim~(a)} follows from the first and second facts and \Cref{eq:bar_vk}. Moreover, recall Program~\eqref{prog:spm_ap1} for the $\spm$ revenue formula. Here is a sufficient condition for {\bf Claim~(b)}:
\begin{align*}
& (\vv_i - 1) \cdot \qq_i \cdot \prod_{j = 1}^{i - 1} (1 - \qq_j) \geq (v_i - 1) \cdot q_i \cdot \prod_{j = 1}^{i - 1} (1 - q_j), && \forall i \in [n].
\end{align*}
From the third fact and \Cref{eq:bar_vk}, we know $(\vv_i - 1) \cdot \qq_i \geq (v_i - 1) \cdot q_i$ and $1 - \qq_j \geq 1 - q_j$ for all $i \in [n]$. Given these, the above inequality and thus {\bf Claim (b)} follows immediately.

Respecting {\bf Claims (a, b)}, we can easily construct a desired triangular instance by induction. This completes the proof of \Cref{lem:spm_ap_relexation}.
\end{proof}

\subsection{Proof of \texorpdfstring{\Cref{lem:spm_ap_inequality}}{}}
\label{subapp:spm:ineq}

We have proved \Cref{lem:ineq1,lem:ineq2} respectively in \Cref{subapp:ineq1,subapp:ineq2}. Both facts are useful for proving \Cref{lem:spm_ap_inequality}.

\vspace{.1in}
\noindent
{\bf \Cref{lem:ineq1}.}
{\em $G(x, y) \leq 0$ for any $y \geq x > 1$, where
\[
G(x, y) \eqdef (1 - x^{-1}) \cdot (e^{\R(x) - \R(y)} - 1) + \big(\R(y) - \R(x)\big) - \big(\Q(y) - \Q(x)\big).
\]}

\noindent
{\bf \Cref{lem:ineq2}.}
{\em $H(x, y) \geq 0$ for any $y \geq x > 1$, where
\[
	H(x, y) \eqdef x^{-1} \cdot (e^{\R(x) - \R(y)} - 1) - (e^{\Q(x) - \Q(y)} - 1).
\]}

\noindent
{\bf \Cref{lem:spm_ap_inequality}.}
{\em Given any triangular instance $\{\tri(v_i, q_i)\}_{i = 1}^n$ such that $v_1 > v_2 > \cdots > v_n \geq 1$ and the parameters $\{q_i\}_{i = 1}^n$ satisfying constraint~\eqref{cstr:spm_ap5}. The following holds for each $k \in [n]$:
\[
\mbox{$(v_k - 1) \cdot q_k \cdot \prod_{j = 1}^{k - 1} (1 - q_j) \leq \Int{v_k}{v_{k - 1}} (x - 1) \cdot \big(-\Q'(x)\big) \cdot e^{-\Q(x)} \cdot \dd x$}.
\]}

\vspace{-.1in}
\begin{proof}
It suffices to justify the following two inequalities:
\begin{align}
\label{ineq:lem_spm_ap1}
& (v_k - 1) \cdot q_k \cdot (1 - q_k)^{-1} \leq \mbox{$\Int{v_k}{v_{k - 1}} (x - 1) \cdot \big(-\Q'(x)\big) \cdot \dd x$}; \\
\label{ineq:lem_spm_ap2}
& \prod_{j = 1}^k (1 - q_j) \leq e^{-\Q(x)},
\hspace{2cm} \forall x \in [v_k, v_{k - 1}].
\end{align}
To see inequality~\eqref{ineq:lem_spm_ap1}, recall \Cref{lem:ineq1}, constraint~\eqref{cstr:spm_ap5} that $q_k = \frac{e^{\R(v_k) - \R(v_{k - 1})} - 1}{v_k + e^{\R(v_k) - \R(v_{k - 1})} - 1}$, and Part 1 of \Cref{lem:spm_ap_RQ} that $\R'(p) = p \cdot \Q'(p)$ for any $p \in (1, \infty)$. It follows that
\[
\begin{aligned}
\lhs \mbox{ of } \eqref{ineq:lem_spm_ap1}
& \overset{\eqref{cstr:spm_ap5}}{=} (1 - v_k^{-1}) \cdot (e^{\R(v_k) - \R(v_{k - 1})} - 1) \\
& \,\,\,\,\leq\,\,\,\, \big(\R(v_k) - \R(v_{k - 1})\big) - \big(\Q(v_k) - \Q(v_{k - 1})\big)
&& \mbox{\tt (by \Cref{lem:ineq1})} \\
& \,\,\,\,=\,\,\,\, \rhs \mbox{ of } \eqref{ineq:lem_spm_ap1}.
&& \mbox{\tt (by Part 1 of \Cref{lem:spm_ap_RQ})}
\end{aligned}
\]
Similarly, to deal with inequality~\eqref{ineq:lem_spm_ap1}, recall \Cref{lem:spm_ap_RQ} that $\Q$ is a decreasing function on interval $p \in (1, \infty)$, and that $\Q(v_0) = \Q(\infty) = 0$. Combining these arguments with \Cref{lem:ineq2} results in
\[
\begin{aligned}
\lhs \mbox{ of } \eqref{ineq:lem_spm_ap2}
& \overset{\eqref{cstr:spm_ap5}}{=} \prod_{j = 1}^k \big[1 + v_j^{-1} \cdot (e^{\R(v_j) - \R(v_{j - 1})} - 1)\big]^{-1} \\
& \,\,\,\,\leq\,\,\,\, \prod_{j = 1}^k \big(e^{\Q(v_j) - \Q(v_{j - 1})}\big)^{-1}
&& \mbox{\tt (by \Cref{lem:ineq2})} \\
& \,\,\,\,=\,\,\,\, e^{-\Q(v_k)}
\leq \rhs \mbox{ of } \eqref{ineq:lem_spm_ap2}.
&& \mbox{\tt (by \Cref{lem:spm_ap_RQ})}
\end{aligned}
\]
This completes the proof of \Cref{lem:spm_ap_inequality}.
\end{proof}

\subsection{Proof of \texorpdfstring{\Cref{lem:spm_ap_lower2}}{}}
\label{subapp:spm_ap_lower}

{\bf Lemma~\ref{lem:spm_ap_lower2}.}
{\em Given any constant $\eps \in (0, 1)$, let $a = \min \big\{(1 + \eps), \Q^{-1}(\ln\eps^{-1})\big\}$ and $b = (1 + \eps^{-1})$. Then, $2 + \Int{a}{b} (x - 1) \cdot \big(-\Q'(x)\big) \cdot e^{-\Q(x)} \cdot \dd x \geq \C - 4 \cdot \eps$.}

\begin{proof}
By the definition of $\C \approx 2.6202$, it suffices to justify the next two inequalities:
\begin{align}
\label{eq:spm_lower1}
& \mbox{$\Int{1}{a} (x - 1) \cdot \big(-\Q'(x)\big) \cdot e^{-\Q(x)} \cdot \dd x \leq 2 \cdot \eps$}; \\
\label{eq:spm_lower2}
& \mbox{$\Int{b}{\infty} (x - 1) \cdot \big(-\Q'(x)\big) \cdot e^{-\Q(x)} \cdot \dd x \leq 2 \cdot \eps$}.
\end{align}
Observe that
\begin{enumerate*}[label = (\alph*), font = {\bfseries}]
\item function $\Q$ is a decreasing function on interval $p \in (1, \infty)$;
\item $\Q(1^+) = \infty$; and
\item $\Q(p) = -\ln(1 - p^{-2}) - \frac{1}{2} \cdot \sum_{k = 1}^{\infty} k^{-2} \cdot p^{-2k} \leq -\ln(1 - p^{-2})$.
\end{enumerate*}
Given these, we have
\begin{align*}
\lhs \mbox{ of } \eqref{eq:spm_lower1}
& \leq \mbox{$\Int{1}{a} x \cdot \big(-\Q'(x)\big) \cdot e^{-\Q(x)} \cdot \dd x$} \\
& = \mbox{$\Int{1}{a} x \cdot \dd e^{-\Q(x)}$}
&& \mbox{\tt (integration by parts)} \\
& \leq a \cdot e^{-\Q(a)}
&& \mbox{\tt (by {\bf Claims~(a,b)} above)} \\
& \leq (1 + \eps) \cdot \eps
\leq \rhs \mbox{ of } \eqref{eq:spm_lower1};
&& \mbox{\tt (as $a = \min \big\{(1 + \eps), \Q^{-1}(\ln\eps^{-1})\big\}$)} \\
%=====================================%
\lhs \mbox{ of } \eqref{eq:spm_lower2}
& = \mbox{$\Int{b}{\infty} x \cdot \big(-\Q'(x)\big) \cdot e^{-\Q(x)} \cdot \dd x$} \\
& = \mbox{$b \cdot \big(1 - e^{-\Q(b)}\big) + \Int{b}{\infty} \big(1 - e^{-\Q(x)}\big) \cdot \dd x$}
&& \mbox{\tt (integration by parts)} \\
& \leq \mbox{$b \cdot b^{-2} + \Int{b}{\infty} x^{-2} \cdot \dd x$}
&& \mbox{\tt (by {\bf Claim~(c)} above)} \\
& = 2 \cdot \eps / (1 + \eps)
\leq \rhs \mbox{ of } \eqref{eq:spm_lower2}.
&& \mbox{\tt (as $b = 1 + \eps^{-1}$)}
\end{align*}
This completes the proof of \Cref{lem:spm_ap_lower2}.
\end{proof}

\section{Proof of Lemma~\ref{lem:F_n_irregular}}
\label{subapp:ar_ap_lower_formula}
\begin{lemma}
\label{lem:F_n_irregular}
For any $n \geq 2$, distribution $F_n(p) \eqdef
\begin{cases}
0 & \forall p \in [0, 1] \\
\big(1 - \frac{1}{p}\big)^{\frac{1}{n}} & \forall p \in (1, \infty)
\end{cases}$ is irregular.
\end{lemma}

\begin{proof}
As mentioned, a distribution is regular iff its revenue-quantile curve is a concave function. Hence, it suffices to prove that the revenue-quantile curve $r_n(q) \eqdef \frac{q}{1 - (1 - q)^n} = \frac{1}{\sum_{k = 0}^{n - 1} (1 - q)^k}$ of distribution $F_n$ is a convex function on interval $q \in (0, 1)$. In this range:
\begin{align*}
r_n'(q)
& = \frac{\sum_{k = 0}^{n - 1} k \cdot (1 - q)^{k - 1}}{\big[\sum_{k = 0}^{n - 1} (1 - q)^k\big]^2}; \\
r_n''(q)
& = \frac{2 \cdot \big[\sum_{k = 0}^{n - 1} k \cdot (1 - q)^{k - 1}\big]^2}{\big[\sum_{k = 0}^{n - 1} (1 - q)^k\big]^3} - \frac{\big[\sum_{k = 0}^{n - 1} k \cdot (k - 1) \cdot (1 - q)^{k - 1}\big]}{\big[\sum_{k = 0}^{n - 1} (1 - q)^k\big]^2} \\
& \geq \frac{\big[\sum_{k = 0}^{n - 1} k \cdot (1 - q)^{k - 1}\big]^2}{\big[\sum_{k = 0}^{n - 1} (1 - q)^k\big]^3} - \frac{\big[\sum_{k = 0}^{n - 1} k^2 \cdot (1 - q)^{k - 2}\big]}{\big[\sum_{k = 0}^{n - 1} (1 - q)^k\big]^2} \\
& \geq 0. \quad\quad\quad\quad \mbox{\tt (by the Cauchy-Schwarz inequality)}
\end{align*}
This completes the proof of \Cref{lem:F_n_irregular}.
\end{proof}

\section{Proof of Lemma~\ref{lem:ar_rev_monotone}}
\label{app:opt_ar}
{\bf \Cref{lem:ar_rev_monotone}.}
{\em Given any triangular instance $\{\tri(v_i, q_i)\}_{i = 1}^n$ that $v_1 \geq v_2 \cdots \geq v_n > v_{n + 1} \eqdef 0$, the best $\ar$ revenue is
achieved by reserve price $p = v_i$, for some $i \in [n]$.}

\begin{proof}
Recall that a triangular distribution $\tri(v_i, q_i)$ has a CDF of $F_i(p) = \frac{(1 - q_i) \cdot p}{(1 - q_i) \cdot p + v_i q_i}$ when $p \in [0, v_i)$, and $F_i(p) = 1$ when $p \in [v_i, \infty)$. Hence, the {\em virtual value function} $\varphi_i$ of triangular distribution $F_i$ maps any value $p \in (0, v_i)$ to a {\em negative constant} virtual value of
\begin{equation}
\label{eq:lem:ar_rev_monotone1}
\varphi_i(p) \eqdef p - \frac{1 - F_i(p)}{f_i(p)} = -\frac{v_i q_i}{1 - q_i}.
\end{equation}
To see the lemma, it suffices to show that $\ar(p)$ is a non-decreasing function on each interval $p \in (v_{i + 1}, v_i)$. Assume w.l.o.g.\ that this interval is non-empty, then
\begin{align*}
& \ar(p) = \mbox{$p \cdot \big(1 - \prod_{j = 1}^i F_j(p)\big) + \Int{p}{\infty} \Big\{1 - \prod_{j: v_j \geq x} F_j(x) \cdot \big[1 + \sum_{j: v_j \geq x} \big(\frac{1}{F_j(x)} - 1\big)\big]\Big\} \cdot \dd x$}; \\
\Rightarrow \quad& \ar'(p) = \prod_{j = 1}^i F_j(p) \cdot \sum_{j = 1}^i \frac{f_j(p)}{F_j(p)} \cdot \big(-\varphi_j(p)\big) \overset{\eqref{eq:lem:ar_rev_monotone1}}{\geq} 0.
\end{align*}
This completes the proof of \Cref{lem:ar_rev_monotone}.
\end{proof}

\bibliographystyle{apalike}
\bibliography{main}

\begin{thebibliography}{}

\bibitem[Abolhassani et~al., 2017]{AEEHK17}
Abolhassani, M., Ehsani, S., Esfandiari, H., Hajiaghayi, M., Kleinberg, R.~D.,
  and Lucier, B. (2017).
\newblock Beating 1-1/e for ordered prophets.
\newblock In {\em Proceedings of the 49th Annual {ACM} {SIGACT} Symposium on
  Theory of Computing, {STOC} 2017, Montreal, QC, Canada, June 19-23, 2017},
  pages 61--71.

\bibitem[Alaei, 2014]{A14}
Alaei, S. (2014).
\newblock Bayesian combinatorial auctions: Expanding single buyer mechanisms to
  many buyers.
\newblock {\em {SIAM} J. Comput.}, 43(2):930--972.

\bibitem[Alaei et~al., 2019]{AHNPY15}
Alaei, S., Hartline, J.~D., Niazadeh, R., Pountourakis, E., and Yuan, Y.
  (2019).
\newblock Optimal auctions vs. anonymous pricing.
\newblock {\em Games and Economic Behavior}, 118:494--510.

\bibitem[Anari et~al., 2019]{ec/AnariNSS19}
Anari, N., Niazadeh, R., Saberi, A., and Shameli, A. (2019).
\newblock Nearly optimal pricing algorithms for production constrained and
  laminar bayesian selection.
\newblock In {\em Proceedings of the 2019 {ACM} Conference on Economics and
  Computation, {EC} 2019, Phoenix, AZ, USA, June 24-28, 2019}, pages 91--92.

\bibitem[Anshelevich and Sekar, 2017]{AS17}
Anshelevich, E. and Sekar, S. (2017).
\newblock Price doubling and item halving: Robust revenue guarantees for item
  pricing.
\newblock In {\em Proceedings of the 2017 {ACM} Conference on Economics and
  Computation, {EC} '17, Cambridge, MA, USA, June 26-30, 2017}, pages 325--342.

\bibitem[Azar et~al., 2018]{ACK18}
Azar, Y., Chiplunkar, A., and Kaplan, H. (2018).
\newblock Prophet secretary: Surpassing the 1-1/e barrier.
\newblock In {\em Proceedings of the 2018 {ACM} Conference on Economics and
  Computation, Ithaca, NY, USA, June 18-22, 2018}, pages 303--318.

\bibitem[Babaioff et~al., 2007]{BIK07}
Babaioff, M., Immorlica, N., and Kleinberg, R. (2007).
\newblock Matroids, secretary problems, and online mechanisms.
\newblock In {\em Proceedings of the Eighteenth Annual {ACM-SIAM} Symposium on
  Discrete Algorithms, {SODA} 2007, New Orleans, Louisiana, USA, January 7-9,
  2007}, pages 434--443.

\bibitem[Babaioff et~al., 2014]{BILW14}
Babaioff, M., Immorlica, N., Lucier, B., and Weinberg, S.~M. (2014).
\newblock A simple and approximately optimal mechanism for an additive buyer.
\newblock In {\em 55th {IEEE} Annual Symposium on Foundations of Computer
  Science, {FOCS} 2014, Philadelphia, PA, USA, October 18-21, 2014}, pages
  21--30.

\bibitem[Bar{-}Yossef et~al., 2002]{BHW02}
Bar{-}Yossef, Z., Hildrum, K., and Wu, F. (2002).
\newblock Incentive-compatible online auctions for digital goods.
\newblock In {\em Proceedings of the Thirteenth Annual {ACM-SIAM} Symposium on
  Discrete Algorithms, January 6-8, 2002, San Francisco, CA, {USA.}}, pages
  964--970.

\bibitem[Beyhaghi et~al., 2018]{BGPPS18}
Beyhaghi, H., Golrezaei, N., Leme, R.~P., Pal, M., and Sivan, B. (2018).
\newblock Improved approximations for free-order prophets and second-price
  auctions.
\newblock {\em CoRR}, abs/1807.03435.

\bibitem[Birmpas et~al., 2017]{BMTT17}
Birmpas, G., Markakis, E., Telelis, O., and Tsikiridis, A. (2017).
\newblock Tight welfare guarantees for pure nash equilibria of the uniform
  price auction.
\newblock In {\em Algorithmic Game Theory - 10th International Symposium,
  {SAGT} 2017, L'Aquila, Italy, September 12-14, 2017, Proceedings}, pages
  16--28.

\bibitem[Bulow and Klemperer, 1994]{BK94}
Bulow, J. and Klemperer, P. (1994).
\newblock Auctions vs. negotiations.
\newblock Technical report, National Bureau of Economic Research.

\bibitem[Cai and Daskalakis, 2015]{CD15}
Cai, Y. and Daskalakis, C. (2015).
\newblock Extreme value theorems for optimal multidimensional pricing.
\newblock {\em Games and Economic Behavior}, 92:266--305.

\bibitem[Cai et~al., 2016]{CDW16}
Cai, Y., Devanur, N.~R., and Weinberg, S.~M. (2016).
\newblock A duality based unified approach to bayesian mechanism design.
\newblock In {\em Proceedings of the 48th Annual {ACM} {SIGACT} Symposium on
  Theory of Computing, {STOC} 2016, Cambridge, MA, USA, June 18-21, 2016},
  pages 926--939.

\bibitem[Cai and Zhao, 2017]{CZ17}
Cai, Y. and Zhao, M. (2017).
\newblock Simple mechanisms for subadditive buyers via duality.
\newblock In {\em Proceedings of the 49th Annual {ACM} {SIGACT} Symposium on
  Theory of Computing, {STOC} 2017, Montreal, QC, Canada, June 19-23, 2017},
  pages 170--183.

\bibitem[Cesa{-}Bianchi et~al., 2015]{CGM15}
Cesa{-}Bianchi, N., Gentile, C., and Mansour, Y. (2015).
\newblock Regret minimization for reserve prices in second-price auctions.
\newblock {\em {IEEE} Trans. Information Theory}, 61(1):549--564.

\bibitem[Chawla et~al., 2007]{CHK07}
Chawla, S., Hartline, J.~D., and Kleinberg, R.~D. (2007).
\newblock Algorithmic pricing via virtual valuations.
\newblock In {\em Proceedings 8th {ACM} Conference on Electronic Commerce
  (EC-2007), San Diego, California, USA, June 11-15, 2007}, pages 243--251.

\bibitem[Chawla et~al., 2010]{CHMS10}
Chawla, S., Hartline, J.~D., Malec, D.~L., and Sivan, B. (2010).
\newblock Multi-parameter mechanism design and sequential posted pricing.
\newblock In {\em Proceedings of the 42nd {ACM} Symposium on Theory of
  Computing, {STOC} 2010, Cambridge, Massachusetts, USA, 5-8 June 2010}, pages
  311--320.

\bibitem[Chawla et~al., 2015]{CMS15}
Chawla, S., Malec, D.~L., and Sivan, B. (2015).
\newblock The power of randomness in bayesian optimal mechanism design.
\newblock {\em Games and Economic Behavior}, 91:297--317.

\bibitem[Chawla and Miller, 2016]{CM16}
Chawla, S. and Miller, J.~B. (2016).
\newblock Mechanism design for subadditive agents via an ex ante relaxation.
\newblock In {\em Proceedings of the 2016 {ACM} Conference on Economics and
  Computation, {EC} '16, Maastricht, The Netherlands, July 24-28, 2016}, pages
  579--596.

\bibitem[Chen et~al., 2014]{CGL14}
Chen, N., Gravin, N., and Lu, P. (2014).
\newblock Optimal competitive auctions.
\newblock In {\em Symposium on Theory of Computing, {STOC} 2014, New York, NY,
  USA, May 31 - June 03, 2014}, pages 253--262.

\bibitem[Chen et~al., 2015]{CGL15}
Chen, N., Gravin, N., and Lu, P. (2015).
\newblock Competitive analysis via benchmark decomposition.
\newblock In {\em Proceedings of the Sixteenth {ACM} Conference on Economics
  and Computation, {EC} '15, Portland, OR, USA, June 15-19, 2015}, pages
  363--376.

\bibitem[Chen et~al., 2011]{CHLW11}
Chen, X., Hu, G., Lu, P., and Wang, L. (2011).
\newblock On the approximation ratio of k-lookahead auction.
\newblock In {\em WINE}, pages 61--71. Springer.

\bibitem[Correa et~al., 2019a]{CSZ19}
Correa, J., Saona, R., and Ziliotto, B. (2019a).
\newblock Prophet secretary through blind strategies.
\newblock In {\em Proceedings of the Thirtieth Annual {ACM-SIAM} Symposium on
  Discrete Algorithms, {SODA} 2019, San Diego, California, USA, January 6-9,
  2019}, pages 1946--1961.

\bibitem[Correa et~al., 2019b]{ec/CorreaDFS19}
Correa, J.~R., D{\"{u}}tting, P., Fischer, F.~A., and Schewior, K. (2019b).
\newblock Prophet inequalities for {I.I.D.} random variables from an unknown
  distribution.
\newblock In {\em Proceedings of the 2019 {ACM} Conference on Economics and
  Computation, {EC} 2019, Phoenix, AZ, USA, June 24-28, 2019}, pages 3--17.

\bibitem[Correa et~al., 2017]{CFHOV17}
Correa, J.~R., Foncea, P., Hoeksma, R., Oosterwijk, T., and Vredeveld, T.
  (2017).
\newblock Posted price mechanisms for a random stream of customers.
\newblock In {\em Proceedings of the 2017 {ACM} Conference on Economics and
  Computation, {EC} '17, Cambridge, MA, USA, June 26-30, 2017}, pages 169--186.

\bibitem[Cremer and McLean, 1988]{cremer1988full}
Cremer, J. and McLean, R.~P. (1988).
\newblock Full extraction of the surplus in bayesian and dominant strategy
  auctions.
\newblock {\em Econometrica: Journal of the Econometric Society}, pages
  1247--1257.

\bibitem[Dobzinski et~al., 2015]{DFK15}
Dobzinski, S., Fu, H., and Kleinberg, R. (2015).
\newblock Approximately optimal auctions for correlated bidders.
\newblock {\em Games and Economic Behavior}, 92:349--369.

\bibitem[Duetting et~al., 2017]{focs/DuettingFKL17}
Duetting, P., Feldman, M., Kesselheim, T., and Lucier, B. (2017).
\newblock Prophet inequalities made easy: Stochastic optimization by pricing
  non-stochastic inputs.
\newblock In {\em 58th {IEEE} Annual Symposium on Foundations of Computer
  Science, {FOCS} 2017, Berkeley, CA, USA, October 15-17, 2017}, pages
  540--551.

\bibitem[D{\"{u}}tting et~al., 2016]{DFK16}
D{\"{u}}tting, P., Fischer, F.~A., and Klimm, M. (2016).
\newblock Revenue gaps for static and dynamic posted pricing of homogeneous
  goods.
\newblock {\em CoRR}, abs/1607.07105.

\bibitem[D{\"{u}}tting and Kesselheim, 2019]{ec/DuttingK19}
D{\"{u}}tting, P. and Kesselheim, T. (2019).
\newblock Posted pricing and prophet inequalities with inaccurate priors.
\newblock In {\em Proceedings of the 2019 {ACM} Conference on Economics and
  Computation, {EC} 2019, Phoenix, AZ, USA, June 24-28, 2019}, pages 111--129.

\bibitem[Eden et~al., 2017]{EFFTW17:b}
Eden, A., Feldman, M., Friedler, O., Talgam{-}Cohen, I., and Weinberg, S.~M.
  (2017).
\newblock The competition complexity of auctions: {A} bulow-klemperer result
  for multi-dimensional bidders.
\newblock In {\em Proceedings of the 2017 {ACM} Conference on Economics and
  Computation, {EC} '17, Cambridge, MA, USA, June 26-30, 2017}, page 343.

\bibitem[Ehsani et~al., 2018]{EHKS18}
Ehsani, S., Hajiaghayi, M., Kesselheim, T., and Singla, S. (2018).
\newblock Prophet secretary for combinatorial auctions and matroids.
\newblock In {\em Proceedings of the Twenty-Ninth Annual {ACM-SIAM} Symposium
  on Discrete Algorithms, {SODA} 2018, New Orleans, LA, USA, January 7-10,
  2018}, pages 700--714.

\bibitem[Esfandiari et~al., 2017]{EHLM17}
Esfandiari, H., Hajiaghayi, M., Liaghat, V., and Monemizadeh, M. (2017).
\newblock Prophet secretary.
\newblock {\em {SIAM} J. Discrete Math.}, 31(3):1685--1701.

\bibitem[Feldman et~al., 2015]{FGL15}
Feldman, M., Gravin, N., and Lucier, B. (2015).
\newblock Combinatorial auctions via posted prices.
\newblock In {\em Proceedings of the Twenty-Sixth Annual {ACM-SIAM} Symposium
  on Discrete Algorithms, {SODA} 2015, San Diego, CA, USA, January 4-6, 2015},
  pages 123--135.

\bibitem[Fu et~al., 2015]{FILS15}
Fu, H., Immorlica, N., Lucier, B., and Strack, P. (2015).
\newblock Randomization beats second price as a prior-independent auction.
\newblock In {\em Proceedings of the Sixteenth {ACM} Conference on Economics
  and Computation, {EC} '15, Portland, OR, USA, June 15-19, 2015}, page 323.

\bibitem[Goldberg et~al., 2001]{GHW01}
Goldberg, A.~V., Hartline, J.~D., and Wright, A. (2001).
\newblock Competitive auctions and digital goods.
\newblock In {\em Proceedings of the Twelfth Annual Symposium on Discrete
  Algorithms, January 7-9, 2001, Washington, DC, {USA.}}, pages 735--744.

\bibitem[Guruswami et~al., 2005]{GHKKKM05}
Guruswami, V., Hartline, J.~D., Karlin, A.~R., Kempe, D., Kenyon, C., and
  McSherry, F. (2005).
\newblock On profit-maximizing envy-free pricing.
\newblock In {\em Proceedings of the Sixteenth Annual {ACM-SIAM} Symposium on
  Discrete Algorithms, {SODA} 2005, Vancouver, British Columbia, Canada,
  January 23-25, 2005}, pages 1164--1173.

\bibitem[Hajiaghayi et~al., 2007]{HKS07}
Hajiaghayi, M.~T., Kleinberg, R.~D., and Sandholm, T. (2007).
\newblock Automated online mechanism design and prophet inequalities.
\newblock In {\em Proceedings of the Twenty-Second {AAAI} Conference on
  Artificial Intelligence, July 22-26, 2007, Vancouver, British Columbia,
  Canada}, pages 58--65.

\bibitem[Hart and Nisan, 2012]{HN12}
Hart, S. and Nisan, N. (2012).
\newblock Approximate revenue maximization with multiple items.
\newblock In {\em {ACM} Conference on Electronic Commerce, {EC} '12, Valencia,
  Spain, June 4-8, 2012}, page 656.

\bibitem[Hartline, 2013]{H13}
Hartline, J.~D. (2013).
\newblock Mechanism design and approximation.
\newblock {\em Book draft. October}, 122.

\bibitem[Hartline and Roughgarden, 2009]{HR09}
Hartline, J.~D. and Roughgarden, T. (2009).
\newblock Simple versus optimal mechanisms.
\newblock In {\em Proceedings 10th {ACM} Conference on Electronic Commerce
  (EC-2009), Stanford, California, USA, July 6--10, 2009}, pages 225--234.

\bibitem[Hill et~al., 1982]{HK82}
Hill, T.~P., Kertz, R.~P., et~al. (1982).
\newblock Comparisons of stop rule and supremum expectations of iid random
  variables.
\newblock {\em The Annals of Probability}, 10(2):336--345.

\bibitem[Jin et~al., 2019a]{JLQTX2019}
Jin, Y., Lu, P., Qi, Q., Tang, Z.~G., and Xiao, T. (2019a).
\newblock Tight approximation ratio of anonymous pricing.
\newblock In {\em Proceedings of the 51st Annual {ACM} {SIGACT} Symposium on
  Theory of Computing, {STOC} 2019, Phoenix, AZ, USA, June 23-26, 2019}, pages
  674--685.

\bibitem[Jin et~al., 2019b]{JLTX19}
Jin, Y., Lu, P., Tang, Z.~G., and Xiao, T. (2019b).
\newblock Tight revenue gaps among simple mechanisms.
\newblock In {\em Proceedings of the Thirtieth Annual {ACM-SIAM} Symposium on
  Discrete Algorithms, {SODA} 2019, San Diego, California, USA, January 6-9,
  2019}, pages 209--228.

\bibitem[Kertz, 1986]{K86}
Kertz, R.~P. (1986).
\newblock Stop rule and supremum expectations of iid random variables: a
  complete comparison by conjugate duality.
\newblock {\em Journal of Multivariate Analysis}, 19(1):88--112.

\bibitem[Kleinberg and Weinberg, 2012]{KW12}
Kleinberg, R. and Weinberg, S.~M. (2012).
\newblock Matroid prophet inequalities.
\newblock In {\em Proceedings of the 44th Symposium on Theory of Computing
  Conference, {STOC} 2012, New York, NY, USA, May 19 - 22, 2012}, pages
  123--136.

\bibitem[Koutsoupias and Pierrakos, 2013]{KP13}
Koutsoupias, E. and Pierrakos, G. (2013).
\newblock On the competitive ratio of online sampling auctions.
\newblock {\em {ACM} Trans. Economics and Comput.}, 1(2):10:1--10:10.

\bibitem[Krengel and Sucheston, 1977]{KS77}
Krengel, U. and Sucheston, L. (1977).
\newblock Semiamarts and finite values.
\newblock {\em Bulletin of the American Mathematical Society}, 83(4):745--747.

\bibitem[Krengel and Sucheston, 1978]{KS78}
Krengel, U. and Sucheston, L. (1978).
\newblock On semiamarts, amarts, and processes with finite value.
\newblock {\em Advances in Prob}, 4:197--266.

\bibitem[Li and Yao, 2013]{LY13}
Li, X. and Yao, A. C.-C. (2013).
\newblock On revenue maximization for selling multiple independently
  distributed items.
\newblock {\em Proceedings of the National Academy of Sciences},
  110(28):11232--11237.

\bibitem[Lucier, 2017]{L17}
Lucier, B. (2017).
\newblock An economic view of prophet inequalities.
\newblock {\em SIGecom Exchanges}, 16(1):24--47.

\bibitem[Myerson, 1981]{M81}
Myerson, R.~B. (1981).
\newblock Optimal auction design.
\newblock {\em Math. Oper. Res.}, 6(1):58--73.

\bibitem[Papadimitriou and Pierrakos, 2015]{PP15}
Papadimitriou, C.~H. and Pierrakos, G. (2015).
\newblock Optimal deterministic auctions with correlated priors.
\newblock {\em Games and Economic Behavior}, 92:430--454.

\bibitem[Ronen, 2001]{R01}
Ronen, A. (2001).
\newblock On approximating optimal auctions.
\newblock In {\em Proceedings of the 3rd ACM conference on Electronic
  Commerce}, pages 11--17. ACM.

\bibitem[Rubinstein, 2016]{R16}
Rubinstein, A. (2016).
\newblock Beyond matroids: secretary problem and prophet inequality with
  general constraints.
\newblock In {\em Proceedings of the 48th Annual {ACM} {SIGACT} Symposium on
  Theory of Computing, {STOC} 2016, Cambridge, MA, USA, June 18-21, 2016},
  pages 324--332.

\bibitem[Rubinstein and Singla, 2017]{RS17}
Rubinstein, A. and Singla, S. (2017).
\newblock Combinatorial prophet inequalities.
\newblock In {\em Proceedings of the Twenty-Eighth Annual {ACM-SIAM} Symposium
  on Discrete Algorithms, {SODA} 2017, Barcelona, Spain, Hotel Porta Fira,
  January 16-19}, pages 1671--1687.

\bibitem[Rubinstein and Weinberg, 2015]{RW15}
Rubinstein, A. and Weinberg, S.~M. (2015).
\newblock Simple mechanisms for a subadditive buyer and applications to revenue
  monotonicity.
\newblock In {\em Proceedings of the Sixteenth {ACM} Conference on Economics
  and Computation, {EC} '15, Portland, OR, USA, June 15-19, 2015}, pages
  377--394.

\bibitem[Yan, 2011]{Y11}
Yan, Q. (2011).
\newblock Mechanism design via correlation gap.
\newblock In {\em Proceedings of the Twenty-Second Annual {ACM-SIAM} Symposium
  on Discrete Algorithms, {SODA} 2011, San Francisco, California, USA, January
  23-25, 2011}, pages 710--719.

\bibitem[Yao, 2015]{Y15}
Yao, A.~C. (2015).
\newblock An \emph{n}-to-1 bidder reduction for multi-item auctions and its
  applications.
\newblock In {\em Proceedings of the Twenty-Sixth Annual {ACM-SIAM} Symposium
  on Discrete Algorithms, {SODA} 2015, San Diego, CA, USA, January 4-6, 2015},
  pages 92--109.

\end{thebibliography}

\end{document}